\documentclass[showpacs,amsmath,amssymb,twocolumn,pra,superscriptaddress,notitlepage]{revtex4-2}

\usepackage[dvips]{graphicx}
\usepackage{amsmath,amssymb,amsthm,mathrsfs,amsfonts,dsfont}
\usepackage{subfigure}
\usepackage{braket}
\usepackage{bm}
\usepackage{enumerate}
\usepackage{xcolor}
\usepackage{graphicx}
\usepackage{mathrsfs}
\usepackage{mathtools}
\usepackage{multirow}
\usepackage[normalem]{ulem}

\usepackage[colorlinks=true, linkcolor=black]{hyperref}
\newtheorem{theorem}{Theorem}
\newtheorem{lemma}{Lemma}
\newtheorem{corollary}{Corollary}

\newtheorem{definition}{Definition}

\let\originalleft\left
\let\originalright\right
\renewcommand{\left}{\mathopen{}\mathclose\bgroup\originalleft}
\renewcommand{\right}{\aftergroup\egroup\originalright}

\newcommand{\mc}{\mathcal}
\newcommand{\tr}{\mathrm{Tr}}

\newcommand{\id}{\mathbb{I}}
\newcommand{\s}{\mathbb{S}}
\renewcommand{\d}{\mathrm{d}}
\newcommand{\normH}[1]{{\left\vert\kern-0.25ex\left\vert\kern-0.25ex\left\vert #1
   \right\vert\kern-0.25ex\right\vert\kern-0.25ex\right\vert}}

\begin{document}

\let\oldacl\addcontentsline
\renewcommand{\addcontentsline}[3]{}


\title{Hamiltonian simulation with random inputs}

\author{Qi Zhao}
\affiliation{Joint Center for Quantum Information and Computer Science, University of Maryland, College Park, Maryland 20742, USA}
\author{You Zhou}
\affiliation{Nanyang Quantum Hub, School of Physical and Mathematical Sciences, Nanyang Technological University, Singapore 637371}
\author{Alexander F. Shaw}
\affiliation{Joint Center for Quantum Information and Computer Science, University of Maryland, College Park, Maryland 20742, USA}
\author{Tongyang Li}
\affiliation{Center on Frontiers of Computing Studies, Peking University, Beijing 100871, China}
\affiliation{School of Computer Science, Peking University, Beijing 100871, China}
\affiliation{Center for Theoretical Physics, Massachusetts Institute of Technology, Cambridge, MA 02139, USA}
\author{Andrew M. Childs}
\affiliation{Joint Center for Quantum Information and Computer Science, University of Maryland, College Park, Maryland 20742, USA}
\affiliation{Department of Computer Science and Institute for Advanced Computer Studies, University of Maryland, College Park, Maryland 20742, USA}

\begin{abstract}
The algorithmic error of digital quantum simulations is usually explored in terms of the spectral norm distance between the actual and ideal evolution operators. In practice, this worst-case error analysis may be unnecessarily pessimistic. To address this, we develop a theory of average-case performance of Hamiltonian simulation with random initial states. We relate the average-case error to the Frobenius norm of the multiplicative error
and give upper bounds for the product formula (PF) and truncated Taylor series methods. As applications, we estimate average-case error for digital Hamiltonian simulation of general lattice Hamiltonians and $k$-local Hamiltonians. In particular, for the nearest-neighbor Heisenberg chain with $n$ spins, the error is quadratically reduced from $\mathcal O(n)$ in the worst case to $\mathcal O(\sqrt{n})$ on average for both the PF method and the Taylor series method. Numerical evidence suggests that this theory accurately characterizes the average error for concrete models. We also apply our results to error analysis in the simulation of quantum scrambling.

\end{abstract}

\maketitle

\section{Introduction}
Simulating the time evolution of quantum systems is one of the most promising
applications of quantum computers \cite{Feynman1982}.
Quantum simulation allows quantum computers to efficiently mimic quantum dynamics, a task that is believed to be classically intractable.
Quantum simulation could be applied to study numerous systems, including spin models~\cite{somma2002simulating}, fermionic lattice models~\cite{PhysRevA.92.062318}, quantum chemistry~\cite{PhysRevA.90.022305,babbush2015chemical,babbush2014adiabatic},  and quantum field theories~\cite{jordan2012quantum}.

Following the first concrete digital quantum simulation algorithm proposed by Lloyd \cite{sethuniversal}, many improved algorithms have been developed \cite{berry2007efficient,berry2012black,TaylorSeries,PhysRevLett.118.010501,low2019hamiltonian}.
Algorithms are now known that have optimal or nearly optimal gate complexity with respect to several key parameters \cite{Berry15optimal,childs2019nearly,PhysRevLett.118.010501,low2019hamiltonian,low2019STOC}.
However, since outperforming classical computers requires controlling many qubits with high accuracy,
it is challenging to realize digital quantum simulation on current hardware.

In particular, consider a Hamiltonian of the form $H=\sum^L_{l=1} H_l$, $\|H_l\|\le 1$. Most known algorithms have algorithmic error scaling at least linearly with the number of terms $L$, making simulations of large systems costly or inaccurate \cite{TaylorSeries,PhysRevLett.118.010501,low2019hamiltonian,childs2018toward,childs2020theory}.
Typical error analysis quantifies error in terms of the spectral norm, which characterizes the worst-case input and output in a Hamiltonian simulation problem.
However, such an error bound can be pessimistic in practice, especially given prior knowledge of input states or measurements. In particular, error bounds can be tightened if the initial state is in a low-energy subspace \cite{sahinoglu2020hamiltonian}, within the $\eta$-electron manifold \cite{su2020nearly}, for measurements of local observables \cite{heyl2019quantum,chen2020quantum}, and for simulations of unbounded time-dependent Hamiltonians \cite{an2021time}.

Alternatively, instead of considering the worst-case error, it is natural to quantify performance in terms of the average error for instances drawn at random from some ensemble.
In this work, we study such average-case performance of Hamiltonian simulation algorithms. While mathematically one can consider states drawn from the Haar measure, our analysis only requires the much weaker 1-design property, i.e.,
indistinguishability from the Haar measure given a single copy of the state. Note that many easily prepared sets of states form 1-designs. For instance, a locally random state $\bigotimes_{i=1}^n U_i\ket{0}^{\otimes n}$, with each $U_i$ being a single-qubit Haar-random unitary, forms a 1-design. The uniform distribution over any orthonormal basis, such as the computational basis ensemble
$\mathcal{E}=\left\{(\frac{1}{2^n},\ket{0\ldots00}),(\frac{1}{2^n},\ket{0\ldots01}),\dots, (\frac{1}{2^n},\ket{1\ldots11})\right\}$, also gives a 1-design.
In general, we relate the average performance to the Frobenius norm of the multiplicative error.
Intuitively, whereas worst-case error bounds scale with the largest eigenvalue of the multiplicative error, the Frobenius norm captures the \emph{average} (specifically, root mean square) eigenvalue.

We upper bound the average error for both the $p$th-order product formula (PF$p$) method and the truncated Taylor series method \cite{TaylorSeries}.
For PF$p$, we give a bound in terms of the sum of the Frobenius norms of the $(p+1)$-layer nested commutators.
In particular, we give bounds for PF$1$ and PF2 with detailed prefactors.
Similarly to the worst-case analysis \cite{Tran_2020}, for the average case we also observe destructive error interference in a nearest-neighbor Hamiltonian with the PF1 method, further tightening the error bound.
In particular, we show that for a one-dimensional nearest-neighbor Heisenberg Hamiltonian with $n$ spins, for both the PF and Taylor series methods, average-case analysis gives error $\mathcal O(\sqrt{n})$, quadratically smaller than the worse-case error of $\mathcal O(n)$.

We also explore the simulation of general $k$-local Hamiltonians and power-law Hamiltonians.
Numerical results suggest that our analytical bounds are not far from tight in typical examples.
Moreover, our techniques can be directly used to tighten the Trotter error and reduce the gate complexity in studying out-of-time-order correlators in scrambling physics. We hope that our techniques can inspire further improvements to the analytical performance of quantum simulation algorithms by leveraging particular features of specific simulations.

\section{Worst- and average-case error}
Quantum simulation aims to realize the time-evolution operator $U_0(t):=e^{-iHt}$ of a given Hamiltonian $H$. This task is crucial for studying both dynamic and static properties.
In most previous methods, the distance between the ideal evolution $U_0(t)$ and the approximated evolution $U(t)$ (as implemented by some Hamiltonian simulation algorithm) is quantified by the spectral (i.e., operator) norm $\|\cdot\|$, which is the largest singular value. This measure captures the error for the worst-case input state, since
\begin{align}
W(U,U_0)
:=\|U-U_0\|
&=\max_{\ket{\psi}}\|{U\ket{\psi}-U_0\ket{\psi}}\|_{2},
\end{align}
where $\|\cdot\|_{2}$ denotes the $\ell_2$ (i.e., Euclidean) norm.

However, the worst-case error may be a significant overestimate for some initial states \cite{an2021time}. Instead, we consider the typical error for input states chosen at random from some ensemble $\mathcal{E}=\{(p_{i}, \phi_{i})\}$, defined as
\begin{equation}
\begin{aligned}\label{Eq:definition}
R(U,U_0)&:= \mathbb{E}_{\mathcal{E}}[\mathcal{D}(U\ket{\psi},U_0\ket{\psi})],
\end{aligned}
\end{equation}
where $\mathcal{D}$ is a distance measure (which we take to be either the $\ell_2$ norm or the trace norm).
Specifically, we study the average error and its variance, and relate them to the Frobenius norm $\|X\|_F := \sqrt{\tr (XX^{\dagger})}$ (also known as the Hilbert-Schmidt norm).

For concreteness, consider the $\ell_2$ norm. Using the Cauchy-Schwartz inequality, we have the upper bound $R(U,U_0)\le \smash{\left(\mathbb{E}_{\psi}
\|{U\ket{\psi}-U_0\ket{\psi}}\|_{2}^2\right)^{1/2}}$ for the average $\ell_2$ error. We define the variable inside the square root as
\begin{equation}\label{Eq:defS}
S(\psi):=\|{U\ket{\psi}-U_0\ket{\psi}}\|_{2}^2=2-\braket{\psi|U^{\dagger}U_0+U^{\dagger}_0U|\psi}
\end{equation}
and calculate its expectation and variance as follows.
\begin{theorem}\label{M_Th:l2}
For an input state drawn randomly from a 1-design ensemble $\mathcal{E}$ in an $n$-qubit Hilbert space with dimension $d:=2^n$, the average $\ell_2$ error between the ideal evolution $U_0$ and its approximation $U$ satisfies
\begin{equation}
\begin{aligned}
R(U,U_0)\le [\mathbb{E}_{\psi} S(\psi)]^{\frac1{2}}=
\frac{1}{\sqrt{d}} \|\mathscr{M}\|_F,    \\
 \end{aligned}
\end{equation}
where $\mathscr{M}$ is the multiplicative error, i.e.,  $U=U_0(\id+\mathscr{M})$. When the input ensemble is a 2-design, the variance of $S(\psi)$ has the upper bound $ \mathrm{Var}(S(\psi)) \leq \frac{4\|\mathscr{M}\|_F^2}{d(d+1)}$.
\end{theorem}
Armed with the mean and variance, we know that for an input state $\phi$ drawn from $\mathcal{E}$, the squared error $S(\phi)$ is far from the mean value $\mathbb{E}_{\psi} S(\psi)$ with small probability due to the Chebyshev inequality.

This theorem can also be extended to the case where the input state is chosen at random from a subsystem of an $n$-qubit Hilbert space. Specifically, we have the upper bound
$R^{\Pi}(U,U_0)\le \frac{1}{\sqrt{d_1}} \|\mathscr{M}\Pi\|_F$, where $\Pi$ projects onto a subsystem of dimension $d_1$.
These subsystem results can be used to analyze the Trotter error in out-of-time-order correlator (OTOC) problems, as discussed in Section~\ref{sec:applications}. More generally, in Appendix~\ref{Sec:Average} we give results for the trace norm, approximate 1-design ensembles, and the average error with random inputs and random projections.

\section{Average error in Hamiltonian simulation algorithms}
Two major classes of digital Hamiltonian simulation algorithms include product formula (PF) simulations and methods using linear combinations of unitaries (LCU) to directly implement the Taylor series \cite{TaylorSeries}. Here we focus on the average error in PF$p$ and Taylor series methods.

For a short evolution time $t$, the PF1 algorithm for a Hamiltonian $\sum_{l=1}^LH_l$ applies the unitary operation
\begin{equation}
\mathscr{U}_1(t):=e^{-iH_1t}e^{-iH_2t}\cdots e^{-iH_Lt}
=\overrightarrow{\prod_l}e^{-iH_lt}.
\end{equation}
Here the right arrow indicates the product is in the order of increasing indices. Similarly, we write $\overleftarrow{\prod}_l$ to denote a product in decreasing order.
Suzuki's high-order product formulas are defined recursively by
\begin{align}\label{M_Eq:highorder}
\mathscr{U}_2(t)&:=\overrightarrow{\prod_l}e^{-iH_lt/2} \overleftarrow{\prod_l}e^{-iH_lt/2},    \\
\mathscr{U}_{2k}(t)&:= [\mathscr{U}_{2k-2}(p_k t)]^2 \mathscr{U}_{2k-2}((1-4p_k)t) [\mathscr{U}_{2k-2}(p_k t)]^2, \nonumber
\end{align}
where $p_k:=\frac{1}{4-4^{1/(2k-1)}}$ for $k>1$ \cite{suzuki1991general}.
Overall, with $S=2\cdot 5^{k-1}$ stages, the evolution has the form
\begin{equation}
  \mathscr{U}_{2k}(t)=  \prod_{s=1}^{S} \prod_{l=1}^{L} e^{-it a_s H_{\pi_s(l)}}.
  \label{M_eq:pf2k}
\end{equation}
In each stage $s$, we implement evolution according to the terms in increasing or decreasing index order (specified by $\pi_s$, which is either trivial or the reversal) for time $t a_s$.
For a long time $t$, we divide the evolution into $r$ steps and apply the given product formula $r$ times, approximating the evolution as $\mathscr{U}_{2k}^r(t/r)$.
We find the following bound on the average-case performance of PF$p$, where $p=1$ or $p=2k$.

\begin{theorem}\label{M_Th:general}
For the PF$p$ simulation $\mathscr{U}_{p}^r(t/r)$ specified by \eqref{M_eq:pf2k},
the average error in the $\ell_2$ norm for a 1-design input ensemble has the asymptotic upper bound
$R= \mathcal O\left(T_p{t^{p+1}}/{r^{p}}\right)$,
where
\begin{align}\label{eq:Th2sum}
\hspace{-2mm}
T_p:=\sum_{l_1,\dots,l_{p+1}=1}^L
\frac{1}{\sqrt{d}} \left\|[H_{l_1},[H_{l_2},\dots,[H_{l_p},H_{l_{p+1}}]]] \right\|_F.
\end{align}
Therefore $r=\mathcal O\bigl(T_p^{\frac{1}{p}}t^{1+\frac{1}{p}}\varepsilon^{-\frac{1}{p}}\bigr)$ segments suffice to ensure average error at most $\varepsilon$.
Furthermore, the same bounds hold for the average error with respect to the trace norm.
\end{theorem}

For comparison, the worst-case spectral norm error \cite{childs2020theory} is $W(\mathscr{U}_{p}^r(t/r),U_0(t))
=\mathcal O\left(\alpha_{\mathrm{comm},p}t^{p+1}/r^{p}\right)$,
where
\begin{align}
\hspace{-1mm}\alpha_{\mathrm{comm},p}:=\sum_{l_1,\dots,l_{p+1}=1}^L \left\|[H_{l_1},[H_{l_2},\dots,[H_{l_p},H_{l_{p+1}}]]]\right\|.
\end{align}
Observe that $T_p \le \alpha_{\mathrm{comm},p}$, and $\alpha_{\mathrm{comm},p}$ can be much larger than $T_p$.
For instance, for an $n$-qubit nearest-neighbor Hamiltonian,  $\alpha_{\mathrm{comm},p}=\mathcal O(n)$ and $T_p=\mathcal O(\sqrt{n})$. More examples are given in Section~\ref{sec:applications} below.

Note that the summation in Eq.~\eqref{eq:Th2sum} in Theorem~\ref{M_Th:general} appears outside of the Frobenius norm of nested commutators. In the following, for PF1 and PF2 we tighten the bounds by moving some summations inside the Frobenius norm and giving concrete prefactors.

\begin{theorem}[Triangle bound]\label{Th:PF12}
For the PF1 and PF2 algorithms, the average error in the $\ell_2$ norm has the upper bounds
\begin{equation*}
R(\mathscr{U}_1^r(t/r),U_0(t))\le \frac{t^2}{r}T'_1,~R(\mathscr{U}_2^r(t/r),U_0(t))\le \frac{t^3}{r^2}T'_2,
\end{equation*}
where
\begin{align}
T'_1&:=\frac{1}{2\sqrt{d}}\sum_{l_1=1}^{L-1}
\|[H_{l_1},\sum_{l_2=l_1+1}^{L}H_{l_2}]\|_F,\\
T'_2&:=\frac{1}{12\sqrt{d}} \sum_{l_1=1}^{L}\|[\sum_{l_2=l_1+1}^L H_{l_2},[\sum_{l_2=l_1+1}^L H_{l_2},H_{l_1}]]\|_F\nonumber\\
&\quad+  \frac{1}{24\sqrt{d}}  \sum_{l_1=1}^{L}
\| [H_{l_1},[H_{l_1},\sum_{l_2=l_1+1}^L H_{l_2}]]  \|_F.
\end{align}
\end{theorem}

Numerical results shown in Fig.~\ref{Fig:PF12} suggest that these tighter bounds can be close to optimal for PF2.
However, we find a gap between the triangle bounds (green curve) and empirical results (blue curve) for PF1 with a nearest-neighbor Hamiltonian.
This phenomenon results from destructive error interference between different segments that is not captured when applying the triangle inequality, as also seen in previous worst-case analysis \cite{Tran_2020}. Here we further tighten the $t$-dependence of the average error for PF1
from $\mathcal O\bigl(\sqrt{n}\frac{t^2}{r}\bigr)$ to $\mathcal O\bigl(\sqrt{n}\bigl(\frac{t}{r}+\frac{t^3}{r^2}\bigr)\bigr)$.

\begin{theorem}[Interference bound]\label{M_Th:interference}
Consider the Hamiltonian $H=\sum_{j,j+1}H_{j,j+1}$ where
$H_{j,j+1}$ acts nontrivially on qubits $j, j+1$, and $\|H_{j,j+1}\|\le 1$. Let
$U(t):=\left(e^{-iAt/r}e^{-iBt/r}\right)^r$ where $A:=\sum_{\text{odd}~j} H_{j,j+1}$, $B:=\sum_{\text{even}~j} H_{j,j+1}$.
If $\|[A,B]\|t^2/2r$ is at most a constant smaller than 1,
the average error of PF1 is
\begin{equation}
 R(U(t),U_0(t))=\mathcal{O}\left( \sqrt{n}\left(\frac{t}{r}+\frac{t^3}{r^2}\right)\right).
\end{equation}
\end{theorem}
This refined error analysis better reflects the empirical results (see Fig.~\ref{Fig:PF12}). We explain this in more detail in Section~\ref{Sec:numerical} and Appendix~\ref{Sec:interference}.

Finally, we can also bound the average error in the $K$th-order truncated Taylor series method, reducing
the worst-case error $\mathcal O(\alpha t \frac{(\ln 2)^{K+1}}{(K+1)!}) $ to $\mathcal O(\max_i \sqrt{\alpha_i \alpha} \, t \frac{(\ln 2)^{K+1}}{(K+1)!})$ where $H= \sum_{l=1}^L \alpha_l H_l$ and $\alpha:=\sum_{i=1}^L \alpha_i$. However, since the gate complexity of the LCU method is logarithmic in ${1}/{\varepsilon}$, the gate reduction for the Taylor series method is mild.

\begin{figure*}[!htb]
\centering
\includegraphics[width=1\textwidth]{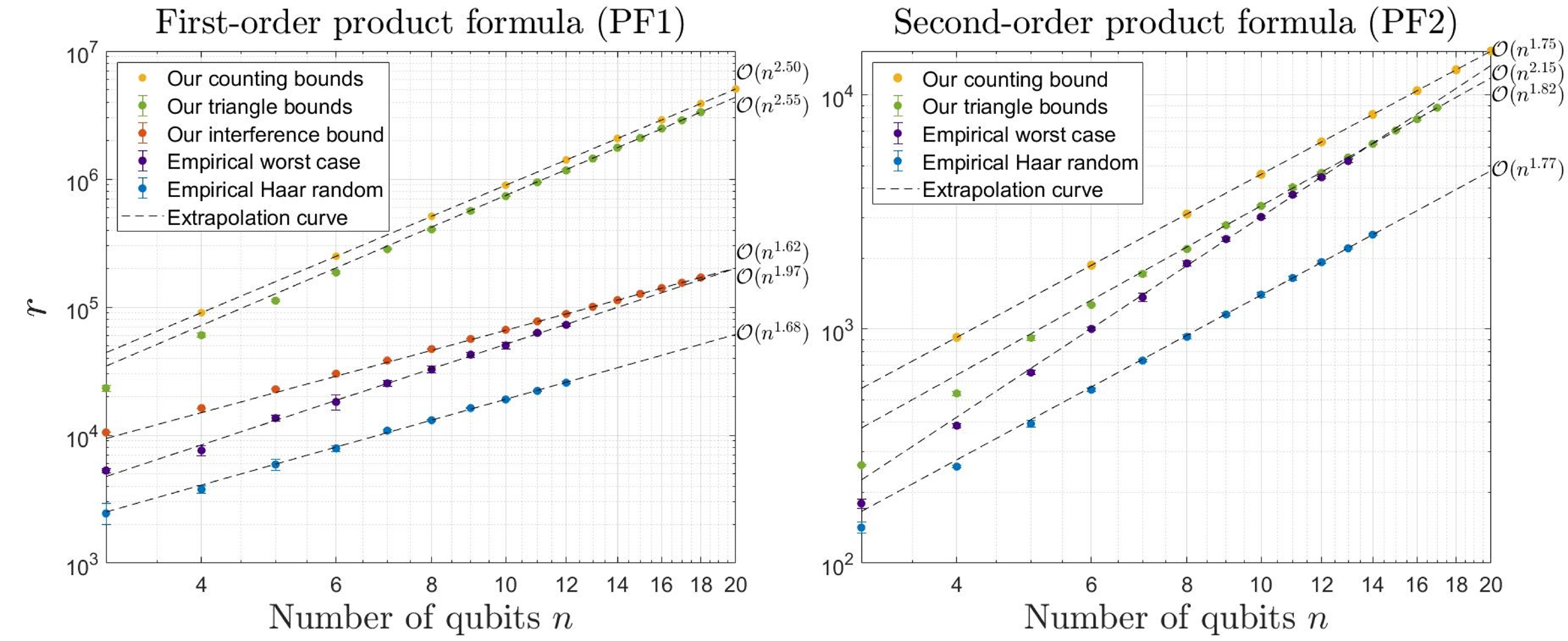}
\caption{Comparison of minimum $r$ using different error bounds for the one-dimensional Heisenberg model in Eq.~\eqref{M_Eq:heisenberg}. For each system size, we generate five Hamiltonians $H_i$ with random
coefficients.
We plot the mean and standard deviation of $r(t,\varepsilon,H_i)$ defined in Eq.~\eqref{Eq:averagetrotter}. Here the triangle and interference bounds correspond to the results in Theorem~\ref{Th:PF12} and Theorem~\ref{M_Th:interference}, respectively. The counting bound (presented in Appendix~\ref{app_count_near}) directly estimates the triangle bound of Theorem~\ref{Th:PF12} without numerically evaluating the trace of the commutators. For PF1, the asymptotic scaling of the interference bound of Theorem~\ref{M_Th:interference} (extrapolated from small system size) matches the empirical result, but is loose by a factor of $3.4$ at $n=12$. For PF2, the asymptotic scaling of the triangle bound in Theorem~\ref{Th:PF12} matches the empirical results, and is only loose by a factor of $2.44$ at $n=14$.}
\label{Fig:PF12}
\end{figure*}

\section{Applications}\label{sec:applications}
\paragraph{Lattice Hamiltonians.}
For a lattice Hamiltonian with nearest-neighbor interactions, we have
$T_p=\mathcal O(\sqrt{n})$ for the $p$th-order Trotter algorithm ($p\ge 1$). For comparison, the corresponding worst-case parameter is $\alpha_{\mathrm{comm},p}= \mathcal O(n)$.
Consequently, the asymptotic gate complexity is reduced
from $\mathcal O(n^2t^2)$ in the worst case to $\mathcal O(n^{1.5}t^2)$ on average for PF1 (triangle bound), and from $\mathcal O(n^{1.5}t^{1.5})$ in the worst case to $\mathcal O(n^{1.25}t^{1.5})$ on average for PF2. These theoretical bounds for PF methods agree well with the empirical results shown in Fig.~\ref{Fig:PF12}.

\paragraph{$k$-local Hamiltonians.}
Consider a $k$-local Hamiltonian $H=\sum_{l_1,\dots,l_k}
H_{l_1,\dots,l_k}$ acting on $n$ qubits, where each $H_{l_1,\dots,l_k}$
acts nontrivially on at most $k$ qubits. We show that the average error is related to the sum of Frobenius norms  $\|H\|_{1,F}:=\sum_{l_1,\dots,l_k}\|H_{l_1,\dots,l_k}\|_F$, an induced permutation norm $\normH{H}_{\mathrm{per}}$ defined in Eq.~\eqref{eq:perNorm}, and the induced one-norm $\normH{H}_1$ defined in Ref.~\cite{childs2020theory}. Specifically, we find $T_1'= \mathcal O\bigl(\frac{1}{\sqrt{d}}\|H\|_{1,F}\,\normH{H}_{\mathrm{per}}\bigr)$ and $T_2'=\mathcal O\bigl(\frac{1}{\sqrt{d}}\|H\|_{1,F}\, \normH{H}_1\,\normH{H}_{\mathrm{per}}\bigr)$. For comparison, the commutator bounds in the worst case are $\alpha_{\mathrm{comm},1}=\mathcal O\left(\|H\|_1\,\normH{H}_{1}\right)$ and $\alpha_{\mathrm{comm},2}=\mathcal O\bigl(\|H\|_1\,\normH{H}_1^2\bigr)$. For simplicity, consider the case where $\|H_{l_1,\dots,l_k}\|\le 1$ for each term.
Then $\|H\|_1\le n^k$, $\frac{1}{\sqrt{d}}\|H\|_{1,F}\le n^k$, $\normH{H}_{\mathrm{per}}=\mathcal O(n^{\frac{k-1}{2}})$, and   $\normH{H}_1=\mathcal O(n^{k-1})$. We summarize the resulting errors for the PF1 and PF2 methods in Table~\ref{Table:klocal}.

\begin{table}[htb]
\renewcommand{\arraystretch}{1.5}
\begin{tabular}{|c|c|c|}
\hline
Order & Worst-case error \cite{childs2020theory} & Average error \\ \hline
$p=1$   &   $\mathcal O\bigl(\frac{t^2}{r}n^{2k-1}\bigr) $        &    $ \mathcal O\bigl(\frac{t^2}{r}n^{\frac{3k-1}{2}}\bigr)  $         \\ \hline
$p=2$   &$\mathcal O\bigl(\frac{t^3}{r^2}n^{3k-2}\bigr)$   &    $\mathcal O\bigl(\frac{t^3}{r^2}n^{\frac{5k-3}{2}}\bigr)$          \\ \hline
\end{tabular}
\caption{Errors for $k$-local Hamiltonian simulation with PF1 and PF2 methods.}
\label{Table:klocal}
\end{table}

\paragraph{Power-law interactions.}
Consider power-law interactions on a $D$-dimensional lattice $\Lambda \subset \mathbb{R}^D$, with 2-site interactions $H = \sum_{i,j\in\Lambda} H_{i,j}$. Suppose the interaction strength decays as the power law
\begin{equation}\label{M_eq:power}
\|H_{i,j}\| \le \left\{
\begin{aligned}
&1 ~&i=j  \\
&\frac{1}{\|i-j\|^{\alpha}} ~&i\neq j
\end{aligned}
\right.
\end{equation}
for some $\alpha \ge 0$,
where $\|i-j\|$ denotes the Euclidean distance.
We show in Appendix~\ref{Sec:App_powerlaw} that $\normH{H}^2_{\mathrm{per}}\le \normH{H}_1$.
We present the error scaling for power-law interactions with $0\le \alpha<D$, $\alpha=D$, and $\alpha>D$ in Table~\ref{Table:power}. The comparison to empirical performance shown in Fig.~\ref{Fig:PFpower} suggests that our theoretical bounds are reasonably tight.

\begin{table*}[htb]
\renewcommand{\arraystretch}{1.5}
\begin{tabular}{|c|c|c|c|c|c|c|}
\hline
\multirow{2}{*}{Order} & \multicolumn{2}{c|}{$0\le \alpha<D$}      & \multicolumn{2}{c|}{$\alpha=D$}      & \multicolumn{2}{c|}{$\alpha>D$}      \\ \cline{2-7}
                       & Worst-case error& Average error & Worst-case error& Average error & Worst-case error & Average error \\ \hline
$p=1$                    &   $\mathcal O (\frac{t^2}{r}n^{3-2\alpha/ D})$         &       $\mathcal O \bigl(\frac{t^2}{r}n^{\frac{5}{2}-3\alpha/2D} \bigr)$        &    $\mathcal O \bigl(\frac{t^2}{r}n\log^{2}(n)\bigr)   $     &       $\mathcal O \bigl(\frac{t^2}{r}n\log^{\frac{3}{2}}(n)\bigr)$         &    $\mathcal O  \bigl(\frac{t^2}{r}n\bigr)$          &      $\mathcal O  \bigl(\frac{t^2}{r}n\bigr)$         \\ \hline
$p=2 $                   &      $\mathcal O (\frac{t^3}{r^2}n^{4-3\alpha/ D})$      &      $\mathcal O \bigl(\frac{t^3}{r^2}n^{\frac{7}{2}-5\alpha/2D }\bigr)$         &    $\mathcal O \bigl(\frac{t^3}{r^2}n\log^{3}(n)\bigr)$        &    $\mathcal O \bigl(\frac{t^3}{r^2}n\log^{\frac{5}{2}}(n)\bigr)$            &  $\mathcal O  \bigl(\frac{t^3}{r^2}n\bigr)$         &           $\mathcal O  \bigl(\frac{t^3}{r^2}n\bigr)$    \\ \hline
\end{tabular}\caption{Errors for PF1 and PF2 simulations of power-law interaction Hamiltonians.}\label{Table:power}
\end{table*}

\paragraph{Out-of-time-order correlators.}
As a final example, consider the (infinite-temperature) out-of-time-order correlator (OTOC) for
two commuting local observables $X$ and $Y$,
$\langle O(t) \rangle :=    \langle Y^{\dagger}(t) X^{\dagger}  Y(t) X \rangle$ where  $Y(t):= e^{iHt}Ye^{-iHt}$ in the Heisenberg picture \cite{shenker2014black,maldacena2016bound}.
To measure the OTOC in an experiment \cite{li2017measuring}, an initial state $\rho \otimes I_{d_1}/d_1$ is prepared where $\rho$ is the state of the first qubit
and $I_{d_1}/d_1$ is the maximally mixed state of the remaining $n-1$ qubits, where $d_1 := 2^{n-1}$. After the unitary evolution $V_0:=e^{iHt}Ye^{-iHt}$ (assuming $Y$ is unitary), the measurement $X$ is performed on the first qubit.
In practice, $e^{iHt}$ and $e^{-iHt}$ can be approximated via PF methods, giving $V=V_0(I+\mathscr{M})$.
Viewing the initial state $\rho \otimes I_{d_1}/d_1$ as a mixture of randomly chosen states on an $(n-1)$-qubit subsystem, we can use our results to quantify the Trotter error, reducing it from $\|\mathscr{M}\|$ to $\frac{1}{\sqrt{d_1}}\|\mathscr{M}\|_F$.
For example, suppose $H$ is a nearest-neighbor Hamiltonian (as described above) and $e^{iHt}$ and $e^{-iHt}$ are approximated via PF2. Then for a given constant Trotter error $\varepsilon$, we can reduce the gate complexity of an OTOC measurement from $\mathcal O(n^{1.5}t^{1.5})$ with worst-case analysis to $\mathcal O(n^{1.25}t^{1.5})$.

Similar techniques can be used to quantify the Trotter error in algorithms for estimating $\tr(e^{-iHt})$,
with applications to trace estimation in the one clean qubit model \cite{PhysRevLett.81.5672}.
See Appendix~\ref{App_trace} for further details.

\section{Numerical results}\label{Sec:numerical}

To investigate the tightness of our theoretical bounds, we compare them with empirical error data.
We measure the complexity in terms of the \textit{average Trotter number} for the $p$th-order formula $\mathscr{U}_p$ for a given instance of Hamiltonian $H_i$,
\begin{align}\label{Eq:averagetrotter}
\hspace{-2mm} r(t,\varepsilon,H_i)&:=\min\{r\!:\! \mathbb{E}_{\psi}\|\left(\mathscr{U}^r_p(\tfrac{t}{r}) \!-\! e^{-itH_i}\right)\!\ket{\psi}\|_{2}\!\le\! \varepsilon \},
\end{align}
which guarantees the expectation value of the error is below a given simulation accuracy $\varepsilon$. For each instance of $H_i$, we set evolution time $t=n$, error threshold $\varepsilon= 10^{-3}$, and generate 20 Haar-random inputs.
For the worst-case empirical analysis, we directly compute the spectral norm error and obtain worst-case Trotter number $r_W(t,\varepsilon,H_i) :=\min\{r: \|\mathscr{U}^r _p(\frac{t}{r})- e^{-itH_i}\|\le \varepsilon\}$.
In this section, we plot the mean and standard deviation of $r(t,\varepsilon,H_i)$ and $r_W(t,\varepsilon,H_i)$.

We first consider the one-dimensional Heisenberg model with a random magnetic field,
\begin{equation}\label{M_Eq:heisenberg}
  \begin{aligned}
   \hspace{-2mm} H = \sum^{n-1}_{j=1} \left(X_j X_{j+1} + Y_jY_{j+1} + Z_j Z_{j+1}\right) +\sum^{n}_{j=1} h_j Z_j,
\end{aligned}
\end{equation}
with uniformly random coefficients $h_j \in [-1,1]$.
The Hamiltonian summands can be partitioned into two sets in an even-odd pattern \cite{Childs2019Product}.
Fig.~\ref{Fig:PF12} compares the empirical Trotter number for PF1 and PF2 with the theoretical bounds in Theorems \ref{Th:PF12} and \ref{M_Th:interference}.
For PF1, the interference bound curve matches the empirical result well in the regime $n\le 18$.
However, because of the additional assumption of Theorem~\ref{M_Th:interference} that $\|[A,B]\|t^2/2r=\mathcal O(nt^2/r)$ must be less than a constant smaller than 1, for large $t$ and $n$, $\mathcal O(nt^2/r)$ will dominate, significantly increasing $r$.
The analysis in the worst case also suffers from a similar problem \cite{Tran_2020}. We elaborate on this point in Appendix~\ref{Sec:numerical_inter}. We conjecture that the additional assumptions could be relaxed for both worst- and average-case error, but we leave this as a question for future research.

\begin{figure*}[!htb]
\centering
\includegraphics[width=1\textwidth]{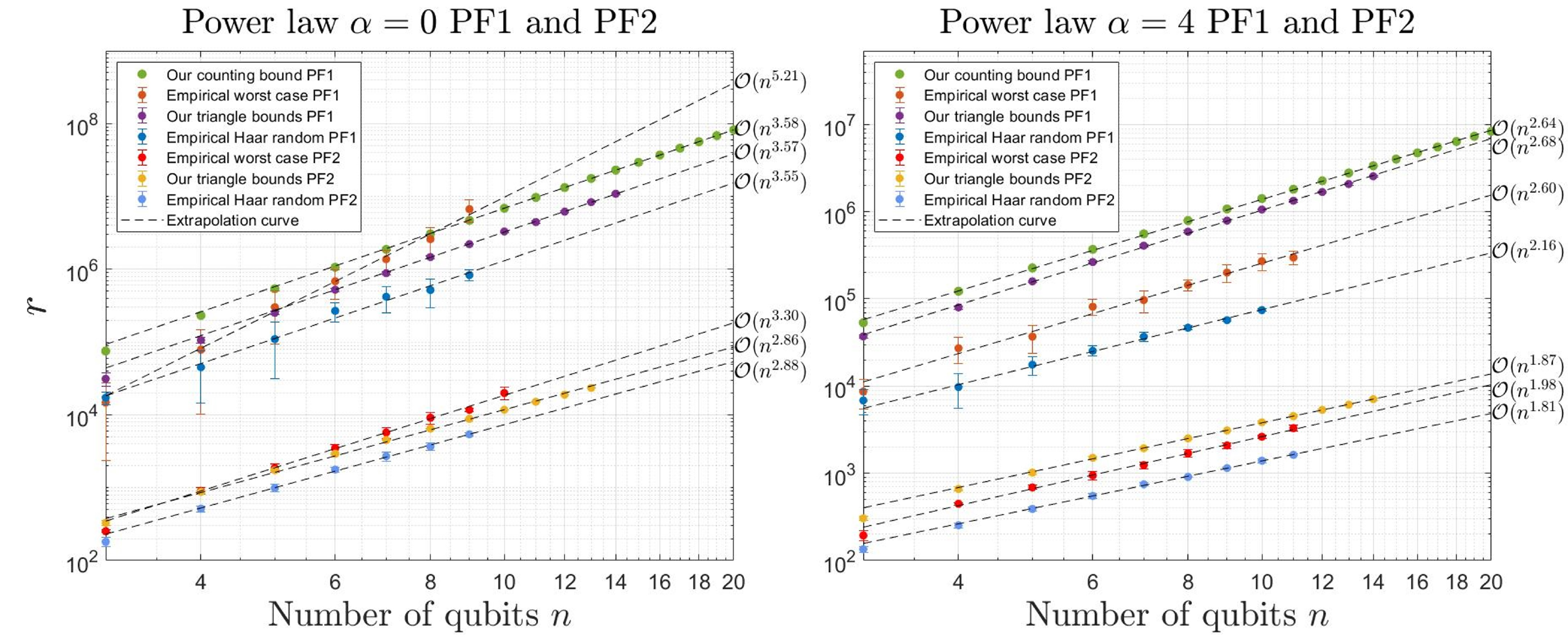}
\caption{Comparison of minimum $r$ using different error bounds for the one-dimensional Heisenberg model with power-law interactions in Eq.~\eqref{M_Eq:powerlaw}.
Here the triangle and counting bounds correspond to the results in Theorem~\ref{Th:PF12} and Appendix~\ref{app_count_power}, respectively. For $\alpha=0$, our theoretical asymptotic scaling using Theorem \ref{Th:PF12} agrees well with empirical results; it is loose by factors of $2.68$ and $1.64$ for PF1 ($n=7$) and PF2 ($n=9$), respectively.
For $\alpha=4$, the results for PF2 also show good agreement, being loose by a factor of $2.73$ at $n=11$. We find that there is a clear gap between the empirical results and theoretical predictions for PF1. We speculate that this is because we use the triangle inequality in the theoretical analysis, and that the error between Trotter steps also interferes instead of adding linearly for this type of Hamiltonian.
}\label{Fig:PFpower}
\end{figure*}

We also consider the one-dimensional Heisenberg model with power-law interactions, with the Hamiltonian
\begin{equation}\label{M_Eq:powerlaw}
\begin{aligned}
\hspace{-3mm}\sum^{n-1}_{j=1}\sum^n_{k = j+1} \! \frac{1}{|j - k|^\alpha} \left(X_j X_{k} + Y_j Y_{k} + Z_j Z_{k}\right)+\sum^{n}_{j=1}h_j Z_j\!\!
\end{aligned}
\end{equation}
with uniformly randomly coefficients $h_j\in [-1,1]$ and $\alpha$ a parameter controlling the decay of the interactions. We consider a rapidly decaying power law with $\alpha=4$ and the infinite-range case with $\alpha=0$.
We analyze product formulas with $X$-$Y$-$Z$ order \cite{childs2018toward}, as shown in Fig.~\ref{Fig:PFpower}. We again find that our error bounds are reasonably tight in many cases, but are somewhat loose for PF1 with shorter-range interactions.

In Appendix~\ref{sec:numerical_empirical}, we also show empirical results for PF4 and PF6 (which align well with our asymptotic results in Theorem~\ref{M_Th:general}), other 1-design inputs, and the standard deviations of random inputs.

\section{Conclusions and open problems}
In this work, we have developed
a theory of average error for Hamiltonian simulation with random input states. Though previous methods already provide optimal performance in terms of worst-case error, we find further improvement when considering the average case.
We also show that the error in an OTOC measurement is related to the error for random initial states, and thereby reduce the gate complexity of OTOC simulations using product formulas, a potentially promising approach to demonstrating quantum information scrambling.
Our techniques might also be extended to quantify the algorithmic error in imaginary time evolution \cite{motta2020determining} and quantum Monte Carlo methods \cite{bravyi2015monte,bravyi2017poly}.
More detailed analysis of the average error in subsystems or subspaces, such as in a low-energy subspace, is also worth further exploration.
In Appendix~\ref{sec:cauchycompare}, we also explore the possibility of direct calculation of $R(U,U_0)$ without using the Cauchy inequality, which could potentially lead to tighter bounds.

\section*{Acknowledgements}
We are grateful to Jiaqi Leng and Xiaodi Wu for useful discussions. We also thank Jiaqi Leng for assistance with our numerical results.

After this work was completed, we became aware of related work by Chi-Fang Chen and Fernando Brandão that also analyzes Trotter error with random input states. We thank them for letting us know about their work.

QZ and AFS acknowledge the support of the Department of Defense through the QuICS Hartree Postdoctoral Fellowship and Lanczos Graduate Fellowship, respectively.  YZ is supported by the National Research Foundation of Singapore under its NRF-ANR joint program (NRF2017-NRF-ANR004 VanQuTe), the Quantum Engineering Program QEP-SP3, the Singapore Ministry of Education Tier 1 grant RG162/19, FQXi-RFP-IPW-1903 from the Foundational Questions Institute and Fetzer Franklin Fund, a donor advised fund of Silicon Valley Community Foundation. Any opinions, findings and conclusions or recommendations expressed in this material are those of the author(s) and do not reflect the views of the National Research Foundation, Singapore. TL was supported by the NSF grant PHY-1818914 and a Samsung Advanced Institute of Technology Global Research Partnership. AMC received support from the National Science Foundation (grants CCF-1813814 and OMA-2120757) and the Department of Energy, Office of Science, Office of Advanced Scientific Computing Research, Quantum Algorithms Teams and Accelerated Research in Quantum Computing programs.

\bibliographystyle{apsrev4-2}
\bibliography{bibsimulation}

\onecolumngrid
\newpage

\renewcommand{\addcontentsline}{\oldacl}
\renewcommand{\tocname}{Appendix Contents}
\tableofcontents

\clearpage

\begin{appendix}

\section{Average error and its variance}\label{Sec:Average}

We consider the average error for input states chosen at random from some ensemble.

\begin{definition}
For an ensemble $\mc{E}$ of quantum states, the average error between $U$ and $U_0$ with respect to the distance measure $\mc{D}$ is
\begin{equation}
\begin{aligned}
 R_{\mc{D}}^{\mc{E}}(U_0,U)&:= \mathbb{E}_{\mathcal{E}}[\mathcal{D}(U\ket{\psi},U_0\ket{\psi})].
\end{aligned}
\end{equation}
\end{definition}
Here $\mathcal{E}$ can be a discrete or continuous ensemble. Hereafter we take $\mathcal{D}$ to be the $\ell_2$ norm or the trace norm, but one can similarly consider other distance measures.

While it is simple to consider Haar-random inputs, this distribution is challenging to realize in practice. Fortunately, our results hold under weaker assumptions on the distribution of inputs---in particular, it suffices for the input to be a 1-design. In general, a complex projective $t$-design is a distribution that agrees with the Haar ensemble for any homogeneous degree-$t$ polynomial of the state and its conjugate. Formally, a probability distribution over quantum states $\mathcal{E}=\{(p_{i}, \phi_{i})\}$ is a complex projective $t$-design if
\begin{equation}
\begin{aligned}\label{Eq:tdesign}
\sum_{\phi_i \in \mathcal{E}} p_i\ket{\phi_i}\bra{\phi_i}^{\otimes t}=\int_{\mathrm{Haar}} \ket{\psi}\bra{\psi}^{\otimes t} \d \psi,
\end{aligned}
\end{equation}
where the sum on the left is over the (discrete) ensemble $\mathcal{E}$, and the integral on the right is over the Haar measure. It is clear that if $\mathcal{E}$ is a $t$-design then it is also a $(t-1)$-design. For a more detailed introduction to $t$-designs, see for example Refs.~\cite{low2010pseudo,zhu2016clifford}.

Let $\mc{H}_d$ denote a $d$-dimensional Hilbert space.
The integral over Haar-random states of the projector onto $t$ copies of the state is proportional to the projector $\Pi_+$ onto the symmetric subspace of $\mc{H}_d^{\otimes t}$ \cite{Harrow2013symmetric}:
\begin{equation}\label{Eq:tdesignsym}
\begin{aligned}
\int_{\mathrm{Haar}} \ket{\psi}\bra{\psi}^{\otimes t} \d \psi=\frac{\Pi_+}{D_+}&=\frac{\sum_{\pi\in S_t}W_{\pi}}{t!D_+}.
\end{aligned}
\end{equation}
Here $D_+:=\binom{d+t-1}{t}$ is the dimension of the symmetric subspace, $S_t$ is the symmetric group of order $t$, and $W_{\pi}$ is the unitary representation of $\pi\in S_t$ that permutes the states of each copy in $\mc{H}_d^{\otimes t}$ according to $\pi$.
For $t=1$, the only element in the symmetric group is the identity, so
\begin{equation}\label{Eq:onedesignsym}
\begin{aligned}
\int_{\mathrm{Haar}} \ket{\psi}\bra{\psi}\d \psi=\id/d.
\end{aligned}
\end{equation}
For $t=2$, $S_2$ has two elements, the permutations $(1)(2)$ and $(1,2)$. We have $W_{(1)(2)}=\id^{\otimes 2}$ and $W_{(1,2)}=\s$ (the swap operator on the 2-copy space, with $\s \ket{\psi_1}\ket{\psi_2}= \ket{\psi_2}\ket{\psi_1}$). Therefore
\begin{equation}\label{Eq:twodesignsym}
\begin{aligned}
\int_{\mathrm{Haar}} \ket{\psi}\bra{\psi}^{\otimes 2}\d \psi=\frac{\id^{\otimes 2}+\s}{d(d+1)}.
\end{aligned}
\end{equation}

A locally random state $\bigotimes_{i=1}^n u_i\ket{0}^{\otimes n}$, with each $u_i$ being a single-qubit Haar random unitary,
is a 1-design. This also remains a 1-design if we choose the local unitary from $\{\id_i,X_i\}$ with equal probability,
so the computational basis (or the uniform distribution over any basis of the Hilbert space) forms 1-design.
The set of random stabilizer states, which can be obtained by applying a random Clifford circuit to the state $\ket{0}^{\otimes n}$, forms 3-design \cite{Webb15,Zhu15}. Note that if the state ensemble $\mathcal{E}=\{(p_i,\phi_i)\}$ is a $t$-design, then for any unitary operation $V$, the ensemble $\mathcal{E}'=\{(p_i,V\phi_iV^{\dag})\}$ is also a $t$-design.

We can also extend our results to the approximate case.
We say that an ensemble $\mathcal{E}=\{(p_{i}, \phi_{i})\}$ is an $\epsilon$-approximate $t$-design if
\begin{equation}\label{eq:AppDesign}
\begin{aligned}
(1-\epsilon)\int_{\mathrm{Haar}} \ket{\psi}\bra{\psi}^{\otimes t} \d \psi\le \sum_{\phi_i \in \mathcal{E}} p_i\ket{\phi_i}\bra{\phi_i}^{\otimes t}\le (1+\epsilon)\int_{\mathrm{Haar}} \ket{\psi}\bra{\psi} ^{\otimes t}\d \psi.
\end{aligned}
\end{equation}

We can relate the average error for a Haar-random state to the average over a $t$-design ensemble. Considering the $\ell_2$ norm for example, the average error is
\begin{equation}\label{eq:l2upp}
\begin{aligned}
R_{\ell_2}^{\mathrm{Haar}}(U_0,U)&=\mathbb{E}_{\psi \in \mathrm{Haar}}\|(U-U_0)\ket{\psi}\|_{\ell_2}\\
&=\mathbb{E}_{\psi \in \mathrm{Haar}} \sqrt{2-\braket{\psi |U^{\dagger}U_0|\psi}-\braket{\psi |U^{\dagger}_0U|\psi}}\\
&\le \left(2-\mathbb{E}_{\psi \in \mathrm{Haar}} \braket{\psi |U^{\dagger}U_0|\psi}+\braket{\psi |U^{\dagger}_0U|\psi} \right)^{\frac{1}{2}} =: \tilde{R}_{\ell_2}(U_0,U),
\end{aligned}
\end{equation}
where the bound is due to the Cauchy–Schwarz inequality.

Since this bound only involves a first-order function of $\psi$, the average error with respect to a 1-design $\mc{E}$ satisfies
\begin{equation}\label{eq:l2upp1}
\begin{aligned}
R_{\ell_2}^{\mc{E}}(U_0,U)\leq \tilde{R}_{\ell_2}(U_0,U).
\end{aligned}
\end{equation}
We can also in principle evaluate the error by considering exact formulas for evaluating integrals of the form
\begin{equation}
\label{eq:haarintegral}
    \int_\mathrm{Haar} f(\bra{\psi} A \ket{\psi} ) \, \d \psi,
\end{equation}
where $A$ is a Hermitian operator (see for example \cite[Eq.~19]{Jones_1991}). In the present case, we have $A = 2 - U^\dagger U_0 - U_0^\dagger U$ and $f(x)=\sqrt{x}$. These formulas are in terms of the eigenvalues of $A$ and antiderivatives of $f$, and can be difficult to evaluate in practice. However, one may estimate the integral by sampling. We elaborate on this approach and how it compares to the Cauchy-Schwarz bound $\tilde R_{\ell_2}$ in Appendix \ref{sec:cauchycompare}.

\subsection{Average \texorpdfstring{$\ell_2$}{l2} norm  with random inputs}\label{subsec:l-2}
Here we focus on quantifying the error performance in terms of the $\ell_2$ norm. Recall that Eq.~\eqref{eq:l2upp} bounds the error by \begin{equation}\label{eq:l2upp-defn}
\begin{aligned}
\tilde{R}_{\ell_2}(U_0,U)=\left(2-\mathbb{E}_{\psi} \braket{\psi|U^{\dagger}U_0+U^{\dagger}_0U|\psi} \right)^{\frac{1}{2}}
= [\mathbb{E}_{\psi} S(\psi)]^{\frac1{2}},
\end{aligned}
\end{equation}
where $S(\psi) := 2 - \braket{\psi|U^{\dagger}U_0+U^{\dagger}_0U|\psi}$. The following Lemma computes the mean and variance of $S$ for a 1-design and a 2-design, respectively, in terms of the Frobenius norm $\|X\|_F := \sqrt{\tr (XX^{\dagger})}$.

\begin{lemma}\label{Lemma:S}
For two $d$-dimensional unitaries $U$ and $U_0$ with $U=U_0(\id+\mathscr{M})$, the expectation of $S$ with respect to a 1-design ensemble $\mathcal{E}$ is
\begin{equation}\label{eq:avSresult}
  \mathbb{E}_{\psi} S(\psi)=\frac{1}{d} \|\mathscr{M}\|_F^2.
\end{equation}
If $\mathcal{E}$ is also a 2-design, the variance has the upper bound
\begin{equation}\label{eq:VarSresult}
   \mathrm{Var}(S(\psi)) \leq \frac{4\|\mathscr{M}\|_F^2}{d(d+1)}.
\end{equation}
\end{lemma}

\begin{proof}
For simplicity, let $A:=U^{\dagger}_0U=\id+\mathscr{M}$ in this proof. The expectation of $S(\psi)$ is
\begin{equation}\label{eq:avSpsi}
\begin{aligned}
\mathbb{E}_{\psi} S(\psi)
&= 2-\mathbb{E}_{\psi}\tr\left[(A+A^{\dag}) \ket{\psi}\bra{\psi} \right]\\
&=2-\tr\left[(A+A^{\dag})\id/d\right]\\
&=2-\frac{\tr(A)+\tr(A^{\dag})}{d}=-\frac{\tr(\mathscr{M})+\tr(\mathscr{M}^\dagger)}{d}\\
&=\frac{\tr(\mathscr{M}\mathscr{M}^{\dagger})}{d}=\frac{1}{d} \|\mathscr{M}\|_F^2.
\end{aligned}
\end{equation}
Here in the final line we use the fact that $\mathscr{M}\mathscr{M}^{\dagger}=-(\mathscr{M}+\mathscr{M}^\dagger)$ by the unitarity of $U^{\dagger}_0U$.

To calculate the second moment of $S(\psi)$, we first compute
\begin{equation}
\begin{aligned}\label{}
\mathbb{E}_{\psi}(2-S(\psi))^2 &=
\mathbb{E}_{\psi} \tr\left[(A+A^{\dag})^{\otimes 2}\ket{\psi}\bra{\psi}^{\otimes 2}\right]\\
&=\frac{1}{d(d+1)}\tr\left[(A+A^{\dag})^{\otimes 2}(\id^{\otimes 2}+\s)\right]\\
&=\frac{[\tr(A)+\tr(A^{\dag})]^2+\tr[(A+A^{\dag})^2]}{d(d+1)}\\
&=\frac{d^2(2-\mathbb{E}_{\psi}S(\psi))^2+\tr(A^2)+\tr(A^{\dag2})+2d}{d(d+1)},
\end{aligned}
\end{equation}
where the second equality is by Eq.~\eqref{Eq:twodesignsym} and the last equality uses Eq.~\eqref{eq:avSpsi}.
As a result, the variance is
\begin{equation}
\begin{aligned}
\mathrm{Var}(S(\psi))&=\mathrm{Var}(2-S(\psi))\\
&=\mathbb{E}_{\psi}(2-S(\psi))^2-(2-\mathbb{E}_{\psi}S(\psi))^2\\
&=\frac{\tr(A^2)+\tr(A^{\dag2})+2d-d(2-\mathbb{E}_{\psi}S(\psi))^2}{d(d+1)}\\
&=\frac{\tr(\mathscr{M}^2)+\tr(\mathscr{M}^{\dag2})-2[\tr(\mathscr{M})+\tr(\mathscr{M}^{\dag})]}{d(d+1)}-\frac{[\tr(\mathscr{M})+\tr(\mathscr{M}^{\dag})]^2}{d^2(d+1)}\\
&\le  \frac{\tr(\mathscr{M}^2)+\tr(\mathscr{M}^{\dag2})+2\tr(\mathscr{M}\mathscr{M}^{\dag})}{d(d+1)}\\
&\le  \frac{4\tr(\mathscr{M}\mathscr{M}^{\dag})}{d(d+1)},
\end{aligned}
\end{equation}
where in the first inequality we drop the second positive term, and in the second inequality we use the fact that $\tr(\mathscr{M}^2),\tr(\mathscr{M}^{\dag2})\leq \tr(\mathscr{M}\mathscr{M}^{\dag})$.
\end{proof}

\begin{theorem}\label{Th:l2}
For random inputs from a 1-design ensemble $\mc{E}$, the average $\ell_2$ distance between $d$-dimensional unitaries $U$ and $U_0$ with $U=U_0(\id+\mathscr{M})$ is upper bounded by
\begin{equation}
 R_{\ell_2}^{\mc{E}}(U_0,U)\le \frac{1}{\sqrt{d}}\|\mathscr{M}\|_F.
\end{equation}
\end{theorem}

\begin{proof}
The theorem follows from combining the definition of the upper bound in Eq.~\eqref{eq:l2upp1} and Eq.~\eqref{eq:l2upp-defn}, and the average value in Eq.~\eqref{eq:avSresult} of Lemma \ref{Lemma:S}.
\end{proof}

Theorem \ref{Th:l2} upper bounds the expected $\ell_2$ norm using the Cauchy–Schwarz inequality as in Eq.~\eqref{eq:l2upp}.
Next we use the variance result in Eq.~\eqref{eq:VarSresult} of Lemma \ref{Lemma:S} to show that for a 2-design ensemble, the bound is typical, and therefore tight.

\begin{corollary}\label{l2lower}
For random inputs from a 2-design ensemble $\mc{E}$,
if $\varepsilon := \tilde{R}_{\ell_2}(U_0,U) \gg \frac{1}{\sqrt{d}}$, then $R^\mc{E}_{\ell_2}(U_0,U)=\Theta (\varepsilon)$.
\end{corollary}

\begin{proof}

For simplicity we denote the average of $S(\psi)$ as $\bar{S}:=\mathbb{E}_{\psi} S(\psi)$.
By definition,
\begin{equation}
\begin{aligned}
R^\mc{E}_{\ell_2}(U_0,U)
&= \sum_i p_i\sqrt{S(\psi_i)}\\
&=\sum_{\{i|S(\psi)>\bar{S} +a\}} p_i\sqrt{S(\psi_i)}+\sum_{\{i|\bar{S}+a
\ge S(\psi_i)\ge  \bar{S}-a\}} p_i\sqrt{S(\psi_i)}+\sum_{\{i|S(\psi_i)<\bar{S} -a\}} p_i\sqrt{S(\psi_i)}\\
&\geq \sum_{\{i|\bar{S}+a
\ge S(\psi_i)\ge  \bar{S}-a\}} p_i\sqrt{S(\psi_i)}\\
&\ge  \sum_{\{i|\bar{S}+a
\ge S(\psi_i)\ge  \bar{S}-a\}} p_i \ \ \sqrt{\bar{S}-a}.
\end{aligned}
\end{equation}
Here in the second line we break the integral in three parts depending on the value of $S(\psi)$ with respect to $\bar{S}$ (and a parameter $a$ to be determined), in the first inequality we only keep the integral near $\bar{S}$, and in the second inequality we use $\sqrt{\bar{S}-a}$ as a lower bound.

Suppose $\bar{S}=\varepsilon^2$ so that $\mathrm{Var}(S)\le \frac{4\varepsilon^2}{d+1}$ by Lemma \ref{Lemma:S}. Choosing $a=\varepsilon^2/2$ and using the Chebyshev inequality, we have
\begin{equation}
\begin{aligned}
    \mathrm{Pr}\{|S(\psi_i)-\bar{S}|> a\} \le \frac{\mathrm{Var}(S)}{a^2}\le \frac{16}{(d+1)\varepsilon^2}.\\
   \end{aligned}
\end{equation}
When for example $\varepsilon \geq \frac{5}{\sqrt{d}}$, we have $\mathrm{Pr}\{|S-\bar{S}|> a\}\leq\frac{16}{25}$. Then we obtain the lower bound
\begin{equation}
\begin{aligned}
R(U_0,U)\ge  [1- Pr\{|S-\bar{S}|\ge a\}]\sqrt{\varepsilon^2/2}=\frac{9}{25\sqrt{2}}\varepsilon=\Omega(\varepsilon).
\end{aligned}
\end{equation}
With the upper bound $R^\mc{E}_{\ell_2}(U_0,U)\leq \varepsilon$ in Eq.~\eqref{eq:l2upp-defn} , we thus have $R^\mc{E}_{\ell_2}(U_0,U)=\Theta(\varepsilon)$ as claimed.
\end{proof}

Note that the assumption $\tilde{R}_{\ell_2}(U_0,U) =\varepsilon \gg \frac{1}{\sqrt{d}}$ is reasonable in typical applications of Hamiltonian simulation.
We usually require that the error is below some fixed threshold---or that it is polynomially small if Hamiltonian simulation is used within some efficient subroutine---whereas $d$ is exponential in the number of qubits.
A similar argument holds for the assumption in Corollary~\ref{Coro:fidelower}.

\subsection{Average trace norm distance with random inputs}

Here we continue to study the average error, but now using the trace norm (also called the Schatten-1 norm). Define
\begin{equation}
\begin{aligned}\label{}
R_t^{\mc{E}}(U_0,U)&:= \mathbb{E}_{\psi \in \mc{E}}\left[\left\|U\ket{\psi}\bra{\psi}U^{\dagger}-U_0\ket{\psi}\bra{\psi}U_0^{\dagger}\right\|_1\right].
\end{aligned}
\end{equation}
First, by the relationship between the $\ell_2$ norm and the trace norm, we have the following result.

\begin{corollary}\label{Cor:trNorm}
For random inputs from a 1-design ensemble $\mc{E}$, the average trace norm between $d$-dimensional unitaries $U$ and $U_0$ with $U=U_0(\id+\mathscr{M})$ is upper bounded by
\begin{equation}
R_t^\mc{E}(U_0,U)\leq 2R_{\ell_2}^\mc{E}(U_0,U)\le \frac{2}{\sqrt{d}}\|\mathscr{M}\|_F.
\end{equation}
\end{corollary}
The proof is based on the fact that for two pure states $\ket{\psi}$ and $\ket{\phi}$, the trace distance is upper bounded by the $\ell_2$ norm as
\begin{equation}
  \frac{1}{2}  \|\ket{\psi}\bra{\psi}-\ket{\phi}\bra{\phi}\|_1\le  \|\ket{\psi}-\ket{\phi}\|_{\ell_2}.
\end{equation}
Note that there is a factor of $1/2$ relating the trace distance and the trace norm.
We can apply the inequality for the states $U_0\ket{\psi}$ and $U\ket{\psi}$ for every $\ket{\psi}$ from $\mc{E}$, so the same follows for the ensemble average.  As a result, we can bound $R_t^\mc{E}(U(t),U_0(t))$ using the upper bound on $\ell_2$ norm in Theorem \ref{Th:l2}.

On the other hand, we can also use the relationship between trace distance and fidelity to bound the trace norm error. We first define the average fidelity as
\begin{equation}
\begin{aligned}
R_f^{\mc{E}}(U_0,U)&:= \mathbb{E}_{\psi\in \mc{E} }[F(U_0\ket{\psi},U\ket{\psi})]=\mathbb{E}_{\psi\in \mc{E} } [\bra{\psi}U_0^{\dagger}U\ket{\psi} \bra{\psi}U^{\dagger}U_0\ket{\psi}].\\
\end{aligned}
\end{equation}
We have the following result for $R_f(U,U_0)$.

\begin{lemma}\label{Th:FidGlobal}
For random inputs from a 2-design ensemble $\mc{E}$, the average fidelity between $d$-dimensional unitaries $U$ and $U_0$ with $U=U_0(\id+\mathscr{M})$ is
\begin{equation}\label{Eq:avg-fidelity}
 R_f^{\mc{E}}(U_0,U)
 =1-\frac{\|\mathscr{M}\|_F^2}{d+1}+\frac{|\tr(\mathscr{M})|^2}{d(d+1)}.
\end{equation}
\end{lemma}

\begin{proof}
For simplicity, denote $A:=U_0^{\dagger}U$ and thus $A^{\dag}=U^{\dagger}U_0$. The average fidelity is
\begin{equation}\label{AFidG}
\begin{aligned}
R_f^{\mc{E}}(U_0,U)
=\mathbb{E}_{\psi} \left\{\tr[(A\otimes A^{\dag})\psi^{\otimes 2} ]\right\}=\frac{1}{d(d+1)}\tr\left[(A\otimes A^{\dag})(\id^{\otimes 2}+\s)\right]=\frac{|\tr(A)|^2+d}{d(d+1)}.
\end{aligned}
\end{equation}
We then replace $A$ with $U_0^{\dagger}U=\id+\mathscr{M}$ and use $\mathscr{M}\mathscr{M}^{\dagger}=-(\mathscr{M}+\mathscr{M}^\dagger)$ (by unitarity of $A$) to obtain the result in Eq.~\eqref{Eq:avg-fidelity}.
\end{proof}

For any two quantum states $\rho$ and $\sigma$, the fidelity $F(\rho,\sigma):=\tr\left(\sqrt{\sqrt{\sigma}\rho\sqrt{\sigma}}\right)^2$ and the trace distance between the states are related as
\begin{equation}\label{eq:FidTr}
\begin{aligned}
1-\sqrt{F(\rho,\sigma)}\leq \frac1{2}\|\rho-\sigma\|_1\leq \sqrt{1-F(\rho,\sigma)}.
\end{aligned}
\end{equation}
Furthermore, for two pure states, the right inequality is saturated. Using Eq.~\eqref{eq:FidTr}, we can give both upper and lower bounds on the average trace norm in terms of the average fidelity:
\begin{equation}
\begin{aligned}\label{Eq:Rtupp}
R_t^{\mc{E}}(U_0,U)&= \mathbb{E}_{\psi}[\|U\ket{\psi}\bra{\psi}U^{\dagger}-U_0\ket{\psi}\bra{\psi}U_0^{\dagger}\|_1]\\
&=\mathbb{E}_{\psi} [2\sqrt{1-F(U_0\ket{\psi},U\ket{\psi})}]\\
&\le 2\left(1-\mathbb{E}_{\psi} [F(U_0\ket{\psi},U\ket{\psi})]
\right)^{\frac{1}{2}}=2(1-R_f^{\mc{E}}(U_0,U))^{\frac{1}{2}};
\end{aligned}
\end{equation}
and
\begin{equation}
\begin{aligned}\label{Eq:Rtlower}
R_t^{\mc{E}}(U_0,U)&\geq 2\mathbb{E}_{\psi} \left(1-\sqrt{F(U_0\ket{\psi},U\ket{\psi})}\right)\\
&\ge 2\left(1-\sqrt{\mathbb{E}_{\psi} F(U_0\ket{\psi},U\ket{\psi})}
\right)=2\left(1-\sqrt{R_f^{\mc{E}}(U_0,U)}\right),\\
\end{aligned}
\end{equation}
where the last inequality in both cases follows from the Cauchy-Schwartz inequality. By inserting the expression for $R_f(U,U_0)$ from Lemma \ref{Th:FidGlobal}, we obtain the following.

\begin{theorem}\label{Th:t}
For random inputs from a 2-design ensemble $\mc{E}$, the average trace norm distance between $d$-dimensional unitaries $U$ and $U_0$ with $U=U_0(\id+\mathscr{M})$ has the bounds
\begin{equation}\label{uplowTr}
 2\left[1-\sqrt{1-\frac{\|\mathscr{M}\|^2_F}{d+1}+\frac{|\tr(\mathscr{M})|^2}{d(d+1)}}\right] \le R_{t}^{\mc{E}}(U,U_0)\le 2\sqrt{\frac{\|\mathscr{M}\|^2_F}{d+1}-\frac{|\tr(\mathscr{M})|^2}{d(d+1)}}.
\end{equation}
\end{theorem}
The upper bound in Eq.~\eqref{uplowTr} is slightly enhanced compared with the one in Corollary \ref{Cor:trNorm}. Note that the upper bound is smaller than $\frac{2}{\sqrt{d+1}}\|\mathscr{M}(t)\|_F$, by throwing away the second positive term.

For completeness, we also calculate the variance of the average fidelity, and use it to give tighter lower bound for 4-design ensembles.

\begin{lemma}\label{Lem:Rfvariance}
For random $d$-dimensional inputs from a 4-design ensemble $\mc{E}$, the variance of the fidelity is
\begin{equation}\label{eq:varF1}
\mathrm{Var}\left[F(U_0\ket{\psi},U\ket{\psi})\right]=\mathbb{E}_{\psi\in \mc{E}}[F^2(U_0\ket{\psi},U\ket{\psi})]-R_f^{\mc{E}}(U_0,U)^2
\end{equation}
with
\begin{align}
&\mathbb{E}_{\psi}[F^2(U_0\ket{\psi},U\ket{\psi})]\\
&=\frac{1}{(d+3)(d+2)(d+1)d}\{
|\tr(\id+\mathscr{M})|^4+ |\tr(\id+\mathscr{M})|^2(4d+8) +\tr[(\id+\mathscr{M}(t))^2]\tr[(\id+\mathscr{M}^{\dag})]^2 \nonumber \\
&\quad+\tr[(\id+\mathscr{M})]^2\tr[(\id+\mathscr{M}^{\dag})^2]+\tr[(\id+\mathscr{M})^2]\tr[(\id+\mathscr{M}^{\dag})^2]+2d^2+6d \}.\label{Eq:VarF}
\end{align}

\end{lemma}

\begin{proof}
Letting $A:=U_0^{\dagger}U$, the first term reads
\begin{equation}\label{4copyF}
\begin{aligned}
\mathbb{E}_{\psi} [F^2(U_0\ket{\psi},U\ket{\psi})]
&=\mathbb{E}_{\psi} [(\bra{\psi}A\ket{\psi} \bra{\psi}A^{\dag}\ket{\psi})^2] \\
&=\tr[(A\otimes A^{\dag}\otimes A\otimes A^{\dag})\mathbb{E}_{\psi}
\psi^{\otimes 4}]\\
&=\frac1{4!D_+}\sum_{\pi\in S_4}\tr\left[W_{\pi}A\otimes A^{\dag}\otimes A\otimes A^{\dag}\right],
\end{aligned}
\end{equation}
where $D_+=\binom{d+3}{4}$ is the dimension of the symmetric subspace, according to the 4-design case in Eq.~\eqref{Eq:tdesignsym}. It is not hard to see that $\tr\left[W_{\pi}A\otimes A^{\dag}\otimes A\otimes A^{\dag}\right]$ is directly related to the cycle structure of $\pi$; for example, we have $\tr(AA^{\dag}A)\tr(A)$ for $\pi=(123)(4)$.
By exhaustive analysis of the $4!$ permutations in $S_4$, we find
\begin{align}
&\mathbb{E}_{\psi} [F^2(U_0\ket{\psi},U\ket{\psi})\nonumber] \\
&\quad=\frac{|\tr(A)|^4+ |\tr(A)|^2(4d+8) +\tr(A^2)\tr(A^\dag)^2+\tr(A)^2\tr(A^{\dag2})+2d^2+\tr(A^2)\tr(A^{\dag2})+6d}{(d+3)(d+2)(d+1)d}.
\end{align}
Then the claim follows by writing $A=\id+\mathscr{M}$.
\end{proof}

Next we estimate the variance of the average fidelity under additional assumptions.

\begin{lemma}\label{Le:varF}
For random inputs from a $d$-dimensional 4-design ensemble $\mc{E}$,
the  variance of the fidelity with random inputs in Eq.~\eqref{eq:varF1} is $\mathrm{Var}\left[F(U_0\ket{\psi},U\ket{\psi})\right]=\mathcal{O}({1}/{d})$
 when $\|\mathscr{M}(t)\| \ll 1$.
\end{lemma}

\begin{proof}
For convenience, we define $x$ and $y$ such that
\begin{equation}
\begin{aligned}
&|\tr(\id+\mathscr{M})|=d(1+x),\quad |\tr[(\id+\mathscr{M})^2]|=d(1+y).\\
\end{aligned}
\end{equation}
According to the result of $R_f^{\mc{E}}(U_0,U)$ from Eq.~\eqref{AFidG} along with $\overline{F^2}$ from Eq.~\eqref{Eq:VarF},
we have
\begin{align}
 R_f^{\mc{E}}(U_0,U)&=\frac{|\tr(\id+\mathscr{M})|^2+d}{d(d+1)}=\frac{d^2(1+x)^2+d}{d(d+1)}=(1+x)^2+\frac{1}{d}, \\
\overline{F^2}&\le
\frac{1}{(d+3)(d+2)(d+1)d}[
d^4(1+x)^4+ (4d+8)d^2(1+x)^2 + 2(1+y) (1+x)^2 d^3 +\mathcal{O} (d^2)]\nonumber \\
&=(1+x)^4+4(1+x)^2/d+2(1+y) (1+x)^2/d+\mathcal{O} (1/d^2).
\end{align}
Thus the variance is
\begin{equation}
\begin{aligned}
\mathrm{Var}\left[F(U_0\ket{\psi},U\ket{\psi})\right]&=\overline{F^2}-R_f^{\mc{E}}(U_0,U)^2\\
&\le (1+x)^4+4(1+x)^2/d+2(1+y) (1+x)^2/d-(1+x)^4-2(1+x)^2/d+\mathcal{O} (1/d^2)\\
&=(4+2y)(1+x)^2/d+\mathcal{O} (1/d^2)\\
&=\mathcal{O}({1}/{d})
\end{aligned}
\end{equation}
as claimed.
\end{proof}

Note that the leading terms cancel to give a variance scaling like $\frac{1}{d}$, independent of the particular approximate implementation $U$.
Similarly to Corollary~\ref{l2lower}, we can also show that the bound is tight with an additional assumption.
\begin{corollary}\label{Coro:fidelower}
For random inputs from a $d$-dimensional 4-design ensemble $\mc{E}$,
if $R^\mc{E}_f(U_0,U)=1-\varepsilon^2$ with $\varepsilon \gg (\frac{1}{d})^{\frac{1}{4}}$, then the average trace distance satisfies $R_t(U_0,U)=\Theta (\varepsilon)$.
\end{corollary}

\begin{proof}
Here we take $F(U_0\ket{\psi},U\ket{\psi})]$ as a random variable with respect to the random input state $\ket{\psi}$, and abbreviate it as $F(\psi)$. By definition,
\begin{align}
R^\mc{E}_t(U_0,U)
&=\sum_i p_i\sqrt{1-F(\psi_i)}\nonumber \\
&=\sum_{\{i|F(\psi_i)> R_f+a\}} p_i\sqrt{1-F(\psi_i)}+\sum_{\{i|R_f+a\ge F(\psi_i)> R_f-a\}} p_i\sqrt{1-F(\psi_i)}+ \sum_{\{i|F(\psi_i)< R_f-a\}} p_i\sqrt{1-F(\psi_i)}\nonumber \\
&\ge \Pr[R_f+a
\ge F>  R_f-a] \sqrt{1-R_f-a}+
\Pr[ F\le  R_f-a] \sqrt{1-R_f+a}\nonumber \\
&\ge  \left(1-\Pr[|F-R_f|\ge a]\right)\sqrt{1-R_f-a}.
\end{align}

Assuming that
$R_f=1-\varepsilon^2$, and choosing $a=\varepsilon^2/2$, the Chebyshev inequality gives $\Pr[|F-R_f|\ge a] = \mathcal O (\frac{1}{d\varepsilon^4})$ by the variance in Lemma \ref{Le:varF}.
As a result, as $\varepsilon \gg (\frac{1}{d})^{\frac{1}{4}}$,
\begin{equation}
\begin{aligned}
R_t(U_0,U)\ge  [1- o(1)]      \sqrt{\varepsilon^2/2}=\mathcal O(\varepsilon).
\end{aligned}
\end{equation}
On the other hand, the inequality in Eq.~\eqref{Eq:Rtupp} gives
\begin{equation}
\begin{aligned}
R_t(U_0,U)\le  2\varepsilon.
\end{aligned}
\end{equation}
Hence the result follows.
\end{proof}

\subsection{Average measurement error}
In the previous sections, we mainly considered the scenario where the input state is sampled from some ensemble, and characterize the average behavior. If we fix the input state but choose a final measurement at random, we can analyze the performance in a similar way. Here we define the average error with \emph{both} random measurements and random inputs:
\begin{equation}\label{eq:defRM}
\begin{aligned}
R_m(U_0,U)&:=\mathbb{E}_{\psi,O}\left|\tr \left[O\left(U\ket{\psi}\bra{\psi}U^{\dagger}-U_0\ket{\psi}\bra{\psi}U_0^{\dagger}\right)\right]\right|,\\
\end{aligned}
\end{equation}
where the measurement $O$ is chosen from some ensemble. For example, we could consider the projective measurement $O=\ket{\phi}\bra{\phi}$ with $\phi$ sampled from $\mc{E}$. Similarly to previous calculations, by using the Cauchy-Schwartz inequality, we can bound $R_m(U_0,U)$ by
\begin{equation}\label{Eq:RMs}
\begin{aligned}
R_m(U_0,U)&=\mathbb{E}_{\psi,O} \sqrt{  \left(\bra{\psi}U^{\dagger}O U-U_0^{\dagger}O U_0\ket{\psi}\right)^2 }\\
&\le\sqrt{  \mathbb{E}_{\psi,O} \left(\bra{\psi}U^{\dagger}O U-U_0^{\dagger}O U_0\ket{\psi}\right)^2 }=:\sqrt{R_{m,s}(U_0,U)}.
\end{aligned}
\end{equation}
We can evaluate the expectation inside the square root, $R_{m,s}(U_0,U)$, as follows.
\begin{lemma}\label{lemma:RMs}
For random inputs from a $d$-dimensional 2-design ensemble and a random projective measurement $O=\ket{\phi}\bra{\phi}$ with $\phi$ also sampled from a 2-design ensemble,
\begin{equation}
\begin{aligned}
R_{m,s}(U_0,U)=\frac{2}{d(d+1)^2}\|\mathscr{M}\|_F^2-\frac{2}{d^2(d+1)^2}\tr(\mathscr{M})\tr(\mathscr{M}^{\dag}),
\end{aligned}
\end{equation}
where $\mathscr{M}$ is the multiplicative error with $U=U_0(\id+\mathscr{M})$.
\end{lemma}

\begin{proof}
Let $B:=U^{\dagger}O U- U_0^{\dagger}O U_0$ and
compute
\begin{equation}
\begin{aligned}
&\tr\left[ B^{\otimes 2} \mathbb{E}_{\psi} \ket{\psi}\bra{\psi}^{\otimes 2}\right]
=\frac{1}{d(d+1)}\tr\left[B^{\otimes 2}(\id^{\otimes 2}+\s)\right]=\frac{1}{d(d+1)}\left[\tr(B)^2+\tr(B^2)\right]\\
&=\frac{2\tr(O^2)-2\tr(OU_0U^{\dag}OUU_0^{\dag})}{d(d+1)},
\end{aligned}
\end{equation}
where in the last line we use the fact $\tr(B)=0$. With $U=U_0(\id+\mathscr{M})$, the above equation becomes
\begin{equation}
\begin{aligned}
\frac{-2}{d(d+1)}\tr(O^2U_0\mathscr{M}U_0^{\dag}+O^2U_0\mathscr{M}^{\dag}U_0^{\dag}+O U_0\mathscr{M}U_0^{\dag}OU_0\mathscr{M}U_0^{\dag}).
\end{aligned}
\end{equation}
Averaging over the observable $O=\ket{\phi}\bra{\phi}$, we find
\begin{align}
R_{m,s}(U_0,U)&=\frac{-2}{d(d+1)}\mathbb{E}_{\phi}\tr\left(\phi U_0\mathscr{M}U_0^{\dag}+\phi U_0\mathscr{M}^{\dag}U_0^{\dag}+\phi U_0\mathscr{M}^{\dag}U_0^{\dag}\phi U_0\mathscr{M}U_0^{\dag}\right)\nonumber  \\
&=-\frac{2}{d(d+1)}\tr\left(\frac{\id}{d} U_0\mathscr{M}U_0^{\dag}+\frac{\id}{d} U_0\mathscr{M}^{\dag}U_0^{\dag}\right)-\frac{2}{d(d+1)}\tr\left[ \mathbb{E}_{\phi}\phi^{\otimes 2}\left(U_0\mathscr{M}^{\dag}U_0^{\dag}\otimes U_0\mathscr{M}U_0^{\dag}\right)\s \right]\nonumber \\
&=-\frac{2}{d^2(d+1)}\tr(U_0\mathscr{M}U_0^{\dag}+U_0\mathscr{M}^{\dag}U_0^{\dag})-\frac{2}{d^2(d+1)^2}[\tr(U_0\mathscr{M}^{\dag}\mathscr{M}U_0^{\dag})+\tr(U_0\mathscr{M}U_0^{\dag})\tr(U_0\mathscr{M}^{\dag}U_0^{\dag})]\nonumber \\
&=-\frac{2}{d^2(d+1)}\tr(\mathscr{M}+\mathscr{M}^{\dag})-\frac{2}{d^2(d+1)^2}[\tr(\mathscr{M}^{\dag}\mathscr{M})+\tr(\mathscr{M})\tr(\mathscr{M}^{\dag})]\nonumber \\
&=\frac{2}{d(d+1)^2}\tr(\mathscr{M}^{\dag}\mathscr{M})-\frac{2}{d^2(d+1)^2}\tr(\mathscr{M})\tr(\mathscr{M}^{\dag}),
\end{align}
where in the third line we use the 1-design and 2-design property for the first and the second terms, respectively, and the last line uses the unitarity condition $\tr(\mathscr{M}+\mathscr{M}^{\dag})=-\tr(\mathscr{M}^{\dag}\mathscr{M})$.
\end{proof}

By combining Lemma \ref{lemma:RMs} and Eq.~\eqref{Eq:RMs}, we obtain a bound for $R_{m}(U(t),U_0(t))$:

\begin{theorem}\label{th:RMupp}
For random inputs from a $d$-dimensional 2-design ensemble and a random projective measurement $O=\ket{\phi}\bra{\phi}$ with $\phi$ also sampled from a 2-design ensemble, the average error is upper bounded as
\begin{equation}
 R_{m}(U_0,U)\le \left[\frac{2}{d(d+1)^2}\|\mathscr{M}\|^2_F-\frac{2}{d^2(d+1)^2}\tr(\mathscr{M})\tr(\mathscr{M}^{\dag})\right] ^{\frac{1}{2}},
\end{equation}
where $\mathscr{M}$ is the multiplicative error with $U(t)=U_0(\id+\mathscr{M})$.
\end{theorem}

We can also bound $R_{m}(U(t),U_0(t))$ using the result on the average trace norm in Theorem \ref{Th:t}.
\begin{equation}\label{}
\begin{aligned}
\left|\tr \left[O\left(U\ket{\psi}\bra{\psi}U^{\dagger}-U_0\ket{\psi}\bra{\psi}U_0^{\dagger}\right)\right]\right|\leq \|O\| \|U\ket{\psi}\bra{\psi}U^{\dagger}-U_0\ket{\psi}\bra{\psi}U_0^{\dagger}\|_1.
\end{aligned}
\end{equation}
Thus, if $\|O\|\le 1$ for any $O$ in the ensemble (as for a projective measurement), then $R_m(U_0,U)\leq R_t(U_0,U)$. Compared with the upper bound on $R_t(U_0,U)$ in Theorem \ref{Th:t}, Theorem \ref{th:RMupp} here shows a further enhancement, since we also consider the average effect of the measurement, not just the random input. One can also study other kinds of measurement ensembles, such as local measurement or general positive operator valued measurements, and again expect average-case enhancement.

\subsection{Random inputs from subspaces and subsystems}

In the previous sections, we assume that input states are randomly chosen from the entire $n$-qubit Hilbert space $\mc{H}_d$.
We now consider the possibility that they are randomly chosen from a subspace and obey the 1-design condition there. We let $\Pi=\Pi^2$ denote the projection onto this subspace. The average $\ell_2$ error within the subspace is
\begin{equation}
\begin{aligned}
R_{\ell_2}^{\Pi}(U_0,U)
&:=\mathbb{E}_{\psi \in \Pi}\|(U-U_0)\ket{\psi}\|_{\ell_2}\\
&=\mathbb{E}_{\psi \in \Pi} \sqrt{2-\braket{\psi |U^{\dagger}U_0|\psi}-\braket{\psi |U^{\dagger}_0U|\psi}}\\
&\le \left(2-\mathbb{E}_{\psi \in \Pi} \braket{\psi |U^{\dagger}U_0|\psi}+\braket{\psi |U^{\dagger}_0U|\psi} \right)^{\frac{1}{2}} \\
&= \sqrt{\tr \left(\mathscr{M}\mathscr{M}^{\dagger}\frac{\Pi}{d_1}\right)}
=\frac{1}{\sqrt{d_1}} \|\mathscr{M}\Pi\|_F,
\end{aligned}
\end{equation}
where we use $\mathbb{E}_{\psi \in \Pi} (\ket{\psi}\bra{\psi})= \Pi/d_1$.

As a special case, we give a concrete upper bound on the average error for random inputs chosen from a \emph{subsystem}. Specifically, we consider an $n$-qubit input $\ket{\phi}_{[k]}\otimes \ket{\psi}_{[n-k]}$ where $\ket{\phi}_{[k]}$ is an arbitrary input state of the first $k$ qubits and $\ket{\psi}_{[n-k]}$ is a randomly chosen state of the remaining $n-k$ qubits. In this case, the input comes from the subspace $\ket{\phi}_{[k]} \otimes \mc{H}_2^{\otimes (n-k)}$.
For any $\ket{\phi}_{[k]}$, we have the upper bound
\begin{equation}\label{eq:subsysError}
\begin{aligned}
R_{\ell_2}^{\mathrm{sub},k}(U_0,U)&:= \mathbb{E}_{\psi_{[n-k]}}\Big\|(U-U_0) \ket{\phi}_{[k]}\ket{\psi}_{[n-k]}\Big\|_{\ell_2}\\
&\le \left(2-\mathbb{E}_{\psi_{[n-k]}} \bra{\phi}_{[k]} \bra{\psi}_{[n-k]} U^{\dagger}U_0 + U^{\dagger}_0U  \ket{\phi}_{[k]}\ket{\psi}_{[n-k]} \right)^{\frac{1}{2}}\\
&=\sqrt{\tr \left(\mathscr{M}\mathscr{M}^{\dagger}\ket{\phi}\bra{\phi}_{[k]}\otimes \frac{I_{d_1}}{d_1}\right)}\\
&\le \frac{1}{\sqrt{d_1}} \|\mathscr{M}\|_F,
\end{aligned}
\end{equation}
where $d_1=2^{n-k}$ is the dimension of an $(n-k)$-qubit state. The last inequality follows by observing that $\mathscr{M}\mathscr{M}^{\dagger}$ is positive semidefinite, and that $\ket{\phi}\bra{\phi}_{[k]}\otimes \frac{I_{d_1}}{d_1}<\frac{I_{d}}{d_1}$ for any state $\ket{\phi}_{[k]}$. We remark that this upper bound is not tight and one may enhance it by using knowledge of $\ket{\phi}_{[k]}$. However, the bound is reasonably tight if $k$ is small, as will be the case where we bound the error of out-of-time-order correlators in Appendix~\ref{sec:OTOC}.

\subsection{Multi-segment performance via the triangle inequality}\label{Sec:one-multi}

For worst-case analysis, the error $\|\mathscr{M}\|$ can first be estimated within a small time evolution segment, and then the total error in a long-time evolution can be bounded by the triangle inequality as $\|\mathscr{M}(t)\|\le r \|\mathscr{M}(t/r)\|$ (i.e., the error is subadditive). However, for average-case error as measured by the Frobenius norm $\|\mathscr{M}\|_F$, the inequality $\|\mathscr{M}(t)\|_F\le r \|\mathscr{M}(t/r)\|_F$ may not hold.
Instead, in this section we first show that the average error with Haar-random inputs $R(U(t),U_0(t))$ is subadditive, due to the left-invariance property of the Haar measure. For the $\ell_2$ norm, we further show that the upper bound $\tilde{R}_{\ell_2}(U(t),U_0(t))$ in Eq.~\eqref{eq:l2upp1} is subadditive under the weaker condition that the input state comes from a 1-design ensemble. This allows us to quantify the total average error by studying the average error in one simulation segment, which is useful for our later analysis.

\begin{lemma}\label{Lemma:segment1}
For random inputs from the Haar ensemble, for any $t \in \mathbb{R}$ and $r \in \mathbb{N}$, the average error satisfies
\begin{equation}
R(U^r(t/r),U_0^r(t/r))\le r  R(U(t/r),U_0(t/r)),
\end{equation}
where the average error measure $R$ can be $R_{\ell_2}$ or $R_{t}$ corresponding to $\ell_2$ norm or trace norm, respectively.
\end{lemma}

\begin{proof}
For simplicity, in this proof we denote $U:=U(t/r)$ and $U_0:=U_0(t/r)$,
so that $U$ is the unitary evolution actually implemented in a segment and $U_0$ is the corresponding ideal evolution.
Notice that the trace norm and $\ell_2$ norm both induce metrics satisfying the triangle inequality, so
\begin{equation}
\begin{aligned}
R(U^r(t/r),U_0^r(t/r))&= \mathbb{E}_{\psi}[\mc{D}(U^r\ket{\psi},U_0^r\ket{\psi})]\\
&=\int_{\psi} \mc{D}(U^r\ket{\psi},U_0^r\ket{\psi}) \d\psi\\
&\le \sum_{k=0}^{r-1} \int_{\psi} \mc{D}(U^{r-k}U_0^{k}\ket{\psi},U^{r-k-1}U_0^{k+1}\ket{\psi}) \d\psi\\
&=r\int_{\psi} \mc{D}(U\ket{\psi},U_0\ket{\psi}) \d\psi\\
&=rR(U(t/r),U_0(t/r)).
\end{aligned}
\end{equation}
Here $\mc{D}$ denotes the relevant distance, coming from either the $\ell_2$ norm or the trace norm.
The second-to-last equality is due to the invariance of the Haar measure.
\end{proof}

For a discrete ensemble $\mathcal{E}$, the average error $R^{\mathcal{E}}$ cannot have the corresponding left-invariance property. However, we can use this property of its upper bound $\tilde{R}_{\ell_2}$ in Eq.~\eqref{eq:l2upp1}.
Now we show how to carry out this approach over multiple segments for $R^{\mathcal{E}}_{\ell_2}$.

\begin{lemma}\label{Lemma:segment2}
For random inputs from a 1-design ensemble $\mathcal{E}$,
for any $t \in \mathbb{R}$ and $r \in \mathbb{N}$,
\begin{equation}
R^{\mathcal{E}}_{\ell_2}(U^r(t/r),U_0^r(t/r))\le r  \tilde{R}_{\ell_2}(U(t/r),U_0(t/r)),
\end{equation}
where $\tilde{R}_{\ell_2}$ is defined in Eq.~\eqref{eq:l2upp}.
\end{lemma}

\begin{proof}
For simplicity, we again denote $U:=U(t/r)$ and $U_0:=U_0(t/r)$.
Then we have
\begin{equation}
\begin{aligned}
R_{\ell_2}(U^r,U_0^r)&=\mathbb{E}_{\psi\in \mc{E}} \|(U^r-U_0^r)\ket{\psi}\|_{\ell_2}\\
&\le \mathbb{E}_{\psi\in \mc{E}} \sum_{k=0}^{r-1} \|(U^{r-k}U_0^{k}-U^{r-k-1}U_0^{k+1})\ket{\psi}\|_{\ell_2}\\
&= \mathbb{E}_{\psi\in \mc{E}} \sum_{k=0}^{r-1} \|(UU_0^{k}-U_0^{k+1})\ket{\psi}\|_{\ell_2}\\
&= \sum_{k=0}^{r-1} R^{\mathcal{E}}_{\ell_2}(UU_0^k,U_0^{k+1})\\
&\le \sum_{k=0}^{r-1}
\tilde{R}_{\ell_2}(UU_0^k,U_0^{k+1})\\
&=r\tilde{R}_{\ell_2}(U,U_0).
\end{aligned}
\end{equation}
Here the first inequality is due to the triangle inequality and the second inequality is by Eq.~\eqref{eq:l2upp1}. In the last equality, for any $k$, $\tilde{R}_{\ell_2}(UU_0^k,U_0^{k+1})= \tilde{R}_{\ell_2}(U,U_0)$. This is
because $\tilde{R}_{\ell_2}$ in Eq.~\eqref{eq:l2upp} only involves a linear function of $\psi$, and thus the 1-design property is sufficient to apply invariance of the measure.
\end{proof}

Furthermore, for approximate 1-design inputs as defined in Eq.~\eqref{eq:AppDesign}, we have
\begin{align}
R^{\mathcal{E}}_{\ell_2}(U^r,U_0^r)
&\le \sum_{k=0}^{r-1} \left(2-\mathbb{E}_{\psi \in \mathcal{E}} \braket{\psi |(U^{\dagger}_0)^kU^{\dagger}U_0^{k+1}|\psi}+\braket{\psi |(U^{\dagger}_0)^{k+1}UU_0^k|\psi} \right)^{\frac{1}{2}}\nonumber \\
&= \sum_{k=0}^{r-1} \left(2-\mathbb{E}_{\psi \in \mathrm{Haar}} \braket{\psi |(U^{\dagger}_0)^kU^{\dagger}U_0^{k+1}|\psi}+\braket{\psi |(U^{\dagger}_0)^{k+1}UU_0^k|\psi}+  \tr\{[(U^{\dagger}_0)^kU^{\dagger}U_0^{k+1}+(U^{\dagger}_0)^{k+1}UU_0^k]C\}
\right)^{\frac{1}{2}}\nonumber \\
&\le r\left(2-\int_{\psi}\braket{\psi |U^{\dagger}U_0|\psi}+\braket{\psi |U^{\dagger}_0U|\psi}\d\psi +2d \|C\| \right)^{\frac{1}{2}}\nonumber \\
&\le  r\left(\frac{1}{d} \|\mathscr{M}(t/r)\|_F^2+2\epsilon \right)^{\frac{1}{2}},
\end{align}
where $C:=\sum_{\phi_i \in\mathcal{E}} p_i\ket{\phi_i}\bra{\phi_i}-\int_{\mathrm{Haar}} \ket{\psi}\bra{\psi} \d \psi$, $\|V\|=1$.
The second inequality holds because
\begin{equation}
\begin{aligned}
\tr((V+V^{\dag})C)\leq 2d\|V\|\|C\|=2d\|C\|,
\end{aligned}
\end{equation}
with $V=(U^{\dagger}_0)^kU^{\dagger}U_0^{k+1}$. Note that $C$ is Hermitian and $-\epsilon \id/d \leq  C \leq \epsilon \id/d$ according to the definition in Eq.~\eqref{eq:AppDesign}. The third inequality is due to $2d\|C\|\le 2\epsilon$.

Finally, we show the triangle inequality for inputs drawn from a 1-design ensemble of a subsystem.
\begin{lemma}\label{Lemma:subsystem}
For an input state $\ket{\phi}_{[k]}\otimes \ket{\psi}_{[n-k]}$ in an n-qubit Hilbert space, where $\ket{\phi}_{[k]}$ is a fixed state of $m$ qubits and $\ket{\psi}_{[n-k]}$
is drawn randomly from a 1-design  ensemble $\mathcal{E}$ on the remaining $n-k$ qubits, the  average  error  between  the  ideal evolution $U_0(t) =U_0^r(t/r)$ and its approximation $U(t) =U^r(t/r)$ is upper bounded as
\begin{equation}
R_{\ell_2}^{\mathrm{sub},k}(U(t),U_0(t))\le \frac{r}{\sqrt{d_1}} \|\mathscr{M}(t/r)\|_F,
\end{equation}
where $d_1:=2^{n-m}$.
\end{lemma}

\begin{proof}
We have
\begin{equation}
\begin{aligned}
R_{\ell_2}^{\mathrm{sub},k}(U(t),U_0(t))&=\mathbb{E}_{\psi_{[n-k]}} \left\|(U^{r}-U_0^{r})\ket{\phi}_{[k]}\ket{\psi}_{[n-k]}\right\|_{\ell_2}\\
&= \mathbb{E}_{\psi_{[n-k]}}  \left\|\sum_{k=0}^{r-1} (U^{r-k}U_0^k-U^{r-(k+1)}U_0^{k+1})\ket{\phi}_{[k]}\ket{\psi}_{[n-k]}\right\|_{\ell_2}\\
&\le \mathbb{E}_{\psi_{[n-k]}} \sum_{k=0}^{r-1} \left\|(U^{r-k}U_0^k-U^{r-(k+1)}U_0^{k+1})\ket{\phi}_{[k]}\ket{\psi}_{[n-k]}\right\|_{\ell_2}\\
&= \mathbb{E}_{\psi_{[n-k]}} \sum_{k=0}^{r-1} \left\|U^{r-(k+1)}(U-U_0)\  U_0^{k}\ket{\phi}_{[k]}\ket{\psi}_{[n-k]}\right\|_{\ell_2}\\
&\le \sum_{k=0}^{r-1} \sqrt{\tr[ (\mathscr{M}(t/r)\mathscr{M}^{\dagger}(t/r)\ U_0^k\ket{\phi_{m}}\bra{\phi_{m}}\otimes \frac{I_{d_1}}{d_1}(U_0^k)^{\dagger}]}\\
&\le \sum_{k=0}^{r-1} \sqrt{\tr [\mathscr{M}(t/r)\mathscr{M}^{\dagger}(t/r)\ U_0^k I_{2^m} \otimes \frac{I_{d_1}}{d_1}(U_0^k)^{\dagger}]}\\
&\le \frac{r}{\sqrt{d_1}} \|\mathscr{M}(t/r)\|_F
\end{aligned}
\end{equation}
as claimed.
\end{proof}

\section{General theory of average error in Hamiltonian simulation algorithms}

In this section, we give a general theory of average error for the $p$th-order product formula method (PF$p$) and the Taylor series method.

\subsection{Product formula method}\label{Sec:generalp}

The first-order product formula is defined as
\begin{equation}
\mathscr{U}_1(t)=e^{-iH_1t}e^{-iH_2t}\cdots e^{-iH_Lt}
=\overrightarrow{\prod_l}e^{-iH_lt}.
\end{equation}
Here the arrow denotes the ordering of the product, i.e., the direction in which indices increase.
The $2k$th-order product formulas are defined
recursively by
\begin{equation}
\begin{aligned}\label{Eq:highorder}
\mathscr{U}_2(t)&=\overrightarrow{\prod_l}e^{-iH_lt/2} \overleftarrow{\prod_l}e^{-iH_lt/2},\\
\mathscr{U}_{2k}(t)&= [\mathscr{U}_{2k-2}(p_k t)]^2 \mathscr{U}_{2k-2}((1-4p_k)t) [\mathscr{U}_{2k-2}(p_k t)]^2,
\end{aligned}
\end{equation}
where $p_k:=\frac{1}{4-4^{1/(2k-1)}}$ and $k>1$ \cite{suzuki1991general}.
Overall, we have $S=2\cdot 5^{k-1}$ stages (operators of the form $\mathscr{U}_1$ or its reverse ordering) and the evolution can be rewritten as
\begin{equation}\label{eq:pf2k}
  \mathscr{U}_{2k}(t)=  \prod_{s=1}^{S} \prod_{l=1}^{L} e^{-it a_s H_{\pi_s(l)}}.
\end{equation}
Here $\pi_s$ is the identity permutation or reversal permutation and $t a_s$ denotes the simulation time in different stages.

To evolve for a long time $t$, we divide the evolution into $r$ steps and apply the above product formulas $r$ times, giving $\mathscr{U}_{2k}^r(t/r)$.
We give the following bound on the average-case performance of PF$p$ with $p=2k$.

\begin{theorem}\label{Th:general}
For the PF$p$ simulation $\mathscr{U}_{p}^r(t/r)$ specified by \eqref{eq:pf2k},
the average error in the $\ell_2$ norm for a $d$-dimensional 1-design input ensemble has the asymptotic upper bound
$R_{\ell_2}= \mathcal O\left(T_p{t^{p+1}}/{r^{p}}\right)$
where
\begin{equation}\label{eq:Tp}
T_p:=\sum_{l_1,\dots,l_{p+1}=1}^L
\frac{1}{\sqrt{d}} \left\|[H_{l_1},[H_{l_2},\dots,[H_{l_p},H_{l_{p+1}}]]] \right\|_F.
\end{equation}
Therefore $r=\mathcal O\bigl(T_p^{\frac{1}{p}}t^{1+\frac{1}{p}}\varepsilon^{-\frac{1}{p}}\bigr)$ segments suffice to ensure average error at most $\varepsilon$.
\end{theorem}

\begin{proof}
Here we define the order of tuples $(s,l)$ as
$(s,l)\prec (s',l')$ when $s<s'$ or $s=s',l<l'$ and $(s,l)\preceq (s',l')$ when $s<s'$ or $s=s',l\le l'$.
According to Theorem 3 in Ref.~\cite{childs2020theory}, the multiplicative error $\mathscr{M}(t)$
can be expressed as
\begin{equation}
\mathscr{M}(t)=e^{iHt} \int_0^t \d\tau_1 e^{-i(t-\tau_1)H} \mathscr{U}_p(\tau_1) \mathscr{N}(\tau_1)=\int_0^t \d\tau_1 e^{i\tau_1H} \mathscr{U}_p(\tau_1) \mathscr{N}(\tau_1),
\end{equation}
where $\mathscr{U}_p$ is the $p$th-order Trotter formula of Eq.~\eqref{eq:pf2k} and
\begin{equation}
\begin{aligned}
 \mathscr{N}(\tau_1)= &\sum_{(s,l)}
 \overrightarrow{\prod}_{(s',l')\prec(s,l)}
~ e^{\tau_1 a_s H_{\pi_{s'}(l')}} \left(a_s H_{\pi_s(l)}\right)       \overleftarrow{\prod}_{(s',l')\prec(s,l)}
~ e^{-\tau_1 a_s H_{\pi_{s'}(l')}}  \\
&-  \overrightarrow{\prod}_{(s',l')}
~ e^{\tau_1 a_s H_{\pi_{s'}(l')}} H \overleftarrow{\prod}_{(s',l')}
~ e^{-\tau_1 a_s H_{\pi_{s'}(l')}},
\end{aligned}
\end{equation}
where $\mathscr{N}(\tau_1)=\mathcal O(\tau^p_1)$.
Here we define the vector $\vec{j}_{p+1}=(j_1,j_2,\dots,j_{p+1})$ with $p+1$ entries, $j_1,j_2,\dots,j_{p+1} \in \{(s,l): s\in \{1,\dots,S\},l\in\{1,\dots,L\}\}$ and the corresponding nested commutators as
\begin{equation}
 N_{\vec{j}_{p+1}}=[H_{j_1},[H_{j_2},\dots,[H_{j_p},H_{j_{p+1}}]]].
\end{equation}
According to Theorem 5 in Ref.~\cite{childs2020theory}, we rewrite
 $\mathscr{N}(\tau_1)$ and $\mathscr{M}(t)$ as
\begin{equation}
\begin{aligned}
 \mathscr{N}(\tau_1)&= \sum_{i=1,2} \sum_{\vec{j}_{p+1}\in \Gamma_i} \int_0^{\tau_1}\d\tau_2       (\tau_1-\tau_2)^{q(\vec{j}_{p+1})-1}\tau_1^{p-q(\vec{j}_{p+1})} c_{\vec{j}_{p+1}} F_{\vec{j}_{p+1}}^{\dagger} N_{\vec{j}_{p+1}} F_{\vec{j}_{p+1}};\\
 \mathscr{M}(t)&=\int_0^t \d\tau_1 \int_0^{\tau_1}\d\tau_2
 \sum_{i=1,2} \sum_{\vec{j}_{p+1}\in \Gamma_i} (\tau_1-\tau_2)^{q(\vec{j}_{p+1})-1}\tau_1^{p-q(\vec{j}_{p+1})} c_{\vec{j}_{p+1}} E_{\vec{j}_{p+1}} N_{\vec{j}_{p+1}} F_{\vec{j}_{p+1}}.
 \end{aligned}
\end{equation}
Here $\Gamma_1$ and $\Gamma_2$ correspond to the first and second part in  $\mathscr{N}(\tau_1)$, respectively:
\begin{equation}
    \begin{aligned}
\Gamma_1&:=\{(j_1,j_2,\dots,j_{p+1}) : j_1\preceq j_2\preceq\dots \preceq j_{p+1}\},    \\
\Gamma_2&:=\{(j_1,j_2,\dots,j_{p+1}) : j_1\preceq j_2\dots \preceq j_{p}, \, j_{p+1}=(1,l_{p+1})\} .
    \end{aligned}
\end{equation}
The $c_{\vec{j}_{p+1}}$ are real coefficients that are functions of $\vec{j}_{p+1}$ and $p$, satisfying $|c_{\vec{j}_{p+1}}|\le 1$. The function
$q(\vec{j}_{p+1})$
is the maximal number $q$ satisfying $j_1=j_2=\dots=j_q$. The $F_{\vec{j}_{p+1}}$ are unitary, constructed as products of terms of the form $e^{-iH_l\tau_1}$, and $E_{\vec{j}_{p+1}}:= e^{-i(t-\tau_1)H} \mathscr{U}_p(\tau_1)F_{\vec{j}_{p+1}}^{\dagger}$.

Consequently, we find
\begin{align}
& \mathscr{M}(t)\mathscr{M}^{\dag}(t)\nonumber \\
&= \int_0^t \d\tau_1 \int_0^{\tau_1}\d\tau_2
 \int_0^t \d\tau_1' \int_0^{\tau_1'}\d\tau_2'
 \sum_{i,i'=1,2} \sum_{\vec{j}_{p+1}\in \Gamma_i,~\vec{j}_{p+1}'\in \Gamma_{i'}} (\tau_1-\tau_2)^{q(\vec{j}_{p+1})-1}\tau_1^{p-q(\vec{j}_{p+1})}
 (\tau_1'-\tau_2')^{q(\vec{j}'_{p+1})-1}(\tau_1')^{p-q(\vec{j}_{p+1})}\nonumber \\
&\qquad c_{\vec{j}_{p+1}}c_{\vec{j}_{p+1}'} E_{\vec{j}_{p+1}} N_{\vec{j}_{p+1}} F_{\vec{j}_{p+1}}
 F_{\vec{j}_{p+1}'}^{\dagger} N_{\vec{j}_{p+1}'}^{\dagger}E_{\vec{j}_{p+1}'}^{\dagger}.
 \end{align}
 Because the $E$ and $F$ operators are all unitary, using Lemma~\ref{Lemma:traceproduct} below, we have
 \begin{equation}\label{Eq:ENF}
    \left |\tr( E_{\vec{j}_{p+1}}  N_{\vec{j}_{p+1}} F_{\vec{j}_{p+1}}F_{\vec{j}_{p+1}'}^{\dagger} N_{\vec{j}_{p+1}'}^{\dagger} E^{\dagger}_{\vec{j}_{p+1}'})\right|\le \sqrt{\tr(N_{\vec{j}_{p+1}}N^{\dagger}_{\vec{j}_{p+1}})}\sqrt{\tr(N_{\vec{j}'_{p+1}}N^{\dagger}_{\vec{j}'_{p+1}})}.
 \end{equation}
We therefore have the upper bound
\begin{align}\label{Eq:mainENF}
&\tr(\mathscr{M}(t)\mathscr{M}^{\dag}(t))\nonumber\\
&\le \int_0^t \d\tau_1 \int_0^{\tau_1}\d\tau_2
 \int_0^t \d\tau_1' \int_0^{\tau_1'}\d\tau_2'
 \sum_{i,i'=1,2} \sum_{\vec{j}_{p+1}\in \Gamma_i,~\vec{j}_{p+1}'\in \Gamma_{i'}} (\tau_1-\tau_2)^{q(\vec{j}_{p+1})-1}\tau_1^{p-q(\vec{j}_{p+1})}
 (\tau_1'-\tau_2')^{q(\vec{j}'_{p+1})-1}(\tau_1')^{p-q(\vec{j}_{p+1})}\nonumber\\
 &\qquad c_{\vec{j}_{p+1}}c_{\vec{j}_{p+1}'} \sqrt{\tr(N_{\vec{j}_{p+1}}N^{\dagger}_{\vec{j}_{p+1}})}\sqrt{\tr(N'_{\vec{j}_{p+1}}N'^{\dagger}_{\vec{j}_{p+1}})} \nonumber\\
 &\le t^{2p+2}   \sum_{i,i'=1,2} \sum_{\vec{j}_{p+1}\in \Gamma_i,~\vec{j}_{p+1}'\in \Gamma_{i'}} \sqrt{\tr(N_{\vec{j}_{p+1}}N^{\dagger}_{\vec{j}_{p+1}})}\sqrt{\tr(N'_{\vec{j}_{p+1}}N'^{\dagger}_{\vec{j}_{p+1}})}\nonumber\\
 &\le  4t^{2p+2} \sum_{\vec{j}_{p+1},\vec{j}_{p+1}'} \sqrt{\tr(N_{\vec{j}_{p+1}}N^{\dagger}_{\vec{j}_{p+1}})}\sqrt{\tr(N'_{\vec{j}_{p+1}}N'^{\dagger}_{\vec{j}_{p+1}})}\nonumber\\
&= 4t^{2p+2}   \left[\sum_{\vec{j}_{p+1}}\sqrt{\tr(N_{\vec{j}_{p+1}}N^{\dagger}_{\vec{j}_{p+1}})}\right]^2\nonumber \\
&= 4t^{2p+2}S^{2p+2}
\left[\sum_{l_1,\dots,l_{p+1}=1}^L\sqrt{\tr(|[H_{l_1},[H_{l_2},\dots,[H_{l_p},H_{l_{p+1}}]]]|^2)}\right]^2,
\end{align}
where the last equation follows because each $H_l$ could appear in $S$ different stages, so there are $S^{p+1}$ possibilities in total for each $[H_{l_1},[H_{l_2},\dots,[H_{l_p},H_{l_{p+1}}]]]$.

For long-time evolution, we divide the evolution time $t$ into $r$ segments. Our assumption that $p$ is a constant implies $S^{2p+2}=\mathcal O(1)$, so for each segment,
\begin{equation}
 R_{\ell_2}(t/r)\le \sqrt{\frac{\tr\left(\mathscr{M}(t/r)\mathscr{M}(t/r)^{\dag}\right)}{d}}
 =\mathcal O\left(T_p\frac{t^{p+1}}{r^{p+1}}\right).
\end{equation}
Lemma~\ref{Lemma:segment1} then implies $R_{\ell_2}(t)= \mathcal O\bigl(T_p\frac{t^{p+1}}{r^{p}}\bigr)$,
as claimed.
\end{proof}

Finally, we complete the proof of Lemma~\ref{Lemma:traceproduct} which can be used to show the inequality Eq.~\eqref{Eq:ENF}
in the above analysis.
\begin{lemma}\label{Lemma:traceproduct}
Let $M$ and $N$ be complex matrices and $U_1$ and $U_2$ be unitaries, all of which are in $\mathbb{C}^{d\times d}$.
Then
\begin{equation}
|\tr(MN)|\leq \sqrt{\tr(MM^{\dag})}\sqrt{\tr(NN^{\dag})}
\quad\text{and}\quad
|\tr(U_1MU_2N)|\leq \sqrt{\tr(MM^{\dag})}\sqrt{\tr(NN^{\dag})}.
\end{equation}
\end{lemma}

\begin{proof}
The inequality $|\tr(MN)|\leq \sqrt{\tr(MM^{\dag})}\sqrt{\tr(NN^{\dag})}$ follows from the Cauchy-Schwarz inequality.
Using this inequality, we also have
\begin{equation}
   |\tr(U_1MU_2N)|\leq \sqrt{\tr(U_1MM^{\dag}U_1^{\dag})}\sqrt{\tr(NN^{\dag})}=\sqrt{\tr(MM^{\dag})}\sqrt{\tr(NN^{\dag})},
\end{equation}
because $U_1$ and $U_2$ are unitary.
\end{proof}
To clarify the improvement,
we compare our result to the analogous result regarding worst-case error. According to \cite{childs2020theory}, the worst-case error $W(U_0(t),\mathscr{U}_p^r(t/r)):= \|U_0(t)-\mathscr{U}_p^r(t/r)\|$ is
\begin{equation}
W(U_0(t),\mathscr{U}_p^r(t/r))=\mathcal O\left(\alpha_{\mathrm{comm}, p}t^{p+1}/r^{p}\right),~
\alpha_{\mathrm{comm},p}=\sum_{l_1,\dots,l_{p+1}=1}^L \left\|[H_{l_1},[H_{l_2},\dots,[H_{l_p},H_{l_{p+1}}]]]\right\|.
\end{equation}
Note that $\alpha_{\mathrm{comm},p}$ can, in principle, be much larger than $T_p$. In Appendix~\ref{Sec:app}, we calculate $T_p$ and compare it with $\alpha_{\mathrm{comm},p}$ for different types of Hamiltonians.

\subsection{Taylor series method}

In the truncated Taylor series method \cite{TaylorSeries}, we consider Hamiltonians $H$ that can be expressed as a sum of unitaries $H_i$ with positive coefficients $\alpha_i$, i.e., $H=\sum_{i=1}^L \alpha_i H_i$.
We first divide the evolution time $t$ into $r$ segments, where
$r =\frac{\alpha t}{\ln2}$
with $\alpha: =\sum_{i=1}^L \alpha_i$. In each segment, we truncate the Taylor expansion of $e^{-iH\tau}$ with $\tau=t/r$ to $K$th order, giving
\begin{equation}
 \widetilde{U}(\tau) :=  \sum_{j=0}^{K}\frac{(-i\tau H)^j}{j!}=U_0-E,
\end{equation}
where the truncation error is $ E:=\sum_{j=K+1}^{\infty}\frac{(-i\tau H)^j}{j!}$. We have $\|E\|\le \delta$ where $\delta=2\frac{(\alpha \tau)^{K+1}}{(K+1)!}=2\frac{(\ln 2)^{K+1}}{(K+1)!}$. A step of postselected oblivious amplitude amplification \cite{TaylorSeries} implements the evolution
\begin{equation}
\begin{aligned}
\widetilde{V}(\tau):=\frac{3}{2}\widetilde{U}(\tau)-\frac{1}{2}\widetilde{U}(\tau)\widetilde{U}(\tau)^\dag \widetilde{U}(\tau).\\
\end{aligned}
\end{equation}
We rewrite this operator as $\widetilde{V}(\tau)=U_0(\tau)(I+ \mathscr{M})$, where the multiplicative error is
\begin{equation}
\begin{aligned}
\mathscr{M}&=-\frac{3}{2}U_0^\dag E + \frac{1}{2}U_0^\dag
(\widetilde{U}\widetilde{U}^\dag \widetilde{U}- U_0)\\
&=-\frac{3}{2}U_0^\dag E - \frac{1}{2}U_0^\dag (2E+U_0EU_0^\dag - EE^\dag U_0-EU_0^\dag E-U_0E^\dag E+ EE^\dag E)\\
&=-\frac{5}{2}U_0^\dag E - \frac{1}{2}E U_0^\dag +R_{2}.
\end{aligned}
\end{equation}
Here $R_2$ captures the higher-order terms of $E$ such that $\|R_{2}\|=\mathcal O(\delta^2)$. We leave the evolution time $\tau$ implicit for simplicity.

According to the definition of Frobenius norm, we have
\begin{equation}
\begin{aligned}
R_{\ell_2}\le \|\mathscr{M}\|_F^2/d=\tr(\mathscr{M}\mathscr{M}^{\dagger})/d\le 9\tr(EE^\dag)/d+6\|R_{2}\| \|E\|+\|R_2\|^2.
\end{aligned}
\end{equation}
The leading term can be upper bounded using
\begin{equation}
\begin{aligned}
\tr(EE^\dag)&=\sum_{j,j'=K+1}^{\infty}\frac{(-1)^{j+j'} (i\tau )^{j+j'}}{j!j'!} \tr(H^{j+j'})\\
&\le d \max_i (\alpha_i/\alpha)  \sum_{j,j'=K+1}^{\infty}\frac{(\alpha\tau )^{j+j'}}{j!j'!} \\
&\le d \max_i (\alpha_i/\alpha) \, \delta^2,
\end{aligned}
\end{equation}
where the first inequality is due to Lemma~\ref{Lemma:Taylor-bound} below.
Consequently, we have $\|\mathscr{M}\|_F/\sqrt{d}=\mathcal O(\delta \max_i \sqrt{\alpha_i/\alpha})$
and the total average error is upper bounded by $R_{\ell_2}\le  r \|\mathscr{M}\|_F/\sqrt{d}= \mathcal O(r \max_i \sqrt{\alpha_i/\alpha} ~ \delta)= \mathcal O(\max_i \sqrt{\alpha_i \alpha}t \frac{(\ln 2)^{K+1}}{(K+1)!}) $. In comparison, the worst-case upper bound is
$\mathcal O(\alpha t \frac{(\ln 2)^{K+1}}{(K+1)!}) $.
If all the coefficients $\alpha_i$ are similar, then the average-case error is an improvement over the worst-case error by a factor of $\mathcal O(\sqrt{L})$.

\begin{lemma}\label{Lemma:Taylor-bound}
Consider a $d$-dimensional Hamiltonian $H=\sum_{i=1}^L \alpha_i P_i$ where the $P_i$ are Pauli operators with positive coefficients $\alpha_i>0$, and $\alpha:= \sum_{i=1}^L \alpha_i$. Then for any positive integer $j$,
\begin{equation}
\tr(H^j)\le d\alpha^{j-1} \max_i \alpha_i.
\end{equation}
\end{lemma}

\begin{proof}
We consider the $(j-1)$th power of $H$,
\begin{equation}
H^{j-1} =  \sum_{i_1,\dots,i_{j-1}=1}^L \alpha_{i_1} \dots \alpha_{i_{j-1}}P_{i_1} \dots P_{i_{j-1}},
\end{equation}
which has at most $L^{j-1}$ terms. For each term $P_{i_1} \dots P_{i_{j-1}}$ in $H^{j-1}$, there is at most one Pauli operator $P$ in $H$ satisfying $P_{i_1} \dots P_{i_{j-1}}=P$, which implies $\tr(P_{i_1} \dots P_{i_{j-1}} P)=d  \neq 0$. Thus the coefficients of nonzero terms in $H^j$ can be bounded by \begin{equation}
\tr(H^j)\le   d\sum_{i_1,\dots,i_{j-1}=1}^L \alpha_{i_1} \dots \alpha_{i_{j-1}}   \max_i \alpha_i =d \alpha^{j-1}    \max_i \alpha_i
\end{equation}
as claimed.
\end{proof}

The above analysis shows that the error is reduced from
$\mathcal O(\alpha t \frac{(\ln 2)^{K+1}}{(K+1)!})$ in the worst case to
$\mathcal O(\max_i \sqrt{\alpha_i \alpha}t \frac{(\ln 2)^{K+1}}{(K+1)!})$ on average. However, this improvement does not significantly influence the gate complexity of the LCU method, since this complexity is logarithmic in $1/\varepsilon$. Therefore we mainly explore improvements in product formula methods, which have gate complexities with inferior asymptotic error scaling but may nevertheless sometimes perform better in practice.

\section{Average error in first- and second-order product formula algorithms}\label{Sec:PF12}

In this section,
we tighten the error bounds for PF1 and PF2 and calculate their prefactors.
We first consider evolving Hamiltonians with two terms, $H=A+B$, for one time step (Lemma \ref{Lemma:AB}). We then generalize this result to Hamiltonians with multiple terms, $H=\sum_l^L H_l$ (Lemma \ref{Lemma:multiterm}). Finally, we apply these results to derive bounds for evolution with multiple time steps (Theorems \ref{Th:PF1} and \ref{Th:PF2}). We use $|X|^2$ to denote $XX^{\dagger}$ for a complex operator $X$.
\begin{lemma}\label{Lemma:AB}
(First-order product formula, PF1)
For a two-term Hamiltonian $H=A+B$, the first-order product formula $\mathscr{U}_1(t)=e^{-iBt}e^{-iAt}$ satisfies $\mathscr{U}_1(t)=U_0(t)(I+\mathscr{M}(t))$, where
\begin{equation}
\|\mathscr{M}(t)\|_F^2\le \frac{t^4}{4} \tr(|[A,B]|^2),
\end{equation}
and the average $\ell_2$ norm distance assuming a $d$-dimensional 1-design input ensemble $R_{\ell_2}(\mathscr{U}_1(t),U_0(t))$ satisfies
\begin{align}\label{eq:ABtrbound}
R_{\ell_2}(\mathscr{U}_1(t),U_0(t))&\le \frac{t^2}{2} \left(\frac{\tr(|[A,B]|^2)}{d}\right)^{\frac{1}{2}}.
\end{align}
\end{lemma}
\begin{proof}
From \cite{childs2020theory}, we have
\begin{equation}
\begin{aligned}\label{Mhigh}
\mathscr{M}(t)&=\int_0^t \d\tau_1  e^{iH\tau_1}e^{-iB\tau_1}\left(e^{iB\tau_1}iAe^{-iB\tau_1}    -iA\right)e^{-iA\tau_1}\\
&=\int_0^t \d\tau_1  \int_0^{\tau_1} \d\tau_2~e^{iH\tau_1}e^{-iB(\tau_1-\tau_2)}
[iB,iA] e^{-iB\tau_2}e^{-iA\tau_1}.
\end{aligned}
\end{equation}
Therefore
\begin{align}
&\tr(\mathscr{M}(t)\mathscr{M}^{\dag}(t))\nonumber \\
&=\tr\left\{\int_0^t \d\tau_1 \int_0^{\tau_1}  \d\tau_2
\int_0^t \d\tau_1' \int_0^{\tau_1'}  \d\tau_2'
~e^{iH\tau_1}e^{-iB(\tau_1-\tau_2)}[iB,iA] e^{-iB\tau_2}e^{-iA\tau_1}
e^{iA\tau_1'} e^{iB\tau_2'}[iB,iA]^{\dagger}e^{iB(\tau_1'-\tau_2')}
e^{-iH\tau_1'}\right\}\nonumber \\
&=\int_0^t \d\tau_1 \int_0^{\tau_1}  \d\tau_2
\int_0^t \d\tau_1' \int_0^{\tau_1'}  \d\tau_2'
~\tr\left(e^{iH\tau_1}e^{-iB(\tau_1-\tau_2)}[iB,iA] e^{-iB\tau_2}e^{-iA\tau_1}
e^{iA\tau_1'} e^{iB\tau_2'}[iB,iA]^{\dagger}e^{iB(\tau_1'-\tau_2')}
e^{-iH\tau_1'}\right)\nonumber\\
&=\int_0^t \d\tau_1 \int_0^{\tau_1}  \d\tau_2
\int_0^t \d\tau_1' \int_0^{\tau_1'}  \d\tau_2'
~\tr\left(P[iB,iA]Q[iB,iA]^{\dagger}
\right),
\end{align}
where $P=e^{iB(\tau_1'-\tau_2')}e^{-iH\tau_1'}e^{iH\tau_1}e^{-iB(\tau_1-\tau_2)}$ and $Q=e^{-iB\tau_2}e^{-iA\tau_1}
e^{iA\tau_1'} e^{iB\tau_2'}$.
By taking $P[iB,iA]Q=M$ and $[iB,iA]^{\dagger}=N$ as in Lemma \ref{Lemma:traceproduct}, and using the fact that $P$, $Q$ are unitary,
we have
\begin{equation}
\begin{aligned}
 \left|\tr\left(P[iB,iA]Q[iB,iA]^{\dagger}
\right)\right|&\le   \sqrt{\tr(P[iB,iA]QQ^{\dagger}[iB,iA]^{\dagger} P^{\dagger})}\sqrt{\tr([iB,iA][iB,iA]^{\dag})}\\
&\le  \tr(|[A,B]|^2).
\end{aligned}
\end{equation}
Thus we have the upper bound
\begin{equation}
\begin{aligned}\label{}
\left|\tr(\mathscr{M}(t)\mathscr{M}^{\dag}(t))\right|
&=\left|\int_0^t \d\tau_1 \int_0^{\tau_1}  \d\tau_2
\int_0^t \d\tau_1' \int_0^{\tau_1'}  \d\tau_2'
~\tr\left(P[iB,iA]Q[iB,iA]^{\dagger}
\right)\right|\\
&\leq \int_0^t \d\tau_1 \int_0^{\tau_1}  \d\tau_2
\int_0^t \d\tau_1' \int_0^{\tau_1'}  \d\tau_2'
~\left|\tr\left(P[iB,iA]Q[iB,iA]^{\dagger}
\right)\right|\\
&\le \frac{t^4}{4} \tr (|[A,B]|^2 ).
\end{aligned}
\end{equation}
Finally, using Theorem~\ref{Th:l2}, we have
\begin{align}\label{Eq;PF1AB}
R_{\ell_2}(\mathscr{U}_1(t),U_0(t))&\le \frac{t^2}{2} \left(\frac{\tr(|[A,B]|^2)}{d}\right)^{\frac{1}{2}}
\end{align}
as claimed.
\end{proof}

Next, we generalize this average performance with $H=A+B$ to a Hamiltonian which is the sum of multiple terms in the following lemma. We let $U_0[j,k](t):= \exp(-i\sum_{l=j}^{k}H_lt)$. We also define a corresponding product formula $U[j,k](t):=\prod_{l=j}^k e^{-iH_lt}$, with $U[k,k]=e^{-iH_{k}t}$ and $U[1,0](t)=\id$.

\begin{lemma}\label{Lemma:multiterm}
Consider a Hamiltonian $H=\sum_{l=1}^L H_l$.
The
average distance of
the first-order product formula
$\mathscr{U}_1(t)=e^{-iH_1t}e^{-iH_2t}\cdots e^{-iH_Lt}$
from the ideal evolution
$U_0(t)= \exp(-i\sum_{l=1}^{L}H_lt)$ satisfies
\begin{equation}
  R(\mathscr{U}_1(t),U_0(t))\le \sum_{k=1}^{L-1} R(U[k,k]U_0[k+1,L]\ket{\psi},U_0[k,L]\ket{\psi}),
\end{equation}
where $R$ can be $R_{\ell_2}$, $R_{t}$, or $R_{m}$ corresponding to the $\ell_2$ norm, the trace norm, or the absolute value with random projections, respectively.
\end{lemma}

\begin{proof}
We have
\begin{equation}
\begin{aligned}
R(\mathscr{U}_1(t),U_0(t))&=R(U[1,L],U_0[1,L])\\
&= \int_{\psi} \mc D(U[1,L]\ket{\psi},U_0[1,L]\ket{\psi}) \d\psi\\
&\le \sum_{k=1}^{L-1} \int_{\psi} \mc D(U[1,k]U_0[k+1,L]\ket{\psi},U[1,k-1]U_0[k,L]\ket{\psi}) \d\psi\\
&=\sum_{k=1}^{L-1} \int_{\psi} \mc D(U[k,k]U_0[k+1,L]\ket{\psi},U_0[k,L]\ket{\psi}) \d\psi\\
&= \sum_{k=1}^{L-1} R(U[k,k]U_0[k+1,L]\ket{\psi},U_0[k,L]\ket{\psi})
\end{aligned}
\end{equation}
as claimed.
\end{proof}

Using Lemmas \ref{Lemma:AB} and \ref{Lemma:multiterm}, we generalize the results to many-term Hamiltonians as follows.

\begin{lemma}\label{Lemma:one segment}
For a $d$-dimensional Hamiltonian $H=\sum_{l=1}^{L} H_l$,
the average $\ell_2$ norm error  of  $\mathscr{U}_1(t)=e^{-iH_1t}e^{-iH_2t}\cdots e^{-iH_Lt} $ with respect to $U_0(t) = e^{-itH}$ for a 1-design input ensemble satisfies
\begin{equation}
R_{\ell_2}(\mathscr{U}_1(t),U_0(t))\le \frac{t^2}{2}\sum_{l_1=1}^{L-1}\frac{1}{\sqrt{d}}\|[H_{l_1},\sum_{l_2=l_1+1}^{L}H_{l_2}]\|_F.
\end{equation}
\end{lemma}

\begin{proof}
From Lemma~\ref{Lemma:multiterm},
\begin{equation}
\begin{aligned}
R(U[1,L],U_0[1,L])\le \sum_{l_1=1}^{L-1} R(e^{-iH_{l_1}t} U_0[l_1+1,L],U_0[l_1,L]).
\end{aligned}
\end{equation}
Taking $A=\sum_{l_2=l_1+1}^{L}H_{l_2}$ and $B=H_{l_1}$ into Lemma \ref{Lemma:AB}, we have
\begin{equation}
 R_{\ell_2}(e^{-iH_{l_1}t} U_0[l_1+1,L]\ket{\psi},U_0[l_1,L]\ket{\psi})
 \le \frac{t^2}{2} \left(\frac{\tr(|[H_{l_1},\sum_{l_2=l_1+1}^{L}H_{l_2}]|^2)}{d}\right)^{\frac{1}{2}}.
\end{equation}
Then we can obtain the final result inductively.
\end{proof}

All of the above lemmas in this section describe time evolution by a single short time step (or segment). Here we apply the multi-segment results of Appendix~\ref{Sec:one-multi} to get an arbitrary-time bound for PF1.

\begin{theorem}[Triangle bound]\label{Th:PF1}
Let $H=\sum_{l=1}^{L} H_l$ be a $d$-dimensional Hamiltonian. For PF1 with $r$ segments, i.e., $\mathscr{U}^r_1(t/r)$ where $t/r < 1$,
the average $\ell_2$ norm error for a 1-design input ensemble $R_{\ell_2}(\mathscr{U}^r_1(t/r),U_0(t))$ has the upper bound
\begin{equation}
R_{\ell_2}(\mathscr{U}^r_1(t/r),U_0(t))\le \frac{t^2}{2r}\sum_{l_1=1}^{L-1}
\frac{1}{\sqrt{d}}
\|[H_{l_1},\sum_{l_2=l_1+1}^{L}H_{l_2}]\|_F.
\end{equation}
\end{theorem}

\begin{proof}
The result directly follows by combining Lemmas \ref{Lemma:one segment} and \ref{Lemma:segment1}.
\end{proof}

We prove a bound for PF2 in a similar fashion.

\begin{theorem}[Triangle bound]\label{Th:PF2}
Let $H=\sum_{l=1}^{L} H_l$ be a $d$-dimensional Hamiltonian. For PF2 with $r$ segments, i.e., $\mathscr{U}^r_2(t/r)$, the
average $\ell_2$ norm error for a 1-design input ensemble has the upper bound
\begin{equation}
R_{\ell_2}(\mathscr{U}^r_2(t/r),U_0(t))\le \frac{t^3}{r^2}\left\{\frac{1}{12} \sum_{l_1=1}^{L}
\frac{1}{\sqrt{d}}\|[\sum_{l_2=l_1+1}^L H_{l_2},[\sum_{l_2=l_1+1}^L H_{l_2},H_{l_1}]]  \|_F+ \frac{1}{24}  \sum_{l_1=1}^{L}\frac{1}{\sqrt{d}}
\|[H_{l_1},[H_{l_1},\sum_{l_2=l_1+1}^L H_{l_2}]]\|_F\right\}.
\end{equation}
\end{theorem}

\begin{proof}
We first consider a short time $t$, $H=A+B$, and $\mathscr{U}_2=e^{-iAt/2} e^{-iBt} e^{-iAt/2} $.
According to Appendix L of  Ref.~\cite{childs2020theory},
\begin{equation}
\begin{aligned}\label{Eq:PF2additive}
\mathscr{U}_2(t)=e^{-iHt}+&\int_0^t \d\tau_1\int_0^{\tau_1}\d\tau_2\int_0^{\tau_2} \d\tau_3 e^{-i(t-\tau_1)H}e^{-i\tau_1A/2}\\
&\cdot \left( e^{-i\tau_3B}\left[-iB,\left[-iB,-i\frac{A}{2}\right]\right]e^{i\tau_3B}+ e^{i\tau_3A/2}\left[i\frac{A}{2},\left[i\frac{A}{2},iB\right]\right]e^{-i\tau_3A/2}\right) e^{-\tau_1B} e^{-i\tau_1A/2}.
\end{aligned}
\end{equation}
The Frobenius norm of the additive error $\mathscr{A}=\mathscr{U}_2(t)-e^{-iHt}$ satisfies
\begin{align}
\tr(\mathscr{M}\mathscr{M}^{\dag})&=\tr (\mathscr{A}\mathscr{A}^{\dagger})=
\tr\Bigg\{\int_0^t \d\tau_1\int_0^{\tau_1}\d\tau_2\int_0^{\tau_2} \d\tau_3 \int_0^t \d\tau_1'\int_0^{\tau_1'}\d\tau_2'\int_0^{\tau_2'} \d\tau_3'\nonumber\\
&\left(E_1[-iB,[-iB,-i\frac{A}{2}]]F_1+ E_2[i\frac{A}{2},[i\frac{A}{2},iB]]F_2\right)\left(F_1^\dagger[-iB,[-iB,-i\frac{A}{2}]]^{\dagger} E_1^\dagger+ F_2^\dagger[i\frac{A}{2},[i\frac{A}{2},iB]]^\dagger E_2^\dagger\right)\Bigg\}\nonumber\\
&\le \frac{t^6}{36}
\left(\sqrt{\tr(|[-iB,[-iB,-i\frac{A}{2}]]|^2)}+
\sqrt{\tr(|[i\frac{A}{2},[i\frac{A}{2},iB]]|^2)}\right)^2 \nonumber\\
&=\frac{t^6}{36}
\left(\frac{1}{2}\sqrt{\tr|[B,[B,A]]|^2}+\frac{1}{4}
\sqrt{\tr|[A,[A,B]]|^2}\right)^2,
\end{align}
where $E_1$, $E_2$, $F_1$, $F_2$ are all unitary according to Eq.~\eqref{Eq:PF2additive} and the inequality is due to Lemma~\ref{Lemma:traceproduct}. Then we have
\begin{equation}\label{Eq:PF2AB}
R_{\ell_2}\le     \frac{t^3}{12} \left(\frac{\tr(|[B,[B,A]]|^2)}{d}\right)^{\frac{1}{2}}+ \frac{t^3}{24} \left(\frac{\tr(|[A,[A,B]]|^2)}{d}\right)^{\frac{1}{2}}.
\end{equation}
For $H=H_1+H_2+\dots+H_L$, we define
\begin{equation}
U_2[l_1] := e^{-iH_{l_1}t/2} e^{-i\sum_{l=l_1+1}^LH_{l}}e^{-iH_{l_1}t/2},\quad U_0[l_1]:=e^{-it\sum_{l=l_1}^LH_{l}}, \quad
\end{equation}
and similarly obtain the inequality
\begin{align}
    R_{\ell_2}(U_2,U_0)\le \sum_{l_1=1}^{L-1} R(U_2[l_1], U_0[l_1])).
\end{align}
Consequently, we obtain
\begin{equation}
R_{\ell_2}\le     \frac{t^3}{12} \sum_{l_1=1}^{L}
\left(\frac{\tr|[\sum_{l_2=l_1+1}^L H_{l_2},[\sum_{l_2=l_1+1}^L H_{l_2},H_{l_1}]]|^2}{d}\right)^{\frac{1}{2}} +  \frac{t^3}{24}  \sum_{l_1=1}^{L}\left(\frac{\tr|[H_{l_1},[H_{l_1},\sum_{l_2=l_1+1}^L H_{l_2}]]|^2}{d}\right)^{\frac{1}{2}}.
\end{equation}
For a long-time evolution over time $t$, we apply Lemma~\ref{Lemma:segment1} to complete the proof.
\end{proof}

\section{Error interference with nearest-neighbor interacting Hamiltonians}
\label{Sec:interference}
In general, the average distance of PF1 from the exact time evolution operator for evolution time $t$ can be bounded as
\begin{equation}
R(\mathscr{U}_1^r(t),U_0(t))\le r  R(\mathscr{U}_1(t/r),U_0(t/r)).
\end{equation}
However, this bound does not leverage the fact that errors from different time steps may interfere. This fact can be used to achieve superior worst-case error bounds for PF1 with nearest-neighbor Hamiltonians \cite{Tran_2020}.

Figure~\ref{Fig:PF12} suggests that the empirical average error for PF1 suggests also benefits from destructive error interference for nearest-neighbor Hamiltonians,
so we apply the analysis of \cite{Tran_2020} to the average case. Specifically, in this section, we show that the average error of PF1 with Hamiltonian $H=A+B$ can be further tightened as
\begin{equation}
R(\mathscr{U}_1^r(t/r),U_0(t))\approx \mathcal O\left(\frac{r}{t}  R(\mathscr{U}_1(t/r),U_0(t/r))\right).
\end{equation}
The main idea is to apply the following approximations:
\begin{equation}
\begin{aligned}
&\mathscr{M}_r\approx   \sum_{j=0}^{r-1} (U_0^\dagger)^{j}\mathscr{M}_1 U_0^{j},
~\mathscr{M}_1\approx [B,A]\frac{t^2}{2r^2}=[H,A]\frac{t^2}{2r^2},\\
&\sum_{j=0}^{r-1} (U_0^\dagger)^{j}[H,A] U_0^{j}
\approx \frac{r}{t}\int_0^t \d x \, U_{-x}[H,A]U_x
= \frac{r}{t} \big(U_{-t}(iA)U_{t}-iA\big),
\end{aligned}
\end{equation}
where $ \mathscr{U}_1^r(t/r)=U_0(t)(\id+ \mathscr{M}_r)$, $U_0=e^{-iHt/r}$ is the ideal evolution for one segment, and $U_x=e^{-iHx}$.
The rigorous result is as follows.

\begin{theorem}[Interference bound]\label{Th:interference}
Consider the Hamiltonian $H=\sum_{j,j+1}H_{j,j+1}$ where
$H_{j,j+1}$ acts nontrivially on qubits $j, j+1$, and $\|H_{j,j+1}\|\le 1$.
Let
$\mathscr{U}_1(t)=e^{-iAt}e^{-iBt}$ where $A=\sum_{\text{odd}~j} H_{j,j+1}$, $B=\sum_{\text{even}~j} H_{j,j+1}$.
Then, provided $nt^2/r$ is less than a small constant, the average $\ell_2$  norm error for a 1-design input ensemble has the upper bound
\begin{equation}
 R_{\ell_2}(\mathscr{U}^r_1(t/r),U_0(t))=O\left( \sqrt{n}\left(\frac{t}{r}+\frac{t^3}{r^2}\right)\right).
\end{equation}
\end{theorem}

\begin{proof}
We denote the multiplicative error for $r$ segments and $1$ segment as $\mathscr{M}_r$ and $\mathscr{M}_1$, respectively. Here $U_0=e^{-iHt/r}$ is the ideal evolution in one segment. We have
\begin{equation}\label{}
\begin{aligned}
  \mathscr{M}_r&=(U_0^\dagger)^r[U_0(I+ \mathscr{M}_1)]^r-\id\\
&=\underbrace{\sum_{j=0}^{r-1} (U_0^\dagger)^{j}\mathscr{M}_1 U_0^{j}}_{\Delta_1}+\underbrace{\sum_{j_1=1}^{r-1}\sum_{j_2=1}^{r-j_1} (U_0^\dagger)^r U_0^{j_1}\mathscr{M}_1U_0^{j_2}\mathscr{M}_1U_0^{r-j_1-j_2}}_{\Delta_2}+\sum_{k=3}^{r}\Delta_k,
\end{aligned}\end{equation}
where $\Delta_k$ represents the $k$th-order term of $\mathscr{M}_1$, with
$\|\Delta_k\|=\mathcal O(r^k \|\mathscr{M}_1\|)$.
Assuming $r\|\mathscr{M}_1\|=\mathcal O(nt^2/r)$ is less than a small constant, we can use Lemmas~\ref{Lemma:Delta} and \ref{Lemma:deltak} below to show that
\begin{equation}
\tr(\mathscr{M}_r\mathscr{M}_r^{\dagger})   \le \sum_{k,k'=1}^{r} |\tr(\Delta_k\Delta_{k'}^{\dagger}) |\le   \biggl(\sum_{k,k'=1}^{r} r^{k+k'-2} \|\mathscr{M}_1\|^{k+k'-2}  \biggr)\mathcal O\biggl(d n\left(\frac{t}{r}+\frac{t^3}{r^2}\right)^2\biggr)= \mathcal O\biggl(d n\left(\frac{t}{r}+\frac{t^3}{r^2}\right)^2\biggr),
\end{equation}
where $d$ is the dimension of the Hilbert space acted on by the Hamiltonian.
Therefore, by Theorem~\ref{Th:l2}, we have
\begin{equation}
R_{\ell_2}(\mathscr{U}^r_1(t/r),U_0(t))\le \frac{1}{\sqrt{d}}\|\mathscr{M}(t)   \|_F=   \mathcal O\biggl( \sqrt{n}\left(\frac{t}{r}+\frac{t^3}{r^2}\right)\biggr)
\end{equation}
as claimed.
\end{proof}

Next we prove Lemma~\ref{Lemma:Delta} (using intermediate Lemmas~\ref{Lemma:SV} and \ref{Lemma:Fk}) and Lemma \ref{Lemma:deltak}.

\begin{lemma}\label{Lemma:Delta}
If $\frac{tn}{r}$ is less than a small constant, we have
\begin{equation}
\tr(\Delta_1 \Delta_{1}^{\dagger})=\mathcal O\biggl(d n\left(\frac{t}{r}+\frac{t^3}{r^2}\right)^2\biggr).
\end{equation}
\end{lemma}

\begin{proof}
According to Lemma 1 in Ref.~\cite{Tran_2020},
we have
\begin{equation}
 \mathscr{M}_1=U_0^\dagger \delta, ~\delta=-[H,S]-V,
\end{equation}
where
\begin{equation}
 S=\sum_{k=2}^{\infty}  \frac{(-it)^k}{k!r^k} S_k, ~V= \sum_{k=3}^{\infty}  \frac{(-it)^k}{k!r^k} V_k
\end{equation}
with expressions for $S_k$ and $V_k$ given in Lemma 1 of Ref.~\cite{Tran_2020}.
The destructive error interference
is due to the following equation for any positive integer $a$:
\begin{equation}
\sum_{j=0}^{a-1} (U_0^\dagger)^{j}[H,S] U_0^{j}=\frac{r}{t}
\sum_{k=0}^{\infty}
I_{at/r}(F^{\circ k}(S)),
\end{equation}
where \begin{equation}
\begin{aligned}
&I_{t}(X) =\int_0^t U_s[H,X]U_{-s}\,\d s= U_t iX U_{-t}-iX,\\
&F(X)=  -\frac{r}{t} \int_0^{\frac{t}{r}}\d s \int_0^s \d v \, U_{v}[H,X]U_{-v},
\end{aligned}
\end{equation}
and  $F^{\circ k}(S)$ is the $k$th iterate of the function $F$, namely $F^{\circ k}(S)=\underbrace{F(F(\dots F}_{k}(S)))$, with $F^{\circ 0}(S)=S$.
Now we write
\begin{equation}
    \begin{aligned}
\Delta_1&=  U_0^\dagger\sum_{j=0}^{r-1} (U_0^\dagger)^{j}\delta U_0^{j}
= -U_0^\dagger\sum_{j=0}^{r-1} (U_0^\dagger)^{j}[H,S] U_0^{j}-\sum_{j=0}^{r-1} (U_0^\dagger)^{j+1}V U_0^{j}\\
&= -\frac{r}{t}
\sum_{k=0}^{\infty}
U_0^\dagger I_{t}(F^{\circ k}(S))
-\sum_{j=0}^{r-1} (U_0^\dagger)^{j+1}V U_0^{j}\\
&=\underbrace{-\frac{r}{t}
\sum_{k=0}^{\infty}\left[
U_0^\dagger U_t (iF^{\circ k}(S)) U_{-t}-iF^{\circ k}(S)\right]}_{\widetilde{F}_1}- \underbrace{\sum_{j=0}^{r-1} (U_0^\dagger)^{j+1}V U_0^{j}}_{\widetilde{V}_1}.
   \end{aligned}
\end{equation}
Then we write
\begin{equation}
 \tr(\Delta_1 \Delta_{1}^{\dagger})=  \tr(|\widetilde{F}_1|^2)+
 \tr(|\widetilde{V}_1|^2)+\tr(\widetilde{F}_1\widetilde{V}_1^{\dagger})+ \tr(\widetilde{V}_1\widetilde{F}_1^{\dagger})
 \end{equation}
and separately bound each of the four parts:
 \begin{equation}
 \begin{aligned}
 \tr(|\widetilde{F}_1|^2)
 &\le
 \frac{r^2}{t^2}
\sum_{k,k'=0}^{\infty}\left|\tr\left\{\left[
U_0^\dagger U_t (iF^{\circ k}(S)) U_{-t}-iF^{\circ k}(S)\right]
\left[
U_0^\dagger U_t (iF^{\circ k'}(S)) U_{-t}-iF^{\circ k'}(S)\right] \right\}\right| \quad \text{(Triangle inequality)}\\
&\le \frac{r^2}{t^2}
\sum_{k,k'=0}^{\infty}4 \sqrt{\tr(|F^{\circ k}(S)|^2)}\sqrt{\tr(|F^{\circ k'}(S)|^2)} \quad \text{(Lemma~\ref{Lemma:traceproduct})}\\
&=\frac{4r^2}{t^2}
\left(\sum_{k=0}^{\infty} \sqrt{\tr(|F^{\circ k}(S)|^2)}\right)^2\\
&=\mathcal O\left(d\frac{nt^2}{r^2}\right) \quad \text{(Lemma~\ref{Lemma:Fk})};\\
\tr(|\widetilde{V}_1|^2)&\le \sum_{j,j'=0}^{r-1} \left|\tr\left((U_0^\dagger)^{j+1}V U_0^{j}U_0^{-j'}V^{\dagger} U_0^{j'+1}\right)\right| \quad \text{(Triangle inequality)}\\
&\le r^2 \tr(VV^{\dagger}) \quad \text{(Lemma~\ref{Lemma:traceproduct})} \\
&=\mathcal O\left(d\frac{t^6}{r^4}n\right) \quad \text{(Lemma~\ref{Lemma:SV})};\\
|\tr(\widetilde{F}_1\widetilde{V}_1^{\dagger})|&\le
 \frac{r}{t}
\sum_{k=0}^{\infty}\sum_{j=0}^{r-1}\left|\tr\left\{\left[
U_0^\dagger U_t (iF^{\circ k}(S)) U_{-t}-iF^{\circ k}(S)\right]U_0^{-j}V^{\dagger} U_0^{j+1}\right\}\right| \quad \text{(Triangle inequality)}\\
&\le
 \frac{2r^2}{t} \left(\sum_{k=0}^{\infty} \sqrt{\tr(|F^{\circ k}(S)|^2)}\right) \sqrt{\tr(VV^{\dagger})} \quad \text{(Lemma~\ref{Lemma:traceproduct})}\\
 &=\mathcal O\left(dn\frac{t^4}{r^3}\right) \quad \text{(Lemma~\ref{Lemma:SV} and \ref{Lemma:Fk})};
\end{aligned}\label{Eq:FF}
 \end{equation}
 The upper bound on $\tr(\widetilde{V}_1\widetilde{F}_1^{\dagger})$ is similar to that of  $\tr(\widetilde{F}_1\widetilde{V}_1^{\dagger})$ above.
\end{proof}

\begin{lemma}\label{Lemma:SV}
If $\frac{tn}{r}$ is less than a small constant, we have
\begin{equation}
\begin{aligned}
 &\tr(SS^{\dagger})=\mathcal O\left(d\frac{t^4}{r^4}n\right), ~\tr(VV^\dagger)=  \mathcal O \left(d \frac{t^6}{r^6}n\right), ~|\tr(SV^{\dagger})|=\mathcal O\left(d\frac{t^5}{r^5}n\right).
\end{aligned}
\end{equation}
\end{lemma}
\begin{proof}
First, we calculate the components in $SS^{\dagger}$, $VV^{\dagger}$, and $SV^{\dagger}$.
For completeness, we rephrase some results in Ref.~\cite[Lemmas S1 and S2]{Tran_2020}. We have
\begin{equation}
    \begin{aligned}
&S_{k+1}= S_k H-\sum^{k-1}_{j=0}H^{k-1-j}BH^j,\qquad V_{k+1}=[A,[H,S_k]]+AV_k+V_kB.
    \end{aligned}
\end{equation}
$S_2=B$, $V_2=0$, and $S_k$, $V_k$ (for $k\ge 3$) can be written as
\begin{equation}
    \begin{aligned}
V_k&=\sum_{i=1}^{n_k} v_{k,i}~\text{with}~n_k\le C e^{k-2}n^{k-2} &
S_k&=\sum_{i=1}^{m_k} s_{k,i}~\text{with}~m_k\le \frac{k(k-1)}{2}n^{k-1},
    \end{aligned}
\end{equation}
where $\|v_{k,i}\|, \|s_{k,i}\|\le 1$. We let $\sharp X$ denote the number of terms in an operator $X$ (e.g., $\sharp S_k=m_k\le \frac{k(k-1)}{2}n^{k-1}$).
Then
\begin{equation}
\begin{aligned}
|\tr(S_kS_{k'}^{\dagger})|&= \biggl|\tr\Big(S_kHS_{k'-1} -S_k\sum^{k'-2}_{j=0}H^jB^{\dagger} H^{k'-2-j}\Big)\biggr|\\
&=\biggl|\tr \Big[(S_{k'-1}S_k- S_k\sum^{k'-2}_{j=0}H^jB^{\dagger} H^{k'-3-j})H\Big]\biggr|\\
&=|\tr (T_{k,k'}H)| \\
&\le d \, \sharp T_{k,k'}\\
&\le d m_k (m_{k'-1}+k'n^{k'-2}) \\
&\le d\frac{k(k-1)k'(k'-1)}{4}n^{k+k'-3}\\
&=\mathcal O (d k^2{k'}^2 n^{k+k'-3}),
\end{aligned}
\end{equation}
where $T_{k,k'}=S_{k'-1}S_k- S_k\sum^{k-1}_{j=0}H^jB^{\dagger} H^{k-2-j}$.
The inequality $\tr (T_{k,k'}H)\le d \, \sharp T_{k,k'}$ follows since for each term in $T_{k,k'}$, $H$ has at most one term $H_{j,j+1}$ with the same support as the term in $T_{k,k'}$.
It is also easy to check that $\tr(S_kS_{2}^{\dagger})\le dm_k\le d\frac{k(k-1)}{2}n^{k-1}$ and $\tr(S_2S_{2}^{\dagger})\le dn$.

Similarly, we have
\begin{equation}
\begin{aligned}
|\tr(S_kV_{k'}^{\dagger})|&\le n_{k'}(m_{k-1}+kn^{k-2})
\le d\frac{Ck(k-1)}{2}e^{k'-2}   n^{k+k'-4}
=\mathcal O (d e^{k'-2}k^2 n^{k+k'-4}).
\end{aligned}
\end{equation}
We have $V_3=[A,[H,B]]$ with $\sharp V_3\le 4n$ and
$\tr(V_3V_3^{\dagger})=\mathcal O(n)$. For $k>3$,
\begin{equation}
\begin{aligned}\label{Eq:vkvk'}
|\tr(V_kV_{k'})|&=|\tr(V_{k'}[A,[H,S_{k-1}]])+ \tr(V_{k'}V_{k-1}H)|\\
&\le |\tr(V_{k'}A[H,X]H)| +|\tr(V_{k'}A[X,H]H)|+ |\tr(V_{k'}V_{k-1}H)|\\
&\le d \, \sharp V_{k'}[ 8(k-1)\sharp X +\sharp V_{k-1}]  \\
&\le  dC^2 e^{k+k'-4}n^{k+k'-5},
\end{aligned}
\end{equation}
where $S_{k-1}=XH$ with $X=S_{k-2}-\sum^{k-1}_{j=0}H^{k-3-j}BH^{j-1}$, which has $\sharp X\le \frac{(k-1)(k-2)}{2}n^{k-3}$ terms.
We have
\begin{equation}
\begin{aligned}
\tr(SS^\dagger)=
\sum_{k,k'=2}^{\infty}   \frac{t^{k+k'}}{k!k'!r^{k+k'}} \tr(S_kS_{k'})\le \frac{d}{n} \left(  \sum_{k=2}^{\infty} \frac{t^kn^{k-1}}{2r^k}  \right)^2\le
dn\frac{t^4}{r^4}\left(\frac{1}{1-\frac{nt}{r}}\right)^2=
\mathcal O (dn\frac{t^4}{r^4}).
\end{aligned}
\end{equation}
Also, $\tr(VV^\dagger)=  \mathcal O (dn \frac{t^6}{r^6})$ and $|\tr(SV^{\dagger})|=\mathcal O(dn\frac{t^5}{r^5}) $.
Note that we need the additional assumption that $\frac{tn}{r}$ is less than a small constant to ensure the convergence of $ \sum_{k=2}^{\infty} \frac{t^kn^{k-1}}{2r^k}$.
\end{proof}

\begin{lemma}\label{Lemma:Fk}
For any nonnegative integers $k,k'$, we have
\begin{equation}
|\tr(F^{\circ k}(S) F^{\circ k'}(S)^\dagger)|= \mathcal O\Big(d\frac{t^{k+k'+4}}{r^{k+k'+4}}n^{k+k'+1}\Big),
\end{equation}
As a consequence, we obtain that $\sum_{k=0}^{\infty} \sqrt{\tr(|F^{\circ k}(S)|^2)}=\mathcal O (\sqrt{dn}\frac{t^2}{r^2})$ with the assumption that $\frac{nt}{r}$ is less than a small constant.
\end{lemma}
\begin{proof}
Since $H$ and $U_v$ commute, we can move $H$ into the integrals $\int_0^{\frac{t}{r}}\d s \int_0^s \d v$ and $\int_0^{\frac{t}{r}}\d s' \int_0^{s'} \d v'$ ,
giving
\begin{align}
 \tr\left(|[H,F^{\circ k}(S)]|^2\right)&=\tr\biggl\{\Bigr[H,-\frac{r}{t} \int_0^{\frac{t}{r}}\d s \int_0^s \d v U_{v}[H,F^{\circ k-1}(S)]U_{-v}\Bigl]\Bigl[H,-\frac{r}{t} \int_0^{\frac{t}{r}}\d s' \int_0^{s'} \d v' U_{v'}[H,F^{\circ k-1}(S)]U_{-v'}\Bigr]^{\dagger}\nonumber\biggr\}\\
 & =\frac{r^2}{t^2} \tr\biggl\{\int_0^{\frac{t}{r}}
 \d s \int_0^s \d v U_{v}[H,[H,F^{\circ k-1}(S)]]U_{-v}  \int_0^{\frac{t}{r}}
 \d s' \int_0^{s'} \d v' U_{v'}[H,[H,F^{\circ k-1}(S)]]^{\dagger}U_{-v'}\biggr\}\nonumber\\
 &\le \frac{r^2}{t^2} \int_0^{\frac{t}{r}}\d s \int_0^s \d v \int_0^{\frac{t}{r}}\d s' \int_0^{s'} \d v' \tr\left( \left|[H,[H,F^{\circ k-1}(S)]]\right|^2\right)\nonumber\\
 &=\frac{t^2}{4r^2} \tr\left( \left|[H,[H,F^{\circ k-1}(S)]]\right|^2\right).
\end{align}
We denote a $k$-layer nested commutator by $N^k(S)=\underbrace{[H,[H,\dots,[H}_{k},S]]]$. We have the bound
\begin{equation}
\tr\left( |N^{k}(S)|^2\right)\le 4^{k}n^{2k} \tr(S S^{\dagger}) =\mathcal O\left(d4^{k}\frac{t^4}{r^{4}}
n^{2k+1}\right),
\end{equation}
so
\begin{equation}
\begin{aligned}
 \tr\left(|[H,F^{\circ k-1}(S)]|^2\right)\le\frac{t^{2k-2}}{4^{k-1}r^{2k-2}} \tr\left( |N^{k}(S)|^2\right)=\mathcal O\Big(d\frac{t^{2k+2}}{r^{2k+2}}n^{2k+1}\Big).
\end{aligned}
\end{equation}
When $k,k'\ge 1$, this gives
\begin{equation}
\begin{aligned}
|\tr(F^{\circ k}(S)F^{\circ k'}(S))|&\le  \frac{r^2}{t^2} \int_0^{\frac{t}{r}}\d s \int_0^s \d v \int_0^{\frac{t}{r}}\d s' \int_0^{s'} \d v' |\tr(U_{v}[H,F^{\circ k-1}(S)]U_{-v}U_{v'}[H,F^{\circ k'-1}(S)]U_{-v'})|\\
&\le  \frac{r^2}{t^2} \int_0^{\frac{t}{r}}\d s \int_0^s \d v \int_0^{\frac{t}{r}}\d s' \int_0^{s'} \d v' \sqrt{\tr{|[H,F^{\circ k-1}(S)]|^2}}\sqrt{\tr{|[H,F^{\circ k'-1}(S)]|^2}}\\
&= \frac{t^2}{4r^2} \sqrt{\tr{|[H,F^{\circ k-1}(S)]|^2}}\sqrt{\tr{|[H,F^{\circ k'-1}(S)]|^2}}\\
&=\mathcal O\Big(d\frac{t^{k+k'+4}}{r^{k+k'+4}}n^{k+k'+1}\Big).
\end{aligned}
\end{equation}
When $k=0$ or $k'=0$, $F^{\circ 0}(S)=S$, similarly, we have
\begin{equation}
\begin{aligned}
|\tr(F^{\circ 0}(S)F^{\circ k'}(S))|&\le
 \frac{t^2}{4r^2} \sqrt{\tr{(SS^{\dagger})}}\sqrt{\tr{|[H,F^{\circ k'-1}(S)]|^2}}\\
&=\mathcal O\Big(d\frac{t^{k'+4}}{r^{k'+4}}n^{k'+1}\Big).
\end{aligned}
\end{equation}
When $k, k'=0$, $|\tr(F^{\circ k}(S)F^{\circ k'}(S))|=|\tr(SS^{\dagger})|=\mathcal O\Big(d\frac{t^{4}}{r^{4}}n\Big)$ because of Lemma.~\ref{Lemma:SV}.

We conclude the result
\begin{equation}
|\tr(F^{\circ k}(S) F^{\circ k'}(S)^\dagger)|= \mathcal O\Big(d\frac{t^{k+k'+4}}{r^{k+k'+4}}n^{k+k'+1}\Big),
\end{equation}
for any nonnegative integers $k,k'$. As a special case, when $k=k'$, we have
$|\tr(|F^{\circ k}(S)|^2)|=\mathcal O(d\frac{t^{2k+4}}{r^{2k+4}}n^{2k+1})$.
Then, $\sum_{k=0}^{\infty} \sqrt{\tr(|F^{\circ k}(S)|^2)}\le
\mathcal O (\sqrt{dn}\frac{t^2}{r^2})\frac{1}{1-\mathcal O(n\frac{t}{r})}=\mathcal O (\sqrt{dn}\frac{t^2}{r^2})$ using the assumption that $\frac{nt}{r}$ is less than a small constant.
 \end{proof}

The above three lemmas establish the upper bound for the leading term $\tr(\Delta_1 \Delta_{1}^{\dagger})$.
In the following lemma, we bound the other, higher-order terms.

\begin{lemma}\label{Lemma:deltak}
For any positive integers $k,k'$, we have
\begin{equation}
   | \tr(\Delta_k \Delta_{k'}^{\dagger})|=\mathcal O\biggl(d \frac{n^{k+k'-1} t^{2(k+k'-2)} }{r^{2(k+k'-2)}}\left(\frac{t}{r}+\frac{t^3}{r^2}\right)^2\biggr).
\end{equation}
\end{lemma}

\begin{proof}
We have
\begin{equation}
\begin{aligned}
\Delta_k&=\sum_{j_1=1}^{r-k+1}\sum_{j_2=1}^{r-j_1}\dots \sum_{j_k=1}^{r-j_1\dots-j_{k-1}}
(U_0^\dagger)^{r-1} U_0^{j_1}\mathscr{M}_1U_0^{j_2}\mathscr{M}_1\dots U_0^{j_k}\mathscr{M}_1 U_0^{r-j_1\dots-j_{k-1}}\\
&=\sum_{j_1=1}^{r-k+1}\sum_{j_2=1}^{r-j_1}\dots \sum_{j_k=0}^{r-j_1\dots-j_{k-1}-1}
(U_0^\dagger)^{r-1} U_0^{j_1}\mathscr{M}_1U_0^{j_2}\mathscr{M}_1\dots U_0U_0^{j_k}\mathscr{M}_1 U_0^{-j_k}U_0^{r-j_1\dots-j_{k-1}+j_k}\\
&=\sum_{i=(j_1,\dots,j_{k-1})} \sum_{j_k=0}^{r-j_1\dots-j_{k-1}-1}  P_iU_0^{j_k}\mathscr{M}_1 U_0^{-j_k} Q_i,
\end{aligned}
\end{equation}
where $\|P_i\|\le \|\mathscr{M}_1\|^{k-1}$ and $Q_i$ is unitary.
The operator $\Delta_1^* := \sum_{j_k=0}^{r-j_1\dots-j_{k-1}-1}  U_0^{j_k}\mathscr{M}_1 U_0^{-j_k}$
can be bounded via the same procedure in Lemma~\ref{Lemma:Delta}, giving
$\tr(|\Delta_1^*|^2)=\mathcal O\Bigl(d n\left(\frac{t}{r}+\frac{t^3}{r^2}\right)^2\Bigr).$ Then we have
\begin{equation}
\begin{aligned}
\tr(\Delta_k \Delta_{k'}^{\dagger})&=\sum_{i,i'} \tr\left( P_i\Delta_1^* Q_i Q_{i'}^{\dagger} (\Delta_1^*)^{\dagger}P_{i'}^{\dagger}\right)\\
&\le \sum_{i,i'} \sqrt{\tr(P_{i}^{\dagger}P_i\Delta_1^* Q_i Q_{i}^{\dagger} (\Delta_1^*)^{\dagger})}\sqrt{\tr(P_{i'}^{\dagger}P_{i'}\Delta_1^* Q_{i'} Q_{i'}^{\dagger} (\Delta_1^*)^{\dagger})}\\
&\le \sum_{i,i'}\|P_i\|\|P_{i'}\|\tr(|\Delta_1^*|^2)\\
&\le   r^{k+k'-2} \|\mathscr{M}_1\|^{k+k'-2}  \tr(|\Delta_1^*|^2)\\
&=       O\biggl(d \frac{n^{k+k'-1} t^{2(k+k'-2)} }{r^{2(k+k'-2)}}\left(\frac{t}{r}+\frac{t^3}{r^2}\right)^2\biggr),
\end{aligned}
\end{equation}
where in the last equation, we used the fact that $\|\mathscr{M}_1\|= \mathcal O(\frac{nt^2}{r^2})$.
\end{proof}

This concludes our analysis of the asymptotic average-case performance of PF1 and PF2. In Appendix~\ref{Sec:app}, we give a concrete prefactor for the interference bound.


\section{Comparing bounds with and without the Cauchy-Schwarz inequality}\label{sec:cauchycompare}

According to Eq.~\eqref{eq:l2upp}, the Haar-averaged error in terms of the $\ell_2$ norm can be expressed as
\begin{equation}\label{eq:harrexact}
    R_{\ell_2}^{\mathrm{Haar}}(U_0,U)
    = \mathbb{E}_{\psi\in\mathrm{Haar}} \sqrt{2 - \bra{\psi} U^\dagger U_0 \ket{\psi} - \bra{\psi} U_0^\dagger U \ket{\psi}}
    = \mathbb{E}_{\psi\in\mathrm{Haar}} \sqrt{\bra{\psi} \mathscr{M}\mathscr{M}^\dagger \ket{\psi}}
\end{equation}
where $\mathscr{M} := U_0^\dag U - \mathbb{I}$.
Our analysis in Appendix~\ref{Sec:Average} upper bounds this quantity using the Cauchy-Schwarz inequality.
However, in principle there exist tighter upper bounds that may be evaluated using Eq.~\eqref{eq:exactbound} below, or estimated through sampling.

Here we explore the effect of using the Cauchy-Schwarz inequality in bounding $R_{\ell_2}^{\mathrm{Haar}}(U_0,U)$. Our goal is to understand when it is advantageous to use
a tighter bound that might be harder to calculate,
and when the improvement is negligible.
In this section, all averages are over the Haar measure.

Let $G$ be a positive semidefinite operator with spectrum $\{\lambda_j\}_{j=1}^d$, satisfying
\begin{equation}
    R_{\ell_2}^{\mathrm{Haar}}(U_0,U)
    \leq \mathbb{E}_{\psi\in\mathrm{Haar}} \sqrt{\bra{\psi} G \ket{\psi}}.
\end{equation}
For example, $G$ could be $\mathscr{M}\mathscr{M}^\dagger$, in which case the inequality is actually an equality. At the end of this section, we give an example of a bound for the first-order product formula approximation $U = e^{-itB}e^{-itA}$ of $U_0 = e^{-it(A+B)}$ using $G \ne \mathscr{M}\mathscr{M}^\dagger$. The following exact formula holds \cite[Section 4.1]{Jones_1991}:
\begin{equation}\label{eq:exactbound}
    \mathbb{E}_{\psi}  \sqrt{\bra{\psi} G \ket{\psi}} = \frac{\sqrt{\pi}}{2} \frac{\Gamma(d)}{\Gamma(d+\frac{1}{2})} \sum_{j=1}^d \frac{\lambda_j^{d-\frac{1}{2}}}{\prod_{k\neq j} (\lambda_j - \lambda_k)}.
\end{equation}
In the case of degenerate eigenvalues, we can take limits of the right-hand side.

Using the Cauchy-Schwarz inequality, we can upper bound this quantity by
\begin{equation}\label{eq:cauchybound}
    \mathbb{E}_{\psi}  \sqrt{\bra{\psi} G \ket{\psi}} \leq \sqrt{\frac{\mathrm{Tr}(G)}{d}} =  \sqrt{\frac{\sum_{j=1}^d \lambda_j}{d}}
\end{equation}
(which holds regardless of degeneracy), which is easier to compute or bound in practice.
Define $\Lambda := \max(\{\lambda_j\}_{j=1}^d) = \| G \|.$ Then the difference between \eqref{eq:exactbound} and the right-hand side of \eqref{eq:cauchybound} is
\begin{align}
   \sqrt{\frac{\sum_{j=1}^d \lambda_j}{d}} - \frac{\sqrt{\pi}}{2} \frac{\Gamma(d)}{\Gamma(d+\frac{1}{2})} \sum_{j=1}^d \frac{\lambda_j^{d-\frac{1}{2}}}{\prod_{k\neq j} (\lambda_j - \lambda_k)}\nonumber
    & = \sqrt{\Lambda} \Biggl( \sqrt{\frac{\sum_{j=1}^d \lambda_j/\Lambda}{d}} - \frac{\sqrt{\pi}}{2} \frac{\Gamma(d)}{\Gamma(d+\frac{1}{2})} \sum_{j=1}^d \frac{(\lambda_j/\Lambda)^{d-\frac{1}{2}}}{\prod_{k\neq j} (\lambda_j - \lambda_k)/\Lambda} \Biggr ) \\
    &=:\sqrt{\Lambda}D(\{\lambda_j\}_{j=1}^d) .
    \label{eq:difdefine}
\end{align}
The factor $D(\{\lambda_j\}_{j=1}^d)<1$ represents the difference in these two bounds due to the distribution of eigenvalues of $G$.

    Before analyzing $D(\{\lambda_j\}_{j=1}^d)$, we show how it relates to error bounds for Hamiltonian simulation. Consider a product formula simulation with $r$ Trotter steps. By the triangle inequality, the overall average error $\epsilon$ satisfies $\epsilon \leq r \delta,$ where $\delta$ is the average error of a single Trotter step. Let $\delta_1$ be the error bound of the form Eq.~\eqref{eq:exactbound}. If $G$ is $\mathscr{M}\mathscr{M}^\dag$ where $\mathscr{M}$ is the multiplicative error of a single Trotter step, then $\delta = \delta_1$, as mentioned before. Let $\delta_2$ be the bound on $\delta_1$ using the right-hand side of Eq.~\eqref{eq:cauchybound}. In this notation,
\begin{equation}\label{eq:totsimerrnc}
    \epsilon \leq r \delta_1 \leq r \delta_2 = r \delta_1 + r \sqrt{\Lambda} D(\{\lambda_j\}_{j=1}^d),
\end{equation}
where $\Lambda = \|G\|$.

Note that as $r$ increases, $\Lambda$ can decrease in a way that depends on the order of the product formula. However, in a limited-resource scenario, $r$ may be fixed. Then we would like to know how large $r \sqrt{\Lambda} D(\{\lambda_j\}_{j=1}^d)$ can be. We may alter $\Lambda$ by changing the total evolution time $t$, but the scaling of $\Lambda$ with $t$ again depends on the product formula used, so we focus on $D(\{\lambda_j\}_{j=1}^d)$.

To help understand when $D(\{\lambda_j\}_{j=1}^d)$ is non-negligible, consider the scenario where one eigenvalue of $G$ is much larger than the rest. This scenario maximizes the gap between the worst-case error (as given by $\sqrt{\|G\|}$) and the bound on the average-case error produced by the trace bound Eq.~\eqref{eq:cauchybound}, which is the square root of the average eigenvalue of $G$. However, it is unclear when the gap between the trace bound and the exact formula of Eq.~\eqref{eq:exactbound} is largest.

To evaluate this situation analytically, we consider the spectrum $\lambda_{\mathrm{one}} = \{\Lambda, 0,0,0,\ldots, 0\}$. Then we have
\begin{equation}\label{eq:allbutonebound}
    D(\lambda_{\mathrm{one}}) = \frac{1}{\sqrt{d}} - \frac{\sqrt{\pi}}{2} \frac{\Gamma(d)}{\Gamma(d+\frac{1}{2})}
    =\biggl(1-\frac{\sqrt{\pi}}{2}\biggr)\frac{1}{\sqrt d} + O(d^{-3/2}).
\end{equation}
Out of several examples we consider numerically, we find that this distribution results in the largest $D(\{\lambda_j\}_{j=1}^d)$. For $d\approx 10^3$, we find $ D(\lambda_{\mathrm{one}}) \lesssim 10^{-2}$. In this case, if the total average error resolution is around $10^{-n}$, provided $r \sqrt{\|G\|} \lesssim 10^{2-n}$, the trace bound (Eq. \eqref{eq:cauchybound}) should give a nearly tight result (assuming $\lambda_{\text{one}}$ indeed corresponds to the worst case).

\subsection{Numerics and conclusions}

We now describe the numerical investigation of $D(\{\lambda_j\}_{j=1}^d)$. Since the formula \eqref{eq:exactbound} involves addition of large positive and negative numbers, which is susceptible to numerical error, we instead calculate it by sampling the integral over the Haar measure.

We consider four scenarios. In the first, there is only one nonzero eigenvalue $\lambda_{\mathrm{one}}$. We compare this case to the exact formula of Eq.~\eqref{eq:allbutonebound} as a check on the integration procedure. Second, we suppose the eigenvalues are equally spaced. Third, we suppose the eigenvalues are entirely degenerate, i.e., all equal to $\Lambda$. Finally, we compute the maximum $D(\{\lambda_j\}_{j=1}^d)$ over $1000$ trials, drawing $\{\lambda_j\}_{j=1}^d$ randomly from the exponential distribution, with probability density function $p(x) = e^{-x}$. This produces eigenvalue distributions that are complementary to the three other scenarios.

As seen in Figure~\ref{fig:CvNC}, our numerics agree with Eq.~\eqref{eq:allbutonebound} in the case of a single nonzero eigenvalue, up to statistical error due to sampling the Haar integral. When the spectrum is entirely degenerate, we see numerically that $D = 0$, so we omit these results from the figure. Both the evenly spaced and randomly sampled scenarios have $D$
less than in the case of a single nonzero eigenvalue. Based on these numerics, we conjecture that the distribution $\lambda_\text{one}$ maximizes $D(\{\lambda_j\}_{j=1}^d)$. We leave the proof of this conjecture as an open question.

We summarize the main points of this analysis, and the consequences of this conjecture. Consider a fixed-resource scenario in which we perform a product formula simulation with $r$ steps and wish to bound the total Haar-averaged Trotter error. If $r\sqrt{\|G\|}$ is negligible relative to the error resolution, then we gain little by using Eq.~\eqref{eq:exactbound} (since $D < 1$).  If $r\sqrt{\|G\|}$ is non-negligible, then assuming $\lambda_\text{one}$ corresponds to the worst case, the gap between the two bounds can be at most
\begin{equation}\label{eq:CNCworstcase}
    r \sqrt{\|G\|}\biggl(\frac{1}{\sqrt{d}} - \frac{\sqrt{\pi}}{2} \frac{\Gamma(d)}{\Gamma(d+\frac{1}{2})} \biggr ).
\end{equation}
Finally, consider when $G = \mathscr{M}\mathscr{M}^\dagger$, where $\mathscr{M}$ is
the multiplicative error between two unitaries. Again, the gap between the trace bound on the average error and the true worst-case error $\sqrt{\|\mathscr{M}\mathscr{M}^\dagger \|}$ is maximized for the distribution $\lambda_\text{one}.$ Therefore the truth of the conjecture would imply that the gap between the true worst-case error and the true average-case error occurs for the distribution $\lambda_\text{one}.$

\begin{figure}
    \centering
    \includegraphics{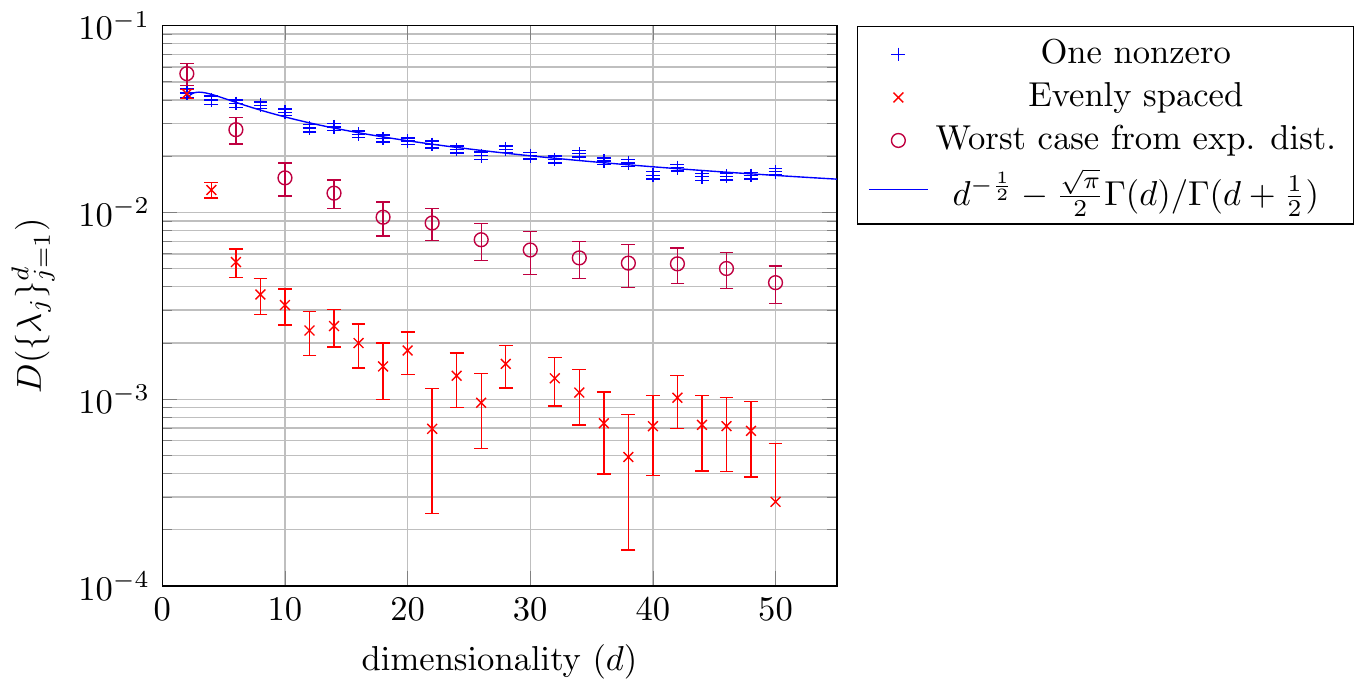}
    \caption{Estimating the difference between the error bound derived using the Cauchy inequality \eqref{eq:cauchybound} and the exact formula \eqref{eq:exactbound}. When $\{\lambda_j\}_{j=1}^d$ is drawn from the exponential distribution, we plot the maximum $D(\{\lambda_j\}_{j=1}^d)$ over 1000 $\{\lambda_j\}_{j=1}^d$, where $D$ is as defined in \eqref{eq:difdefine}.}
    \label{fig:CvNC}
\end{figure}

\subsection{Average error in PF1 without the Cauchy-Schwarz inequality}

We conclude this section with an example, by deriving a bound without using the Cauchy-Schwarz inequality for $R_{\ell_2}$. For simplicity, we only consider the first-order product formula.

It is easier to give a bound by considering the additive error $\mathscr{A} := U - U_0$ instead of the multiplicative error. An operator representation of the additive error for a two-term Hamiltonian $H=A+B$ for PF1 is \cite[Eq.~(117)]{childs2020theory}
\begin{equation}
  \mathscr{A}(t) = \int_0^t \d\tau_1  \int_0^{\tau_1} \d\tau_2~e^{iH\tau_1}e^{-iB(\tau_1-\tau_2)}
[iB,iA] e^{-iB\tau_2}e^{-iA\tau_1}.
\end{equation}
Note that the spectrum of this operator representation, like that of the corresponding multiplicative error, is not straightforward to analyze (relative to, say, that of $[A,B]$) due to the integrals and the unitaries on either side of $[iB,iA]$. However, the operator $G = \frac{t^4}{4}[A,iB]^2$ can be used to more readily bound the average-case error. Using the triangle inequality, Fubini's theorem, and the left-invariance of the Haar measure (in that order), we find
\begin{align}
    R_{\ell_2}(\mathscr{U}_1(t),U_0(t)) & = \mathbb{E}_{\psi}  \| \mathscr{A}(t) \ket{\psi} \| \nonumber \\
    &= \int \d\psi \left \|\int_0^t \d\tau_1 \int_0^{\tau_1} \d\tau_2 \  e^{-i(t-\tau_1)H}e^{-i\tau_1 B}e^{i\tau_2 B}[iB,iA]e^{-i\tau_2 B} e^{-i\tau_1 A}\ket{\psi} \right \| \ \nonumber  \\
    &\leq \int \d\psi \int_0^t \d\tau_1 \int_0^{\tau_1} \d\tau_2 \left \| e^{-i(t-\tau_1)H}e^{-i\tau_1 B}e^{i\tau_2 B}[iB,iA]e^{-i\tau_2 B} e^{-i\tau_1 A}\ket{\psi} \right \| \ \nonumber  \\
    &= \int \d\psi \int_0^t \d\tau_1 \int_0^{\tau_1} \d\tau_2 \left \|[iB,A]e^{-i\tau_2 B} e^{-i\tau_1 A}\ket{\psi} \right \| \ \nonumber  \\
    &=  \int_0^t \d\tau_1 \int_0^{\tau_1} \d\tau_2 \int \d\psi \left \|[iB,A]e^{-i\tau_2 B} e^{-i\tau_1 A}\ket{\psi} \right \| \  \quad \text{(Fubini's Theorem)} \nonumber \\
    &= \int_0^t \d\tau_1 \int_0^{\tau_1} \d\tau_2 \int \d\psi \left \|[iB,A]\ket{\psi} \right \|\  \quad \text{(Left-Invariance of Haar Measure)} \nonumber  \\
    &= \frac{t^2}{2} \int \d\psi \left \|[iB,A]\ket{\psi} \right \| =  \mathbb{E}_{\psi}  \sqrt{\bra{\psi} \biggl( \frac{t^4}{4} [iB,A]^2 \biggr) \ket{\psi}}.
    \label{eq:harrintegrandAB}
\end{align}
Note that unlike in Appendix~\ref{Sec:Average}, we do not apply the Cauchy-Schwarz inequality above. Using this bound, one can apply Eq.~\eqref{eq:exactbound} or sample the integral to get an in-principle tighter bound on $R_{\ell_2}(\mathscr{U}_1(t),U_0(t)),$ compared with Eq.~\eqref{eq:ABtrbound}.

This bound may be generalized to higher-order formulas and Hamiltonians with more summands, but here we derive it only for the purposes of the example. The example model is the one-dimensional Heisenberg model, using the even-odd ordering, as discussed in the following section.

For a simulation with total evolution time $n$ (where $n$ is also the number of qubits on the lattice) and $r$ Trotter steps, we consider the difference term in Eq.~\eqref{eq:totsimerrnc}, which in this case is
\begin{equation}
    r\delta_2 - r\delta_{1} = \frac{n^2}{2r} \bigl\| [iB,A] \bigr \| \ D(\lambda_{[iB,A]^2}).
\end{equation}
We consider this difference relative to the triangle bounds of Figure~\ref{Fig:PF12}. Numerically, we see that $ \| [iB,A]  \| \ \lesssim 10^2 $ for $n\approx 10$.  However, our simulation involves $r \gtrsim 10^5$ Trotter steps. Since $D < 10^{-2}$ (again, assuming $\lambda_\text{one}$ corresponds to the maximum), we see that this difference is bounded by $\approx 10^{-3}/2.$ As the target error is $10^{-3},$ we do not expect to see a significant difference in the bounds with and without the Cauchy-Schwartz inequality. Indeed, we observe negligible difference between the two bounds, so we do not plot the results; the two bounds agree, as in Figure~\ref{Fig:PF12}, up to about $1\%$.

The difference in these bounds would be more relevant for models and simulations in which the multiplicative error has a larger norm, or a norm that scales poorly with the system size. For example, this may occur when the Hamiltonian being simulated is unbounded.

\section{Applications}\label{Sec:app}

Recall that in Appendix~\ref{Sec:generalp} and Appendix~\ref{Sec:PF12}, we show that the average error for PF$p$ is $R=\mathcal O(T_p t^{p+1}/r^p)$, where the quantity $T_p$ defined in Eq.~\eqref{eq:Tp} depends on the Frobenius norms of nested commutators of terms in the Hamiltonian. In this section, we estimate $T_p$ for nearest-neighbor Hamiltonians, general $k$-local Hamiltonians, and power-law interaction Hamiltonians, and thereby evaluate the performance of average-case Hamiltonian simulation algorithms for these systems.
We also describe an application of our techniques to more efficient protocols for measuring out-of-time-order correlators (OTOCs), which provide a probe of scrambling physics.

\subsection{Nearest-neighbor Hamiltonians}\label{Sec:app_nearest}

Consider a nearest-neighbor Hamiltonian $H=\sum_{l=1}^{n-1}H_{l,l+1}$ with $H_{l,l+1}=\sum_i\alpha^i_l \, \sigma_{l}^i\otimes\sigma_{l+1}^i$ acting nontrivially on nearest-neighbor qubits $l$ and $l+1$, with $\max_{l}\|H_{l,l+1}\|\le 1$.
We consider the first-order product formula with summands ordered in an even-odd pattern, bipartitioned as
\begin{equation}
    A=\sum_{j=1}^{\lfloor\frac{n}{2}\rfloor} H_{2j-1,2j},~B=\sum_{j=1}^{\lfloor\frac{n}{2}\rfloor} H_{2j,2j+1}.
\end{equation}
The commutator of these two terms is
\begin{equation}
\begin{aligned}
~[iB,iA]=[A,B]&=\sum_{j,j'}[H_{2j-1,2j},H_{2j',2j'+1}]=\sum_{j} [H_{2j-1,2j},H_{2j,2j+1}]+ [H_{2j-1,2j},H_{2j-2,2j-1}]\\
    &=\sum_{j} \left(H_{\{2j-1,2j,2j+1\}}+H_{\{2j-2,2j-1,2j\}} \right).
\end{aligned}
\end{equation}
Here $[H_{2j-1,2j},H_{2j',2j'+1}]$ can be nonzero only when their supports overlap (i.e., $j=j'$ or $j=j'+1$) and we let $H_{\{i,j,k\}}$ denote the operator whose support is the subset of qubits $\{i,j,k\}$. Each term has the upper bound $\|H_{\{2j-1,2j,2j+1\}}\|\le 2$.

To upper bound the Frobenius norm of the commutator, we compute
\begin{equation}
\tr(|[iB,iA]|^2)=\tr([A,B][A,B]^{\dagger} )  =\sum_{j,j'}\sum_{k,k'=0,1} \tr\left(H_{\{2j+k-2,2j+k-1,2j+k\}}H_{\{2j'+k'-2,2j'+k'-1,2j+k'\}}^{\dagger}\right).
\end{equation}
For each pair of terms $H_{\{2j+k-2,2j+k-1,2j+k\}}$ and $H_{\{2j'+k'-2,2j'+k'-1,2j+k'\}}^{\dagger}$, the trace of the product can be nonzero only when their supports overlap, since the trace of a Pauli matrix is zero.
Thus for each $H_{\{2j+k-2,2j+k-1,2j+k\}}$, there are at most five different
$H_{\{2j'+k'-2,2j'+k'-1,2j+k'\}}^{\dagger}$ satisfying $\tr\bigl(H_{\{2j+k-2,2j+k-1,2j+k\}}H_{\{2j'+k'-2,2j'+k'-1,2j+k'\}}^{\dagger}\bigr)\neq0$.
Consequently, we have
\begin{equation}
\begin{aligned}
 \tr([iB,iA]^2)\le 2\cdot 5 d\sum_{j}\sum_{k=0,1} \left\|H_{\{2j+k-2,2j+k-1,2j+k\}}\right\|\le 40dn=\mathcal O(dn),
\end{aligned}
\end{equation}
where $d=2^n$.
Thus for PF1, we have
$T_1=\mathcal O(\sqrt{n})$ and $ R_{\ell_2}(\mathscr{U}_1^r(t/r),U_0(t))= \mathcal{O}(\frac{\sqrt{n}t^2}{r})$.
For comparison, the worst-case error is
$\mathcal O(\frac{nt^2}{r})$.

To understand the performance of the PF$p$ method for general $p$, we can similarly bound $T_p$ as follows. Consider a $p$-layer nested commutator,
$[X_1,[X_2,\dots,[X_p,X_{p+1}]]]$, where each $X_j$ can be $A$ or $B$. Without loss of generality, we choose $X_{p+1}=B$.
Then the commutator can be expressed as
\begin{equation}
  [X_p,B]=\sum_{j=1}^{\lfloor\frac{n}{2}\rfloor} [X_p, H_{2j,2j+1}] = \sum_{j=1}^{\lfloor\frac{n}{2}\rfloor}
  H_{\{2j-1,\dots,2j+2\}},
\end{equation}
where $H_{\{2j-1,\dots,2j+2\}}$ acts only on qubits from $2j-1$ to $2j+2$ (at most 4 qubits). Its spectral norm satisfies $\|H_{\{2j-1,2j+2\}}\|\le 2\cdot 3=6$ because $X_p$ at most has three operators ($H_{2j-1,2j}, H_{2j,2j+1}, H_{2j+1,2j+2}$) that could have overlap with $H_{2j,\dots,2j+1}$.

We repeat this procedure for all $p$ layers of the commutator.
In each step, $[X, H_{\{j_1,\dots, j_2\}}]$ can only expand the support of $H_{\{j_1,\dots, j_2\}}$ to $H_{\{j_1-1,\dots, j_2+1\}}$, increasing by at most two qubits. Furthermore, $H_{\{j_1,\dots, j_2\}}$ overlaps with at most $j_2-j_1+2$ operators in $X$ (from $H_{\{j_1-1, j_1\}}$ to $H_{\{j_2, j_2+1\}}$), and its spectral norm has the upper bound $\|H_{\{j_1-1,\dots, j_2+1\}}\|\le 2 \|H_{\{j_1,\dots, j_2\}}\|\cdot (j_2-j_1+2)$.
Therefore
\begin{equation}
 [X_1,[X_2,\dots,[X_p,X_{p+1}]]]=\sum_{j=1}^{\lfloor\frac{n}{2}\rfloor} H_{\{2j-p,\dots,2j+1+p\}},
\end{equation}
where each term $H_{\{2j-p,\dots,2j+1+p\}}$ acts on at most $2p+2$ qubits and $\|H_{\{2j-p,\dots,2j+1+p\}}\|\le
2(2p+1) \|H_{\{2j-p+1,2j+p\}}\|< 4(p+1) \|H_{\{2j-p+1,2j+p\}}\|<
4^p (p+1)!$.

Next we consider the Frobenius norm $\tr (|[X_1,[X_2,\dots,[X_p,X_{p+1}]]]|^2)$.
We have
\begin{equation}
|[X_1,[X_2,\dots,[X_p,X_{p+1}]]]|^2=\sum_{j,j'=1}^{\lfloor\frac{n}{2}\rfloor} H_{\{2j-p,\dots,2j+1+p\}} H_{\{2j'-p,\dots,2j'+1+p\}} ^{\dagger}.
\end{equation}
Each $H_{\{2j-p,\dots,2j+1+p\}}$ overlaps the $4p+3$
terms $H_{\{2j-3p-1,\dots,2j-p\}} ^{\dagger}, H_{\{2j-3p,\dots,2j-p+1\}} ^{\dagger},\dots, H_{\{2j+p+1,\dots,2j+3p+2\}} ^{\dagger}$.
Since $\|H_{\{2j-p,\dots,2j+1+p\}}\|<
4^p (p+1)!$, we have
\begin{equation}
\left\|[X_1,[X_2,\dots,[X_p,X_{p+1}]]]\right\|_F^2=\tr \bigl(|[X_1,[X_2,\dots,[X_p,X_{p+1}]]]|^2\bigr) \le dn(4p+3) 4^{2p} (p+1!)^2=\mathcal O(dn),
\end{equation}
where the last step follows since we treat $p$ as a constant.

Overall, we find that
\begin{equation}
T_p=\sum_{l_1,\dots,l_{p+1}=1}^L \frac{1}{\sqrt{d}} \left\|[H_{l_1},[H_{l_2},\dots,[H_{l_p},H_{l_{p+1}}]]]\right\|_F
\le 2^{p+1} \mathcal O(\sqrt{n})= \mathcal O(\sqrt{n}),
\end{equation}
so the PF$p$ method has average error
\begin{equation}
\label{Eq:nearPerror}
R_{\ell_2}(\mathscr{U}_p^r(t/r),U_0(t))=\mathcal O\left( \sqrt{n}\frac{ t^{p+1}}{r^p}\right)
\end{equation}
and corresponding gate complexity
\begin{equation}
\label{Eq:nearPgates}
G_R(\mathscr{U}_p^r(t/r),U_0(t))=\mathcal O\left( n^{1+\frac{1}{2p}}t^{1+\frac{1}{p}}\right).
\end{equation}
For comparison, from Ref.~\cite{childs2020theory}, the worst-case error is
\begin{equation}
W(\mathscr{U}_p^r(t/r),U_0(t))=\mathcal O\left( n\frac{ t^{p+1}}{r^p}\right)
\end{equation}
and the corresponding gate complexity is
\begin{equation}
G_W(\mathscr{U}_p^r(t/r),U_0(t))=
\mathcal O\left( n^{1+\frac{1}{p}}t^{1+\frac{1}{p}}\right).
\end{equation}
In particular, for PF1 and PF2, we have the gate complexity results shown in Table~\ref{table:nn}.

\begin{table}[htb]
\begin{tabular}{|l|l|l|}
\hline
Order & Worst-case error & Average error \\ \hline
$p=1$   &   $\mathcal O\left(n^2t\right) $        &    $ \mathcal O\left(n^{1.5} t\right)  $         \\ \hline
$p=2 $  & $\mathcal O\left(n^{2.5}t^{1.5}\right) $        &    $ \mathcal O\left(n^{2.25} t^{1.5}\right)  $       \\ \hline
\end{tabular}
\caption{Gate complexity of PF1 and PF2 for nearest-neighbor Hamiltonians.\label{table:nn}}
\end{table}

\subsection{\texorpdfstring{$k$}{k}-local Hamiltonians}

Now consider a $k$-local Hamiltonian
\begin{equation}
    H=\sum_{j_1,j_2,\dots,j_k} H_{j_1,j_2,\dots,j_k},
\end{equation}
where $H_{j_1,j_2,\dots,j_k}$ acts nontrivially on qubits $j_1,j_2,\dots,j_k$.
For simplicity, we define the tuples $\Vec{j} := (j_1,j_2,\dots,j_k)$ and $\Vec{j}\setminus l := (j_1,\dots,j_{l-1},j_{l+1},\dots,j_k)$.
We let $S(A)$ denote the support of the operator $A$.
In this section, we first explore the average error of $k$-local Hamiltonians with PF1 and PF2 according to Theorem~\ref{Th:PF1} and~\ref{Th:PF2}. Then we consider the general PF$p$ results according to Theorem~\ref{Th:general}.

It is clear that $[H_{\Vec{i}},H_{\Vec{j}}]$ is nonzero only when there are positions $m,l$ such that $i_m=j_l$.
For a given $H_{\Vec{i}}$ and $H_{\Vec{j}}$, there are $k^2$ ways to choose $m$ and $l$ in $\Vec{i}$ and $\Vec{j}$. The trace of the product of two Pauli strings is nonzero only when they have the same support.
Thus $\tr\bigl( [H_{\Vec{i}},H_{\Vec{j}}][H_{\Vec{i}},H_{\Vec{j'}}]\bigr) $ is nonzero only when the support of $[H_{\Vec{i}},H_{\Vec{j}}]$ is the same as that of $ [H_{\Vec{i}},H_{\Vec{j'}}]$.
We can break the support of $[H_{\Vec{i}},H_{\Vec{j}}]$ into two parts,  $S(H_{\Vec{i}})$ and $S([H_{\Vec{i}},H_{\Vec{j}}])\setminus S(H_{\Vec{i}})$.
Suppose one qubit acting on the $l$th entry in $\Vec{j}$ overlaps $S(H_{\Vec{i}})$.
Then the set $S([H_{\Vec{i}},H_{\Vec{j}}])\setminus S(H_{\Vec{i}})$ contains at most $k-1$ qubits that do not appear in $S(H_{\Vec{i}})$. The set $S(H_{\Vec{j'}})$ has to contain these $k-1$ qubits to satisfy $S([H_{\Vec{i}},H_{\Vec{j}}])=S([H_{\Vec{i}},H_{\Vec{j'}}])$, so there must be some permutation $\pi$ that maps the $k-1$ relevant entries of $\Vec{j'}$
to $S([H_{\Vec{i}},H_{\Vec{j}}])\setminus S(H_{\Vec{i}})$, with $\Vec{j'}\setminus l'=\pi(\Vec{j}\setminus l)$.
Thus we have the upper bound
\begin{equation}
\begin{aligned}
\tr\bigl(|[H_{\Vec{i}},\sum_{\Vec{j}} H_{\Vec{j}}]|^2\bigr)
 &=\tr \biggl( \bigl|\sum_{\Vec{j}}[H_{\Vec{i}},H_{\Vec{j}}]\bigr|^2 \biggr)
 =\tr \left(\sum_{\Vec{j},\Vec{j'}}[H_{\Vec{i}},H_{\Vec{j}}][H_{\Vec{i}},H_{\Vec{j'}}]^{\dagger}\right)\\
 &\le 4\tr(H_{\Vec{i}}H_{\Vec{i}}^{\dagger})\sum_{\Vec{j}\colon S(H_{\Vec{j}})\cap S(H_{\Vec{i}})\ne \emptyset} \|H_{\Vec{j}}\|\sum_{\Vec{j'}\colon S(H_{\Vec{j'}})\cap S(H_{\Vec{i}})\ne \emptyset,\,S([H_{\Vec{i}},H_{\Vec{j}}])=S([H_{\Vec{i}},H_{\Vec{j'}}]) } \|H_{\Vec{j'}}\|\\
&\le 4 k^4  \tr(H_{\Vec{i}}H_{\Vec{i}}^{\dagger}) \max_{l,j_l} \sum_{\Vec{j}\setminus l}  \|H_{\Vec{j}}\|\left(\max_{l',j'_l}
\sum_{\Vec{j'}\setminus l'=\pi(\Vec{j}\setminus l)} \|H_{\Vec{j'}}\|\right), \\
\end{aligned}\label{eq:comtraceineq}
\end{equation}
where the first inequality is due to the second inequality in Lemma~\ref{Lemma:traceupper} below.
Here the factor of $k^4$ comes from the fact that there are at most $k^2$ choices for the overlapping qubit in $[H_{\Vec{i}},H_{\Vec{j}}]$ and $[H_{\Vec{i}},H_{\Vec{j'}}]$.

When more than one entry of $\Vec{j}$ overlaps $S(H_{\Vec{i}})$, a similar analysis holds. For example, suppose two entries $l_1$ and $l_2$ overlap. Similarly, we have the upper bound
\begin{equation}
\begin{aligned}
\tr\bigl(|[H_{\Vec{i}},\sum_{\Vec{j}} H_{\Vec{j}}]|^2\bigr)
&\le 4k^8 \tr(H_{\Vec{i}}H_{\Vec{i}}^{\dagger}) \max_{l_1,l_2,j_{l_1},j_{l_2}} \sum_{\Vec{j}\setminus l_1,l_2}  \|H_{\Vec{j}}\|\left(\max_{l_1',l_2',j_{l_1}',j_{l_2}'}
\sum_{\Vec{j'}\setminus l_1',l_2'=\pi(\Vec{j}\setminus l_1,l_2)} \|H_{\Vec{j'}}\|\right). \\
\end{aligned}
\end{equation}
Then due to the inclusion relationship, we have
\begin{equation}
\max_{l_1,l_2,j_{l_1},j_{l_2}} \sum_{\Vec{j}\setminus l_1,l_2} \|H_{\Vec{j}}\| \le \max_{l,j_l} \sum_{\Vec{j}\setminus l}  \|H_{\Vec{j}}\|,
\end{equation}
and similarly
\begin{equation}
\sum_{\Vec{j'}\setminus l_1',l_2'=\pi(\Vec{j}\setminus l_1,l_2)} \|H_{\Vec{j'}}\| \le \sum_{\Vec{j'}\setminus l'=\pi(\Vec{j}\setminus l)} \|H_{\Vec{j'}}\|.
\end{equation}
Other cases with multi-qubit overlap are analogous.  Since we consider $k=\mathcal O(1)$, it is sufficient to only explore single-qubit overlap to bound the asymptotic error.

For simplicity, we introduce norms $\|H\|_{F}$, $\normH{H}^2_{\mathrm{per}}$ that appear in the above inequality and recall the norms $\|H\|_1$, $\normH{H}_1$ introduced in Ref.~\cite{childs2020theory}:
\begin{equation}\label{eq:perNorm}
\begin{aligned}
\|H\|_1&:=\sum_{\Vec{i}}\|H_{\Vec{i}}\|,\\
\|H\|_{1, F}&:= \sum_{\Vec{i}} \|H_{\Vec{i}}\|_{F}=\sum_{\Vec{i}}\sqrt{\tr(H_{\Vec{i}} H_{\Vec{i}}^{\dagger})}\le \sqrt{d}\|H\|,\\
\normH{H}_1&:=\max_{l,j_l} \sum_{\Vec{j}\setminus l}  \|H_{\Vec{j}}\|,\\
\normH{H}^2_{\mathrm{per}} &:= \max_{l,j_l} \sum_{\Vec{j}\setminus l}  \|H_{\Vec{j}}\|\left(\max_{l',j'_l}
\sum_{\Vec{j'}\setminus l'=\pi(\Vec{j}\setminus l)} \|H_{\Vec{j'}}\|\right).\\
\end{aligned}
\end{equation}
Using these norms, the inequality \eqref{eq:comtraceineq} can be rewritten as
\begin{equation}
\begin{aligned}\label{Eq:tracepermu}
\tr\bigl([H_{\Vec{i}},\sum_{\Vec{j}} H_{\Vec{j}}]^2\bigr)\le 4k^4\|H_{\Vec{i}}\|_F^2
\normH{H}^2_{\mathrm{per}}.
\end{aligned}
\end{equation}
In the following, we describe our average error analysis in terms of $\|H\|_{1, F}$ and $\normH{H}^2_{\mathrm{per}}$ and show that our improvement stems from the differences between $\frac{1}{\sqrt{d}}\|H\|_{1,F}$, $\|H\|$, $\normH{H}_1$, and $\normH{H}_{\mathrm{per}}$.

For the PF1 method, recalling the result in Theorem \ref{Th:PF1}, the average error has the upper bound
\begin{equation}
\begin{aligned}
R_{\ell_2}(\mathscr{U}^r_1(t/r),U_0(t))
&\le \frac{t^2}{2r}\sum_{l_1=1}^{L-1}\frac{1}{\sqrt{d}}\|[H_{l_1},\sum_{l_2=l_1+1}^{L}H_{l_2}]\|_F \\
&\le \frac{t^2 k^2}{r \sqrt{d}} \sum_{l_1=1}^{L-1}\sqrt{\tr(H_{l_1} H_{l_1}^{\dagger})} ~\normH{H}_{\mathrm{per}}\\
&= \frac{t^2 k^2}{r \sqrt{d}} \|H\|_{1,F}\,\normH{H}_{\mathrm{per}} \\
&=\mathcal O\Big( \frac{1}{\sqrt{d}}\|H\|_{1,F}\,\normH{H}_{\mathrm{per}}\frac{t^2}{r}\Big),
\end{aligned}
\end{equation}
where $d$ is the dimension of the Hilbert space acted on by the Hamiltonian $H$.

Next we consider the PF2 method and recall the result of Theorem~\ref{Th:PF2}:
  \begin{equation}
  \begin{aligned}
 & R_{\ell_2}(\mathscr{U}^r_2(t/r),U_0(t))\\
 & \le     \frac{t^3}{12r^2} \sum_{l_1=1}^{L}
\left(\frac{\tr\bigl(|[\sum_{l_2=l_1+1}^L H_{l_2},[\sum_{l_2=l_1+1}^L H_{l_2},H_{l_1}]]|^2\bigr)}{d}\right)^{\frac{1}{2}} +  \frac{t^3}{24r^2}  \sum_{l_1=1}^{L}\left(\frac{\tr\bigl(|[H_{l_1},[H_{l_1},\sum_{l_2=l_1+1}^L H_{l_2}]]|^2\bigr)}{d}\right)^{\frac{1}{2}}.
  \end{aligned}\label{eq:pf2klocal}
\end{equation}
The first term in the above inequality can be upper bounded as
\begin{equation}
\begin{aligned}
\tr\Bigl(|[\sum_{l_2=l_1+1}^L H_{l_2},[\sum_{l_2=l_1+1}^L H_{l_2},H_{l_1}]]|^2\Bigr)
&=\sum_{l_2=l_1+1}^L \sum_{l_2'=l_1+1}^L\tr\left( \Bigl[H_{l_2},\bigl[\sum_{l_2=l_1+1}^L H_{l_2},H_{l_1}\bigr]\Bigr]\Bigl[H_{l_2'},\bigl[\sum_{l_2=l_1+1}^L H_{l_2},H_{l_1}\bigr]\Bigr]^{\dagger}   \right)\\
&\le 4\|[\sum_{l_2=l_1+1}^L H_{l_2},H_{l_1}]\|^2_F \left(\sum_{l_2=l_1+1, ~S(H_{l_2})\cap S(H_{l_1})\ne \emptyset}^L\|H_{l_2}\|\right)^2\\
&\le 4k^4\|[\sum_{l_2=l_1+1}^L H_{l_2},H_{l_1}]\|^2_F  \normH{H}_1^2 \\
&\le 16k^8 \|H_{l_1}\|_F^2 \normH{H}_1^2  \normH{H}_{\mathrm{per}}^2.
\end{aligned}
\end{equation}
Here the first inequality follows from multiple applications of the second inequality in Lemma~\ref{Lemma:traceupper} by taking $X=[ \sum_{l_2=l_1+1}^L H_{l_2},H_{l_1}]$, $Y=H_{l_2}$, and $Z=H_{l_2'}$. The second inequality is due to the definition of the induced 1-norm in Eq.~\eqref{eq:perNorm}.
The third inequality is due to Eq.~\eqref{Eq:tracepermu}.

Similarly, the second term in \eqref{eq:pf2klocal} has the upper bound
\begin{equation}
\begin{aligned}
\tr\bigl(|[H_{l_1},[H_{l_1},\sum_{l_2=l_1+1}^L H_{l_2}]]|^2\bigr)
&\le 4\|H_{l_1}\|^2\tr\bigl(|[[H_{l_1},\sum_{l_2=l_1+1}^L H_{l_2}]|^2\bigr)\\
&\le  16  k^4\|H_{l_1}\|^2\tr(H_{l_1}H_{l_1}^{\dagger})
 \max_{l,j_l} \sum_{\Vec{j}\setminus l}  \|H_{\Vec{j}}\|\left(\max_{l',j'_l}
\sum_{\Vec{j'}\setminus l'=\pi(\Vec{j}\setminus l)} \|H_{\Vec{j'}}\|\right) \\
&= 16k^4\|H_{l_1}\|^2\|H_{l_1}\|_F^2 \normH{H}_{\mathrm{per}}^2.
\end{aligned}
\end{equation}
Here the first inequality is due to the first inequality in Lemma~\ref{Lemma:traceupper} with $X=H_{l_1}$ and $Y=[H_{l_1},\sum_{l_2=l_1+1}^L H_{l_2}]$. The second inequality also applies Eq.~\eqref{Eq:tracepermu}.

Adding all ${l_1}$ terms, we find
\begin{align}
\sum_{l_1=1}^{L}
\left(\frac{\tr\bigl(|[\sum_{l_2=l_1+1}^L H_{l_2},[\sum_{l_2=l_1+1}^L H_{l_2},H_{l_1}]]|^2\bigr)}{d}\right)^{\frac{1}{2}}
&\le 4k^4\frac{1}{\sqrt{d}}\|H\|_{1, F} \normH{H}_1\normH{H}_{\mathrm{per}}=\mathcal O\left(\frac{1}{\sqrt{d}}\|H\|_{1, F}  \normH{H}_1\normH{H}_{\mathrm{per}}\right); \\
\sum_{l_1=1}^{L}\left(\frac{\tr\bigl(|[H_{l_1},[H_{l_1},\sum_{l_2=l_1+1}^L H_{l_2}]]|^2\bigr)}{d}\right)^{\frac{1}{2}}
&\le 4k^2 \frac{1}{\sqrt{d}} \left(\sum_{l_1=1}^{L}\|H_{l_1}\|\|H_{l_1}\|_{F} \right) \normH{H}_{\mathrm{per}} \nonumber\\
&= \mathcal O\left( \frac{1}{\sqrt{d}}\|H\|_{1, F} \normH{H}_{\mathrm{per}}\max_l\|H_{l}\| \right).
\end{align}
Since $\max_l \|H_{l}\|\le \normH{H}_1$, we have
\begin{equation}
R_{\ell_2}(\mathscr{U}_2^r(t/r),U_0(t))=\mathcal O\left(\frac{1}{\sqrt{d}}\|H\|_{1,F} \normH{H}_1\normH{H}_{\mathrm{per}}\frac{t^3}{r^2}\right).
\end{equation}

Finally, we consider general PF$p$ methods, with average error $R_{\ell_2}= \mathcal O\left(T_p{t^{p+1}}/{r^{p}}\right)$ where
\begin{equation}
T_p:=\sum_{l_1,\dots,l_{p+1}=1}^L
\frac{1}{\sqrt{d}} \left\|[H_{l_1},[H_{l_2},\dots,[H_{l_p},H_{l_{p+1}}]]] \right\|_F.
\end{equation}

The PF1 and PF2 analysis in Theorems~\ref{Th:PF1} and \ref{Th:PF2} includes a Frobenius norm of a sum of commutators that leads to $\normH{H}_{\mathrm{per}}$. However, in the higher-order cases described in Theorem~\ref{Th:general}, all the sums are outside of the Frobenius norm. Thus we give a different analysis that does not introduce $\normH{H}_{\mathrm{per}}$.

First we consider the one-layer commutator
\begin{equation}
 \sum_{l_1,l_2}
\frac{1}{\sqrt{d}}\left(\tr(|[H_{l_1},H_{l_{2}}]|^2\right)^{\frac{1}{2}}\le 2k^2\frac{1}{\sqrt{d}}\normH{H}_1  \sum_{l_2}\sqrt{\tr(H_{l_{2}} H_{l_{2}}^{\dagger})}=\mathcal O(\frac{1}{\sqrt{d}}\|H\|_{1, F}\normH{H}_1 ).
\end{equation}
Then we prove the general result by induction, assuming that for up to order $p-1$,
\begin{align}
\sum_{l_1,\dots,l_{p}=1}^L\frac{1}{\sqrt{d}}\left(\tr(|[H_{l_1},[H_{l_2},\dots,[H_{l_{p-1}},H_{l_{p}}]|^2)\right)^{\frac{1}{2}}=\mathcal O(\frac{1}{\sqrt{d}}\|H\|_{1,F} \normH{H}_1^{p-1}) .
\end{align}
Noting that the operator $[H_{l_1},\dots,[H_{l_p},H_{l_{p}}]]$ is supported on at most $k+p(k-1)$ qubits, we have
\begin{equation}
  \begin{aligned}
&\frac{1}{\sqrt{d}} \sum_{l_1,\dots,l_{p+1}=1}^L\left(\tr(|[H_{l_1},[H_{l_2},\dots,[H_{l_p},H_{l_{p+1}}]|^2)\right)^{\frac{1}{2}}\\
&\quad\le
2 k(k+p(k-1)) \normH{H}_1 \sum_{l_2,\dots,l_{p+1}=1}^L\frac{1}{\sqrt{d}}\left(\tr(|[H_{l_2},\dots,[H_{l_p},H_{l_{p+1}}]|^2)\right)^{\frac{1}{2}} \\
&\quad=\mathcal O( \frac{1}{\sqrt{d}}\|H\|_{1,F}\normH{H}_1^{p} ).
\end{aligned}
\end{equation}

We summarize the average errors for PF1, PF2, and PF$p$ in Table~\ref{table:klocal}. For comparison, we also list the corresponding worst-case errors for $k$-local Hamiltonians.
This table shows that for $p > 2$, the improvement only comes from the difference between $\|H\|_{1, F}$ and $\|H\|_{1}$. If each term $H_l$ is a tensor product of Pauli operators, then $H_lH_l^{\dagger}=I$. In this case we have $\frac{1}{\sqrt{d}}\|H_l\|_{F}=\|H_l\|$, so $\frac{1}{\sqrt{d}}\|H\|_{1,F}=\|H\|_1$.
However, if the $H_l$s are sums of many Pauli operators, there can be a significant difference. For example, suppose
\begin{equation}
H_1= \sum_{i,j} X_iX_j,\
H_2= \sum_{i,j} Y_iY_j,\
H_3= \sum_{i,j} Z_iZ_j.
\end{equation}
Then we have $\|H\|_1=3n^2$, $\frac{1}{\sqrt{d}}\|H\|_{1, F}=3n$.

\begin{table}[htb]
\begin{tabular}{|l|l|l|}
\hline
Order & Worst-case error & Average error \\ \hline
$p=1$  &   $\mathcal O\left(\frac{t^2}{r}\|H\|_1 \normH{H}_1\right) $        &    $ \mathcal O\left(\frac{t^2}{r}\frac{1}{\sqrt{d}}\|H\|_{1, F} \normH{H}_{\mathrm{per}}\right)  $         \\ \hline
$p=2$   &     $\mathcal O\left(\frac{t^3}{r^2}\|H\|_1 \normH{H}_1^2\right)$        &       $\mathcal O\left(\frac{t^3}{r^2}\frac{1}{\sqrt{d}}\|H\|_{1, F}\normH{H}_1 \normH{H}_{\mathrm{per}}\right)$        \\ \hline
$p>2$ &      $\mathcal O\left(\frac{t^{p+1}}{r^p}\|H\|_1 \normH{H}_1^p\right)$    &  $\mathcal O\left(\frac{t^{p+1}}{r^p}\frac{1}{\sqrt{d}}\|H\|_{1, F}\normH{H}_1^p\right)$ \\
\hline
\end{tabular}
\caption{Error scaling for $k$-local Hamiltonians. }\label{table:klocal}
\end{table}

To further articulate our results, we consider a simple $k$-local Hamiltonian $H=\sum_{l_1,\dots,l_k}H_{l_1,\dots,l_k}$ acting on $n$ qubits, where each term has norm $\|H_{l_1,\dots,l_k}\|= 1$. Then we have $\|H\|\le n^k$,
$\frac{1}{\sqrt{d}}\|H\|_{1, F}\le n^k$,
$\normH{H}_{\mathrm{per}}=\mathcal O(n^{\frac{k-1}{2}})$, and $\normH{H}_1=\mathcal O(n^{k-1})$. In this case there is no difference between the upper bounds for $\|H\|_1$ and $\frac{1}{\sqrt{d}}\|H\|_{1, F}$. However, for the first- and second-order cases, we have the following comparison:
\begin{equation}
\begin{aligned}
&R_{\ell_2}(\mathscr{U}^r_1(t/r),U_0(t))=\mathcal O\left(\frac{t^2}{r}n^{\frac{3k-1}{2}}\right), ~W(\mathscr{U}^r_1(t/r),U_0(t))=\mathcal O\left(\frac{t^2}{r}n^{2k-1}\right);\\
&R_{\ell_2}(\mathscr{U}_2^r(t/r),U_0(t))=\mathcal O\left(\frac{t^3}{r^2}n^{\frac{5k-3}{2}}\right), ~W(\mathscr{U}_2^r(t/r),U_0(t))=\mathcal O\left(\frac{t^3}{r^2}n^{3k-2}\right).
\end{aligned}
\end{equation}
Here the improvement comes from the difference between $\normH{H}_1$ and $\normH{H}_{\mathrm{per}}$.

We conclude by showing the inequalities used in the above analysis.

\begin{lemma}\label{Lemma:traceupper}
For any complex matrices $X$, $Y$, and $Z$, we have
\begin{equation}
\begin{aligned}
 \tr |[X,Y]|^2 &\le 4\|X\|^2\tr(YY^{\dagger}), \\
 |\tr([X,Y][X,Z]^{\dagger})| &\le 4\|Y\|\|Z\|\tr(XX^{\dagger}).
\end{aligned}
\end{equation}
\end{lemma}

\begin{proof}
We have
\begin{align}
 \tr \bigl(|[X,Y]|^2\bigr)&= 2\tr (X^{\dagger}XYY^{\dagger})
-\tr(Y^{\dagger}X^{\dagger}YX)-   \tr(XYX^{\dagger}Y^{\dagger})\nonumber \\
&\le 2\|X\|^2\tr(YY^{\dagger})+ 2\sqrt{\tr(YY^{\dagger})} \sqrt{\tr( XX^{\dagger}YXX^{\dagger}Y^{\dagger})} \\
&\le
4\|X\|^2\tr(YY^{\dagger}),
\end{align}
and
\begin{align}
|\tr([X,Y][X,Z]^{\dagger})|&= |\tr (XY-YX)(Z^{\dagger}X^{\dagger}-X^{\dagger}Z^{\dagger})| \nonumber\\
&\le |\tr(XYZ^{\dagger}X^{\dagger})|+ |\tr(XYX^{\dagger}Z^{\dagger})|+|\tr(YXZ^{\dagger}X^{\dagger})|+|\tr(YXX^{\dagger}Z^{\dagger})| \nonumber\\
&\le 2\|Y\|\|Z\|\tr(XX^{\dagger})+ 2\sqrt{\tr(XX^{\dagger})} \sqrt{\tr( Y^{\dagger}YX^{\dagger}Z^{\dagger}ZX)} \\
&\le 4\|Y\|\|Z\|\tr(XX^{\dagger}).
\end{align}
The last inequality is due to Lemma~\ref{Lemma:traceproduct}.
\end{proof}

\subsection{Power-law interactions}\label{Sec:App_powerlaw}

Now we consider power-law interactions on a $D$-dimensional lattice $\Lambda \subset \mathbb{R}^D$, with a 2-body interaction of the form
\begin{equation}\label{eq:power}
H=\sum_{i,j} H_{i,j},\ \|H_{i,j}\| \le \left\{
\begin{aligned}
&1 ~&i=j  \\
&\frac{1}{\|i-j\|^{\alpha}} ~&i\neq j
\end{aligned}
\right.
\end{equation}
for some $\alpha>0$,
where $i,j$ are two sites on the lattice and $\|i-j\|$ is the Euclidean distance.

Such a Hamiltonian is 2-local. Using the definition of the norm $\normH{H}_{\mathrm{per}}$ in Eq.~\eqref{eq:perNorm}, we have
\begin{equation}
\begin{aligned}
\normH{H}^2_{\mathrm{per}}&=\max_{i}\sum_{j}  \|H_{i,j}\|\left(\max_{i'}
\|H_{i',j}\|\right)\le \max_{i}\sum_{j}  \|H_{i,j}\|=\normH{H}_1.
\end{aligned}
\end{equation}
Here the first equality follows from the symmetry of the labels $i,j$ in $H$ and the definition of the permutation norm.
The inequality holds because the interaction strength decreases with distance, and is upper bounded by 1 according to Eq.~\eqref{eq:power}.
The third equality follows from the definition of the induced 1-norm in Ref.~\cite{childs2020theory} and from Eq.~\eqref{eq:perNorm}.

Applying the results in Table~\ref{table:klocal}
and using $\frac{1}{\sqrt{d}}\|H\|_{1,F}\le \|H\|_1$, we have the following comparison to the worst-case performance of PF1 and PF2:
\begin{equation}\label{}
\begin{aligned}
&R_{\ell_2}(\mathscr{U}_1^r(t/r),U_0(t))=\mathcal O\left(\frac{t^2}{r}\|H\|_1~\normH{H}_1^{1/2}\right), ~W(\mathscr{U}_1^r(t/r),U_0(t))  =\mathcal O\left(\frac{t^2}{r}\|H\|_1\normH{H}_1\right);\\
&R_{\ell_2}(\mathscr{U}_2^r(t/r),U_0(t))=\mathcal O\left(\frac{t^3}{r^2}\|H\|_1~\normH{H}_1^{3/2}\right),~W(\mathscr{U}_2^r(t/r),U_0(t))  =\mathcal O\left(\frac{t^3}{r^2}\|H\|_1\normH{H}_1^2\right).\\
\end{aligned}
\end{equation}

Using Lemma F.1 of Ref.~\cite{childs2020theory}, we have
\begin{equation}
\normH{H}_1\leq \begin{cases}
\mathcal O (n^{1-\alpha/ D})\ &0\le \alpha<D,\\
\mathcal O (\log(n))\ &\alpha=D,\\
\mathcal O (1)\ &\alpha>D
\end{cases}
\end{equation}
and $\|H\|_1\leq n\normH{H}_1$.
Therefore, the average-case and the worst-case results are
\begin{equation}\label{}
\begin{aligned}
&R_{\ell_2}(\mathscr{U}_1^r(t/r),U_0(t))= \mathcal O\left(\frac{t^2}{r}n\normH{H}_1^{1.5}\right)=\left\{\begin{aligned}
&\mathcal O \left(\frac{t^2}{r}n^{\frac{5}{2}-\frac{3\alpha}{2 D}}\right)\ &0\le \alpha<D,\\
&\mathcal O \left(\frac{t^2}{r}n\log^{\frac{3}{2}}(n)\right)\ &\alpha=D,\\
&\mathcal O  \left(\frac{t^2}{r}n\right)\ &\alpha>D,\\
\end{aligned}
\right.\\
&W_{\ell_2}(\mathscr{U}_1^r(t/r),U_0(t))= \mathcal O\left(\frac{t^2}{r}n\normH{H}_1^{2}\right)=\left\{\begin{aligned}
&\mathcal O \left(\frac{t^2}{r}n^{3-\frac{2\alpha}{ D}}\right)\ &0\le \alpha<D,\\
&\mathcal O \left(\frac{t^2}{r}n\log^{2}(n)\right)\ &\alpha=D,\\
&\mathcal O  \left(\frac{t^2}{r}n\right)\ &\alpha>D,\\
\end{aligned}
\right.\\
\end{aligned}
\end{equation}
for PF1, and
\begin{equation}\label{}
\begin{aligned}
&R_{\ell_2}(\mathscr{U}_2^r(t/r),U_0(t))= \mathcal O\left(\frac{t^3}{r^2}n\normH{H}_1^{2.5}\right)=\left\{\begin{aligned}
&\mathcal O \left(\frac{t^3}{r^2}n^{\frac{7}{2}-\frac{5\alpha}{2 D}}\right)\ &0\le \alpha<D,\\
&\mathcal O \left(\frac{t^3}{r^2}n\log^{\frac{5}{2}}(n)\right)\ &\alpha=D,\\
&\mathcal O  \left(\frac{t^3}{r^2}n\right)\ &\alpha>D,\\
\end{aligned}
\right.\\
&W_{\ell_2}(\mathscr{U}_2^r(t/r),U_0(t))= \mathcal O\left(\frac{t^3}{r^2}n\normH{H}_1^{3}\right)=\left\{\begin{aligned}
&\mathcal O \left(\frac{t^3}{r^2}n^{4-\frac{3\alpha}{ D}}\right)\ &0\le \alpha<D,\\
&\mathcal O \left(\frac{t^3}{r^2}n\log^{3}(n)\right)\ &\alpha=D,\\
&\mathcal O  \left(\frac{t^3}{r^2}n\right)\ &\alpha>D,\\
\end{aligned}
\right.\\
\end{aligned}
\end{equation}
for PF2.

\subsection{Out-of-time-order correlators}\label{sec:OTOC}

The (infinite temperature) out-of-time-order correlator (OTOC) for
two commuting local observables $A$ and $B$ of a $d$-dimensional system is defined as
\begin{equation}\label{eq:otocDef}
\begin{aligned}
\langle O(t) \rangle:=    \langle B^{\dagger}(t) A(0)^{\dagger}  B(t) A(0) \rangle=\frac1{d}\tr (B^{\dagger}(t) A(0)^{\dagger}  B(t) A(0)),
\end{aligned}
\end{equation}
where $B(t) := e^{iHt}Be^{-iHt}$ is the operator $B$ in the Heisenberg picture. OTOCs are of interest since they can be used to characterize quantum information scrambling and quantum chaotic behavior in many-body systems \cite{shenker2014black,maldacena2016bound}.

For concreteness, consider an $n$-qubit system with local Pauli operators $A$ and $B$. In particular, suppose $A=Z_1\otimes \id_2^{\otimes (n-1)}$ and $B=\id_2^{\otimes (n-1)}\otimes X_n$ (acting nontrivially on the first and the last qubit, respectively).
The OTOC can be estimated using the quantum circuit shown in Fig.~\ref{fig:OTOCcircuit}, which has been demonstrated in an NMR experimental platform \cite{li2017measuring}.
First, the system is prepared in the initial state $\rho_{\mathrm{in}}=\phi\otimes \id_{d_1}/d_1$, where $\phi=\ket{0}\bra{0}$ is the state of the first qubit and the other $n-1 = \log_2 d_1$ qubits are maximally mixed.
Then the unitary evolution $V_0=e^{iHt}B^{\dag}e^{-iHt}$ is applied (note that $B=\id_2^{\otimes (n-1)}\otimes X_n$ is unitary). Finally, the observable $A=Z_1\otimes \id_2^{\otimes (n-1)}$ is measured on the first qubit. This procedure effectively measures the OTOC since
\begin{equation}\label{eq:otoc}
\begin{aligned}
\langle O(t) \rangle & =\frac1{d}\tr \left[V_0 (Z_1\otimes \id_2^{\otimes (n-1)}) V^{\dagger}_0 (Z_1\otimes \id_2^{\otimes (n-1)})\right]\\
& =\frac1{d}\tr \left[V_0 (d\rho_{\mathrm{in}}-\id_d) V^{\dagger}_0 (Z_1\otimes \id_2^{\otimes (n-1)})\right]\\
& =\tr [V_0 \rho_{\mathrm{in}}V^{\dagger}_0 Z_1],
\end{aligned}
\end{equation}
where $d=2^n$ and in the last line we use the fact that $\tr [V_0 \id_d V^{\dagger}_0 Z_1]=0$. In the last line we omit the identity operators on the trailing $n-1$ qubits for simplicity.

\begin{figure}[!htb]
    \centering
    \includegraphics[width=0.5\textwidth]{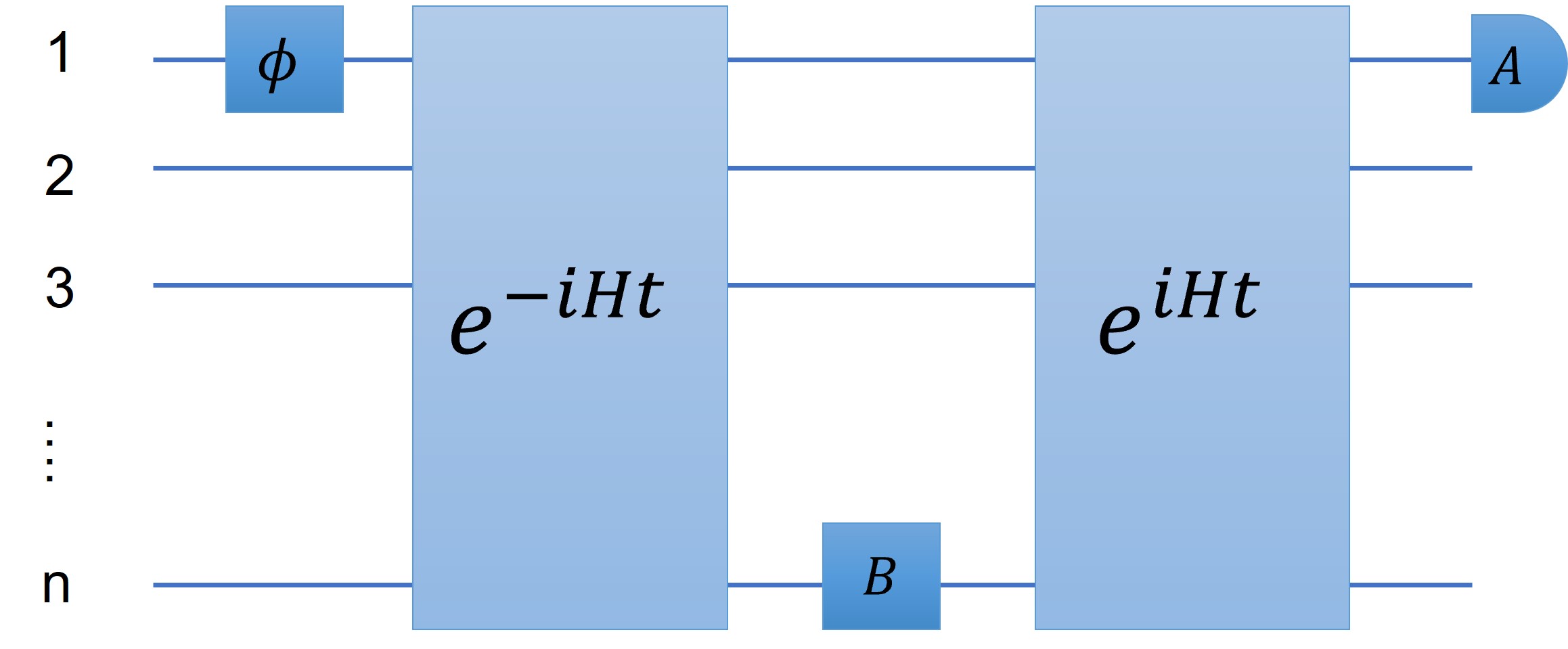}
    \caption{Quantum circuit for measuring an OTOC.}\label{fig:OTOCcircuit}
\end{figure}

The operations $e^{iHt}$ and $e^{-iHt}$ can be approximated using digital quantum simulation methods such as product formulas. For instance, in Ref.~\cite{li2017measuring}, the authors experimentally realized PF2.
Compared with the ideal evolution $V_0$, the approximated evolution $V=V_0(\id+\mathscr{M})$ has multiplicative error $\mathscr{M}$.
Since the input state $\rho_{\mathrm{in}}$ can be viewed as a mixture of
states with the first qubit in the state $\phi$ and the other $n-1$ qubits randomly chosen from a basis, i.e.,
\begin{equation}
\rho_{\mathrm{in}}= \frac{1}{d_1}\sum_{j=0}^{d_1-1}\phi \otimes \ket{j}\bra{j},
\end{equation}
the algorithmic error from a product formula approximation can be bounded using our techniques as
\begin{equation}
\begin{aligned}
 |\langle O(t) \rangle- \langle \widetilde{O(t)} \rangle|   & =|\tr (V_0 \rho_{\mathrm{in}} V^{\dagger}_0 Z_1)-\tr (V \rho_{\mathrm{in}} V^{\dagger} Z_1)|\\
 &\le \| V_0 \rho_{\mathrm{in}} V^{\dagger}_0- V \rho_{\mathrm{in}} V^{\dagger} \|_1 \\
 &=\Big\| V_0 (\phi\otimes \frac{\id_{d_1}}{d_1}) V^{\dagger}_0- V (\phi\otimes \frac{\id_{d_1}}{d_1}) V^{\dagger} \Big\|_1\\
 &\le \frac{1}{d_1}\sum_{j=0}^{d_1-1}\Big\| V_0 (\phi\otimes \ket{j}\bra{j} ) V^{\dagger}_0- V  (\phi\otimes \ket{j}\bra{j}  ) V^{\dagger} \Big\|_1\\
 &\le \frac{2}{d_1}\sum_{j=0}^{d_1-1}\| (V_0-V)\ket{\phi}_{1}\ket{j}\|_{\ell_2}\\
 &=2R_{\ell_2}^{\mathrm{sub},1}(V,V_0),
\end{aligned}
\end{equation}
where the second inequality is by the triangle inequality of the trace norm, the third inequality is by the relation between the trace norm and $\ell_2$ norm, and the last inequality is due to Eq.~\eqref{eq:subsysError}.

If $U_0(t)=e^{-iHt}$ and $U_0^{\dag}(t)=e^{iHt}$ are approximated via PF2, then in a small time segment $t/r$, we have
\begin{align}
U_0(t/r)=e^{-iHt/r}\approx \mathscr{U}_2(t/r)=\overrightarrow{\prod_l}e^{-iH_lt/2r} \overleftarrow{\prod_l}e^{-iH_lt/2r}.
\end{align}
Let us denote the multiplicative error by $m$, i.e., $\mathscr{U}_2(t/r)=U_0(t/r)(\id+m)$.
By Lemma \ref{Lemma:subsystem} and the triangle inequality, the error in the long-time evolution can be upper bounded as
$R_{\ell_2}^{\mathrm{sub},1}(V(t),V_0(t))\le \frac{2r}{\sqrt{d_1}} \|m\|_F$.

As a result, the total error is at most $ |\langle O(t) \rangle- \langle \widetilde{O(t)} \rangle|\le \frac{4r}{\sqrt{d_1}} \|m\|_F$.
If $H$ is a nearest-neighbor Hamiltonian acting on $n$ qubits, then according to Appendix \ref{Sec:app_nearest} and Ref.~\cite{childs2020theory}, $\|m\|_F=\mathcal O(\sqrt{n d_1}(t/r)^3)$ and  $\|m\|=\mathcal O(n(t/r)^3)$.
Thus the error can be upper bounded by
\begin{equation}
    |\langle O(t) \rangle- \langle \widetilde{O(t)} \rangle|\le \frac{4r}{\sqrt{d_1}} \|m\|_F=\mathcal O(\sqrt{n} t^3/r^2).
\end{equation}
To ensure Trotter error at most $\varepsilon$, we choose $r=\mathcal O(n^{0.25}t^{1.5}/\varepsilon)$ Trotter steps, giving gate complexity $G=\mathcal O(n^{1.25}t^{1.5}/\varepsilon)$.

On the other hand, if we use the conventional worst-case analysis,
the algorithmic error is
\begin{equation}
\begin{aligned}
    |\langle O(t) \rangle- \langle \widetilde{O(t)} \rangle|& =|\tr (V_0 \rho_{\mathrm{in}} V^{\dagger}_0 Z_1)-\tr (V \rho_{\mathrm{in}} V^{\dagger} Z_1)|\\
    &\le |\tr (V_0 \rho_{\mathrm{in}} V^{\dagger}_0 Z_1)-\tr (V_0 \rho_{\mathrm{in}} V^{\dagger} Z_1)|+|\tr (V_0 \rho_{\mathrm{in}} V^{\dagger} Z_1)-\tr (V_0 \rho_{\mathrm{in}} V_0^{\dagger} Z_1)|\\
    &=|\tr [(V_0-V)\rho_{\mathrm{in}} V^{\dagger}_0 Z_1]|+|\tr [(V-V_0)^{\dagger} Z_1V_0\rho_{\mathrm{in}} ]|.
    \end{aligned}
\end{equation}
Both terms can be bounded by
\begin{equation}
\begin{aligned}
    |\tr [(V_0-V)\rho_{\mathrm{in}} V^{\dagger}_0 Z_1]|&\le \|V_0-V\|\ \|\rho_{\mathrm{in}} V^{\dagger}_0 Z_1\|_1\\
    &=\|\mathscr{M}(t) \| \tr(\sqrt{\rho_{\mathrm{in}}^2})\\
    &=\|\mathscr{M}(t) \|.
    \end{aligned}
\end{equation}
Using the triangle inequality, we have
\begin{equation}
 |\langle O(t) \rangle- \langle \widetilde{O(t)} \rangle|\le 2\|\mathscr{M}(t) \|\le 2r\|\mathscr{M}(t/r) \|\le 4r\|m\|= \mathcal O(nt^3/r^2).
\end{equation}
To ensure error at most $\varepsilon$, we divide the evolution into$r=\mathcal O(n^{0.5}t^{1.5}/\varepsilon)$ segments, giving gate complexity $G=\mathcal O(n^{1.5}t^{1.5}/\varepsilon)$.
Thus we see that our techniques tighten the error and thereby reduce the required number of Trotter steps and gate complexity in the OTOC measurement.

In an ensemble quantum computer such as an NMR experiment, the maximally mixed state can be directly prepared. However, for other quantum platforms,
the initial input state may be sampled as
$\psi_{\mathrm{in}}= \phi_{[1]} \otimes \psi_{[n-1]}$,
where the state $\psi_{[n-1]}\in \mc{E}$ is drawn from an $(n-1)$-qubit 1-design $\mc{E}$ so that $\mathbb{E}_{\psi_{[n-1]}\in \mc{E}}(\psi_{\mathrm{in}})=\phi_{[1]} \otimes{\id_{d_1}}/{d_1}$. Thus, in addition to the algorithmic error discussed above, one should also account for statistical error. Consider the OTOC
\begin{equation}
\hat{O}_{\psi_{\mathrm{in}}}:=\langle O(t)\rangle_{\psi_{\mathrm{in}}}=\tr (V \psi_{\mathrm{in}} V^{\dagger} Z_1)
\end{equation}
as a random variable with probability distribution determined by $\mc{E}$. The expectation value is $\mathbb{E} _{\psi_{\mathrm{in}}}\langle O(t)\rangle_{\psi_{\mathrm{in}}}=\langle\tilde{O}(t)\rangle$. Suppose one samples $K$ random inputs to obtain identically distributed random variables $\{\hat{O}_1,\hat{O}_2,\cdots \hat{O}_K\}$ and uses the sample average $\hat{O}_{\mathrm{{est}}}=\sum_i \hat{O}_i/K$ as an unbiased estimator of $\langle\tilde{O}(t)\rangle$. For any input $\psi_{\mathrm{in}}$, we have $|\langle O(t)\rangle_{\psi_{\mathrm{in}}}|\leq1$, so by Hoeffding's inequality, we have
\begin{equation}
\mathrm{Pr}\left(|\hat{O}_{\mathrm{{est}}}-\langle\tilde{O}(t)\rangle|\geq \epsilon\right)\leq 2\exp\left(-\frac1{2}K\epsilon^2\right).
\end{equation}
Then the total error can be upper bounded by summing the algorithmic and statistical errors, giving $|\hat{O}_{\mathrm{{est}}}-\langle O(t)\rangle|\leq |\hat{O}_{\mathrm{{est}}}-\langle\tilde{O}(t)\rangle|+|\langle O(t)\rangle-\langle\tilde{O}(t)\rangle|$.

\subsection{Trace estimation}\label{App_trace}
In the one clean qubit model \cite{PhysRevLett.81.5672},
one can estimate $\tr(U)$ for an unitary operator $U$ using only a single pure qubit (and many maximally mixed ones). If
$U = e^{-iHt} \approx \mathscr{U}^r_p(t/r)$ is a Hamiltonian evolution that is approximated with product formulas, then our method can be directly used to bound the Trotter error
\begin{equation}
\begin{aligned}
\varepsilon: =|\tr(e^{-iHt})-\tr(\mathscr{U}_p^r(t/r))| \le r|\tr(e^{-iHt/r}-\mathscr{U}_p(t/r))|\le r \sqrt{d}\|\mathscr M(t/r)\|_F.
\end{aligned}
\end{equation}
Here the first inequality is due to the triangle inequality and the second inequality follows from the Cauchy inequality as in Lemma~\ref{Lemma:traceproduct}.
The error can the be bounded for various types of Hamiltonians as discussed in the previous sections. For example, if $H$ is a nearest-neighbor Hamiltonian acting on $n$ qubits, then according to Appendix~\ref{Sec:app_nearest}, we have $\|\mathscr M(t/r)\|_F=\mathcal O(\sqrt{n d}t^{p+1}/r^{p+1})$, so the error is $\varepsilon=\mathcal O(\sqrt{n}2^nt^{p+1}/r^{p})$.
Therefore a gate complexity of $\mathcal O(n^{1+\frac{1}{2p}}t^{1+\frac{1}{p}}2^{\frac{n}{p}})$ suffices to ensure that the total error is at most a small constant.

\section{Numerical results}

In Figures~\ref{Fig:PF12} and~\ref{Fig:PFpower}, we show numerical results illustrating the extent to which our theoretical analysis captures the actual average-case performance of product formulas.
Figure~\ref{Fig:PF12} shows that for the one-dimensional Heisenberg model in Eq.~\eqref{Eq:heisenberg}, the asymptotic scaling of the PF1 interference bound (red curve) derived in Appendix~\ref{Sec:numerical_inter} is close to the empirical performance. For PF2, the asymptotic scaling of the bound from Appendix~\ref{Sec:numerical_near_tri} is also close to the empirical performance. In Figure~\ref{Fig:PFpower}, we consider the one-dimensional Heisenberg model  with  power-law interactions of exponent $\alpha=0$ or $4$, as defined in Eq.~\eqref{Eq:powerlaw}. For $\alpha=0$, the asymptotic scaling using the bound of Appendix~\ref{Sec:numerical_power_theo} agrees well with the empirical results; for $\alpha=4$,
the asymptotic scaling using the bound of Appendix~\ref{Sec:numerical_power_theo} for PF2 also agrees well with the empirical results, but for PF1 there is a substantial gap between the theoretical and empirical curves. This might be due to destructive error interference between Trotter steps, but we leave a detailed investigation as a problem for future research.

Note that for all our average-case empirical results except those described in part b of Appendix~\ref{sec:numerical_empirical}, we
generate 20 random inputs according to Haar measure. All the average error results were generated using the Julia programming language, except for the calculations described in Appendix~\ref{sec:cauchycompare}, which were generated using Mathematica. The worst-case empirical results were generated using Matlab. See \url{https://github.com/zhaoqthu/Hamiltonian-simulation-with-random-inputs} for the source code used to produce these numerical results.
All extrapolations were obtained using the polynomial curve fitting function (polyfit) in Matlab, neglecting data with small $n$ that showed clear deviation from the apparent asymptotic trend.

In this section, we first present the detailed calculation of various theoretical bounds for specific models: nearest-neighbor Hamiltonians in Appendix~\ref{Sec:Numerical_nearest} and the one-dimensional Heisenberg model with  power-law interactions in Appendix~\ref{Sec:numerical_power}.
In Appendix~\ref{sec:numerical_empirical}, we present further empirical results describing the scaling of the interference bound with $t$ (Figure~\ref{fig:trotter_heisenberg}), the performance of higher-order formulas, the effect of choosing other 1-design input distributions,
and statistical fluctuations in the error with Haar-random inputs.

\subsection{Nearest-neighbor Hamiltonians}\label{Sec:Numerical_nearest}
Consider the one-dimensional Heisenberg model Hamiltonian with a random magnetic field $h_{j}\in [-1,1]$ at each site $j \in \{1,\ldots,n\}$,
\begin{equation}\label{Eq:heisenberg}
    H = \sum^{n-1}_{j=1} \left(X_j X_{j+1} + Y_jY_{j+1} + Z_j Z_{j+1}\right)+ \sum^{n}_{j=1}h_{j}Z_{j}.
\end{equation}
The summands of this Hamiltonian can be partitioned into two groups in an even-odd pattern \cite{Childs2019Product}, giving $H=A+B$ with
\begin{equation}
    \begin{aligned}
    A &= \sum^{\lfloor\frac{n}{2}\rfloor}_{j=1} \left(X_{2j-1} X_{2j} + Y_{2j-1} Y_{2j} + Z_{2j-1}  Z_{2j}\right)+ \sum^{\lfloor\frac{n+1}{2}\rfloor}_{j=1}  h_{2j-1}Z_{2j-1} ,\\
    B &= \sum^{\lceil\frac{n}{2}\rceil-1}_{j=1} \left(X_{2j} X_{2j+1} + Y_{2j} Y_{2j+1} + Z_{2j} Z_{2j+1}  \right)+ \sum^{\lceil\frac{n-1}{2}\rceil}_{j=1} h_{2j}Z_{2j}.
\end{aligned}
\end{equation}

\subsubsection{Interference bound}
\label{Sec:numerical_inter}
Recall that the interference bound in Theorem~\ref{Th:interference} shows that the error is at most
\begin{equation}
 R_{\ell_2}(\mathscr{U}^r_1(t/r),U_0(t))=O\left( \sqrt{n}\left(\frac{t}{r}+\frac{t^3}{r^2}\right)\right).
\end{equation}
We investigate the empirical error of the interference bound by varying the evolution time $t$ with fixed $r$ and $n$, as shown in Fig.~\ref{fig:trotter_heisenberg}.
We observe that the error scales linearly in $t$ for small $t$ and cubically for larger $t$, with a transition around $t\approx \sqrt{r}$, in agreement with Theorem~\ref{Th:interference}.

\begin{figure}[!htb]
    \centering
    \includegraphics[width=0.5\textwidth]{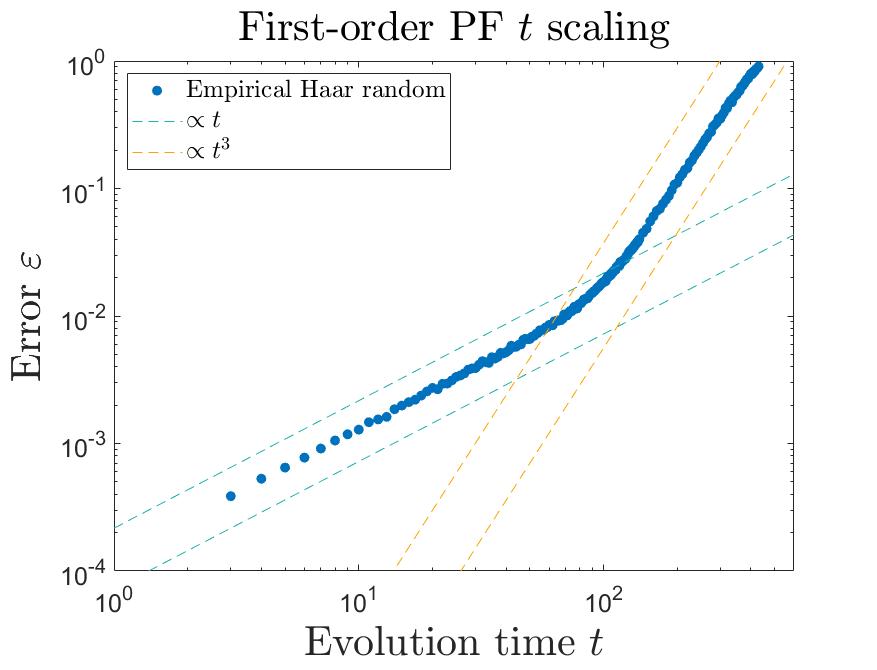}
    \caption{Algorithmic error of the first-order product formula for the one-dimensional Heisenberg model in Eq.~\eqref{Eq:heisenberg}
    with different evolution time $t$. Here we choose $n=6$ and $r=10000$. }
   \label{fig:trotter_heisenberg}
\end{figure}

To obtain a concrete prefactor for the interference bound, we calculate it in more detail as follows.
As shown in the proof of Theorem~\ref{Th:interference},  we have the following inequalities:
\begin{equation}\label{ap:eq:PF1infer}
\begin{aligned}
R_{\ell_2}((\mathscr{U}^r_1(t/r),U_0(t))&\le \frac{1}{\sqrt{d}}\|\mathscr{M}_r(t)   \|_F=  \frac{1}{\sqrt{d}}\sqrt{ \tr(\mathscr{M}_r\mathscr{M}_r^{\dagger}) };\\
\tr(\mathscr{M}_r\mathscr{M}_r^{\dagger})   &\le \sum_{k,k'=1}^{r} |\tr(\Delta_k\Delta_{k'}^{\dagger}) |\le   \left(\sum_{k,k'=1}^{r} r^{k+k'-2} \|\mathscr{M}_1\|^{k+k'-2}  \right)\tr(\Delta_1\Delta_{1}^{\dagger})\\
 &\le  \tr(\Delta_1\Delta_{1}^{\dagger})\left(\frac{1}{1-r\|\mathscr{M}_1\|}\right)^2;\\
 \tr(\Delta_1\Delta_{1}^{\dagger})&\le \left[\frac{2r}{t}
\left(\sum_{k=0}^{\infty} \sqrt{\tr(|F^{\circ k}(S)|^2)}\right)+ r \sqrt{\tr(VV^{\dagger})}\right]^2.
\end{aligned}
\end{equation}
According to the proofs of Lemmas~\ref{Lemma:SV} and \ref{Lemma:Fk}, we can bound the two terms in $\tr(\Delta_1\Delta_{1}^{\dagger})$ as
\begin{equation}
\begin{aligned}
\sum_{k=0}^{\infty} \sqrt{\tr(|F^{\circ k}(S)|^2)} &\le    \sqrt{\tr(SS^{\dagger})}\left(\frac{1}{1-t\|H\|/r}\right)\le \sqrt{dn}\frac{t^2}{r^2} \left(\frac{1}{1-t\|H\|/r}\right)^2,\\
\tr(VV^{\dagger})&\le \frac{d}{36n} \left(C^2 e^2 \|H\|^2 \frac{t^6}{r^6}\right) \left[ \left(\frac{1}{1-e\|H\|t/r}\right)^2-1\right]+ \tr(|[A,[H,B]]|^2)\frac{t^6}{36r^6}\\
&\approx \mathcal O\left(dn^2\frac{t^7}{r^7}+ dn\frac{t^6}{r^6}\right),
\end{aligned}
\end{equation}
where $C=\frac{2048}{e^2(e-1)}$ and  $\|H\|\le 4n$.
In the bound for $\tr(VV^{\dagger})$, we directly calculate $V_3V_{3}^{\dagger}$ with $V_3=[A,[H,B]]$, and apply Eq.~\eqref{Eq:vkvk'} for the other terms $V_kV_{k'}^{\dagger}$.
Combining these bounds on the individual terms, we obtain the interference bound. Note that the term $\tr(|[A,[H,B]]|^2)$ is evaluated numerically.

We use the above concrete bound to draw the theoretical interference bound curve in Figure~\ref{Fig:PF12}. Compared to the empirical curve, the theoretical interference bound performs well for small system sizes, although we expect it to deviate for larger $n$. This is because to keep the error below a threshold $\varepsilon$, the quantity $r\|\mathscr{M}_1\|\le \frac{t^2}{2r}\|[A,B]\|= \mathcal O(\frac{nt^2}{r})$ in the second line of Eq.~\eqref{ap:eq:PF1infer} must be sufficiently small, which is the assumption in Theorem~\ref{Th:interference}.
Consequently, we choose the number of segments to be
\begin{equation}
r=\max\{\mathcal O(\|[A,B]\|t^2/2),   \mathcal O(n^{0.5}t/\varepsilon), \mathcal O(n^{0.25}t^{1.5}/\varepsilon^{0.5})    \}.
\end{equation}
For large $t$ and $n$, the first term dominates, giving $r= \mathcal O(\|[A,B]\|t^2)$. However, in our classically tractable regime $n\le 20$, the second term dominates, giving $r=\mathcal O(n^{0.5}t/\varepsilon)$.

\subsubsection{Triangle bound}
\label{Sec:numerical_near_tri}

For the triangle inequality with $H=A+B$, we directly use
Eqs.~\eqref{Eq;PF1AB} and \eqref{Eq:PF2AB}, evaluating $\tr(|[B,[B,A]]|^2)$, $\tr(|[A,[B,A]]|^2)$, and $\tr(|[B,A]|^2)$ numerically.

\subsubsection{Counting bound}\label{app_count_near}

For the counting bound, we also use Eqs.~\eqref{Eq;PF1AB} and \eqref{Eq:PF2AB} but do not evaluate the trace numerically.
We expand $[A,B]$ as the sum of 8 different types of Pauli strings, namely $X_jY_{j+1}Z_{j+2}$, $Z_jY_{j+1}X_{j+2}$, $Y_jX_{j+1}Z_{j+2}$,   $Z_jX_{j+1}Y_{j+2}$, $X_jZ_{j+1}Y_{j+2}$, $Y_jZ_{j+1}X_{j+2}$, $X_jY_{j+1}$, and $Y_jX_{j+1}$, with the coefficients $2i$ or $-2i$.
To make the trace nonzero, the Pauli operator in $[A, B]$ should be the same as that of $[A,B]^{\dagger}$. Thus each term only has one choice. Then $\tr([A,B]^2)\le d\cdot 8\cdot 2^2\cdot n$. It follows that $T_1'\le 4\sqrt{2}n$ and $R_{\ell_2}\le 2\sqrt{2}nt^2/r$. The results for $T'_2$ are obtained similarly.

\subsection{Power-law interactions}\label{Sec:numerical_power}

Consider the one-dimensional Heisenberg model with power-law interactions of exponent $\alpha>0$ and a random magnetic field $h_j\in [-1,1]$, and the Hamiltonian shows
\begin{equation}\label{Eq:powerlaw}
    H = \sum^{n-1}_{j=1} \sum^n_{k = j+1} \frac{1}{|j - k|^\alpha} \left(X_j X_{k} + Y_j Y_{k} + Z_j Z_{k}\right)+\sum^{n}_{j=1}  h_j Z_j.
\end{equation}
In Appendix~\ref{Sec:App_powerlaw} we have shown the asymptotic scaling of the simulation error for PF1 and PF2 methods.
To compare with the empirical performance, we apply Theorem~\ref{Th:PF1} with the X-Y-Z ordering \cite{childs2018toward} of the Hamiltonian and calculate the prefactors. Specifically, we write the Hamiltonian into
$H=H_X+H_Y+H_Z$ with $H_X$ only containing the terms with Pauli $X$ operators, and similarly for $H_Y$ and $H_Z$. We further denote $H_Z=H_{Z1}+H_{Z2}$, where $H_{Z1}$ is the power-law term and $H_{Z2}$ is the magnetic field term.

\subsubsection{Triangle bound}
\label{Sec:numerical_power_theo}

According to Theorem~\ref{Th:PF1}, the average error is
\begin{equation}\label{Eq:t1'}
R_{\ell_2}(\mathscr{U}^r_1(t/r),U_0(t))\le \frac{t^2}{2r}T_1',
\end{equation}
where
\begin{equation}\label{ap:eq:Pxyz}
\begin{aligned}
T_1'
=\left(\frac{\tr\big(|[H_X,H_Y+H_Z]|^2\big)}{d}\right)^{\frac{1}{2}}+\left(\frac{\tr\big(|[H_Y,H_Z]|^2\big)}{d}\right)^{\frac{1}{2}}.
\end{aligned}
\end{equation}
One can numerically evaluate the above commutators and thus get a concrete bound. The bound for PF2 (Theorem~\ref{Th:PF2}) can be obtained similarly.

\subsubsection{Counting bound}\label{app_count_power}

To compute the theoretical bound in Eq.~\eqref{Eq:t1'}, we still need to calculate the trace of the commutator numerically, which is classically intractable for large systems. To explore the performance of these bounds for large $n$, here we analytically bound $T_1'$, giving an error bound that we call the counting bound.
We first expand the second term as
\begin{equation}\label{Eq:powerYZ}
\begin{aligned}
\frac1{d}\tr\bigl([H_Y,H_Z][H_Y,H_Z]^{\dag}\bigr)&=\frac1{d}\tr\bigl([H_Y,H_{Z1}+H_{Z2}][H_Y,H_{Z1}+H_{Z2}]^{\dag}\bigr)\\
&=\frac1{d} \tr\bigl(|[H_Y,H_{Z1}]|^2\bigr) + \frac1{d} \tr\bigl(|[H_Y,H_{Z2}]|^2\bigr) -\frac2{d}\tr\bigl([H_Y,H_{Z1}][H_Y,H_{Z2}]\bigr).
\end{aligned}
\end{equation}
For the first term in Eq.~\eqref{Eq:powerYZ}, $[H_Y,H_{Z1}]$ can be expressed as a sum of terms of the form $[Y_jY_k, Z_{j'}Z_{k}]=2iY_jZ_{j'}X_k$.
To make the trace nonzero, the Pauli operator in $[H_Y,H_{Z1}]$ should be the same as that of $[H_Y,H_{Z1}]^{\dag}$. Therefore
\begin{equation}\label{Eq:powerHYZ}
\begin{aligned}
\frac1{d}\tr\bigl([H_Y,H_{Z1}][H_Y,H_{Z1}]^{\dag}\bigr)&=\sum_{j<k,~j'< k', \mathrm{one~eq}}  \left[2\frac{1}{|j-k|^{\alpha}}\frac{1}{|j'-k'|^{\alpha}}\right]^2= \sum_{j,j',k}  4\frac{1}{|j-k|^{2\alpha}}\frac{1}{|j'-k|^{2\alpha}}.
\end{aligned}
\end{equation}
Here ``$\mathrm{one~eq}$'' means that
the pair $(j,k)$ and $(j',k')$ have one coincidence. Similarly, for the second term in Eq.~\eqref{Eq:powerYZ}, $[H_Y,H_{Z2}]$ can be expressed as a sum of terms of the form $[Y_jY_k, h_kZ_{k}]=2ih_kY_jX_k$, so
\begin{equation}\label{Eq:powerHYZ2}
\begin{aligned}
\frac1{d}\tr\bigl([H_Y,H_{Z2}][H_Y,H_{Z2}]^{\dag}\bigr)&=\sum_{j<k}  \left(2\frac{1}{|j-k|^{\alpha}}\right)^2 (h_j^2+h_k^2)\le \sum_{j<k} 8\frac{1}{|j-k|^{2\alpha}},
\end{aligned}
\end{equation}
using $h_j^2\leq 1$.
For the third term, we have $\tr\bigl([H_Y,H_{Z1}][H_Y,H_{Z2}]\bigr)=0$ since the Pauli operators in the first and second commutator have weight 3 and 2, respectively.

We can then consider the first term in Eq.~\eqref{ap:eq:Pxyz} and expand it in the same way as in Eq.~\eqref{Eq:powerYZ}, giving
\begin{equation}
\begin{aligned}
\frac1{d}\tr\big(|[H_X,H_Y+H_Z]|^2\big)=\frac1{d} \tr\bigl(|[H_X,H_Y]|^2\bigr) + \frac1{d} \tr\bigl(|[H_X,H_Z]|^2\bigr) -\frac2{d}\tr\bigl([H_X,H_Y][H_X,H_Z]\bigr).
\end{aligned}
\end{equation}
The first two terms are similar to the previous cases.
In particular, $\tr\big(|[H_X,H_Y]|^2\big)=\tr\big(|[H_Y,H_{Z1}]|^2\big)<\tr\big(|[H_Y,H_Z]|^2\big)$ and $\tr\big(|[H_X,H_Z]|^2\big)=\tr\big(|[H_Y,H_Z]|^2\big)$.
We claim that the last term is negative:
\begin{equation}
\begin{aligned}
 -\frac2{d}\tr\bigl([H_X,H_Y][H_X,H_Z]\bigr)&=-\frac2{d}\tr\bigl([H_X,H_Y][H_X,H_{Z1}]\bigr)<0,
 \end{aligned}
\end{equation}
To see this, observe that
the term in the first commutator $[H_X,H_Y]$ has the form $2iX_jZ_kY_l$, and the term in the second commutator $[H_X,H_{Z1}]$ has the form $-2iX_{j'}Z_{k'}Y_{l'}$. To get a nonzero trace, these expressions must coincide, which leads to an overall negative sign.
Actually it shows exactly
\begin{equation}
\begin{aligned}
 -\frac2{d}\tr\bigl([H_X,H_Y][H_X,H_Z]\bigr)&=-\frac2{d}\tr\bigl([H_X,H_Y][H_X,H_{Z1}]\bigr)\\
 &=-\sum_{j,k,l} 8\frac{1}{|j-k|^{\alpha}}\frac{1}{|k-l|^{\alpha}}\frac{1}{|j-l|^{\alpha}}\frac{1}{|l-k|^{\alpha}}.
 \end{aligned}
\end{equation}
As a result, we can upper bound $T_1'$ in Eq.~\eqref{ap:eq:Pxyz} as
\begin{equation}
\begin{aligned}
T_1'&= \left(\frac{\tr\bigl(|[H_X,H_Y]|^2\bigr)+ \tr\bigl(|[H_X,H_Z]|^2\bigr)}{d}\right)^{\frac{1}{2}} +\left(\frac{\tr\bigl(|[H_Y,H_Z]|^2\bigr)}{d}\right)^{\frac{1}{2}}\\
&< (\sqrt{2}+1)\left(\frac{\tr\bigl(|[H_Y,H_Z]|^2\bigr)}{d}\right)^{\frac{1}{2}}\\
&< 2(\sqrt{2}+1)\left[\sum_{j,j',k}  \frac{1}{|j-k|^{2\alpha}}\frac{1}{|j'-k|^{2\alpha}}+\sum_{j<k}\frac{2}{|j-k|^{2\alpha}}\right]^{\frac{1}{2}},
\end{aligned}
\end{equation}
where we insert the results in Eq.~\eqref{Eq:powerHYZ} and \eqref{Eq:powerHYZ2}.


\subsection{Other empirical results}\label{sec:numerical_empirical}

We conclude in this section by presenting some numerical results that shed light on other aspects of the average-case performance of product formulas.

\paragraph{Higher-order formulas}
In addition to the results for PF1 and PF2 described above, we also test the empirical performance of PF4 and PF6. Because we do not have concrete prefactors for our theoretical results in these cases, we only show the empirical data and the corresponding extrapolation curves. See Fig.~\ref{fig:PF46heisenberg} for the nearest-neighbor case in Eq.~\eqref{Eq:heisenberg} and Fig.~\ref{fig:PF46power} for the one dimensional Heisenberg model with power-law interactions $\alpha=0$ and $4$ in Eq.~\eqref{eq:power}.
For comparison, we list the empirical scaling and theoretical scaling in Table~\ref{tab:pf46}.

\begin{figure}[!htb]
    \centering
    \includegraphics[width=0.6\textwidth]{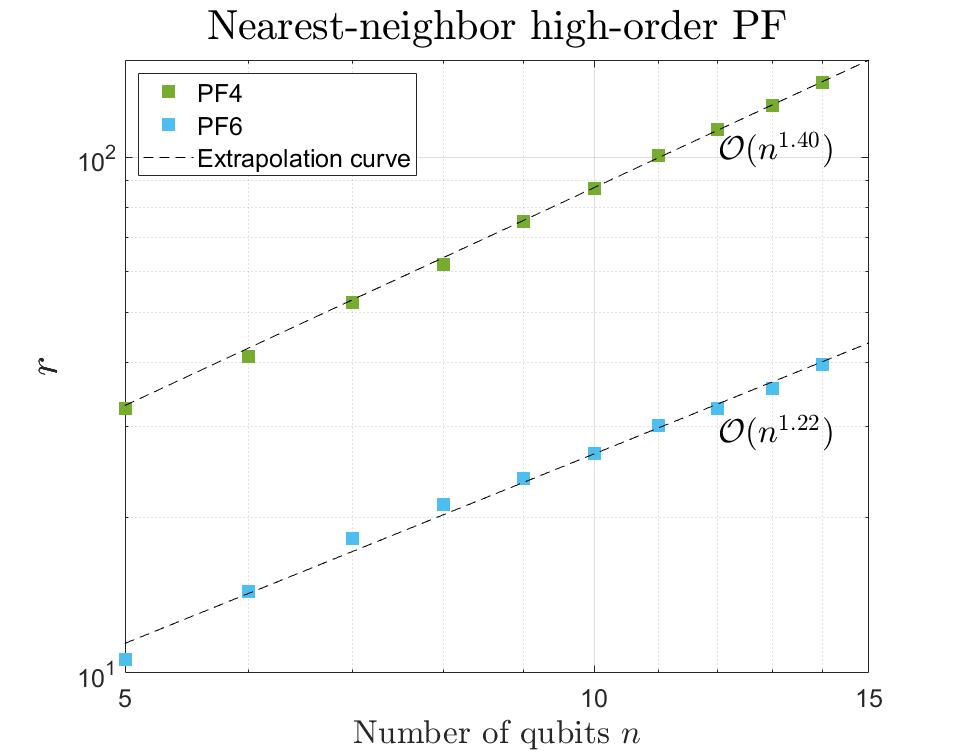}
    \caption{PF4 and PF6 for nearest-neighbor Hamiltonian}
   \label{fig:PF46heisenberg}
\end{figure}

\begin{figure}[!htb]
    \centering
    \includegraphics[width=1\textwidth]{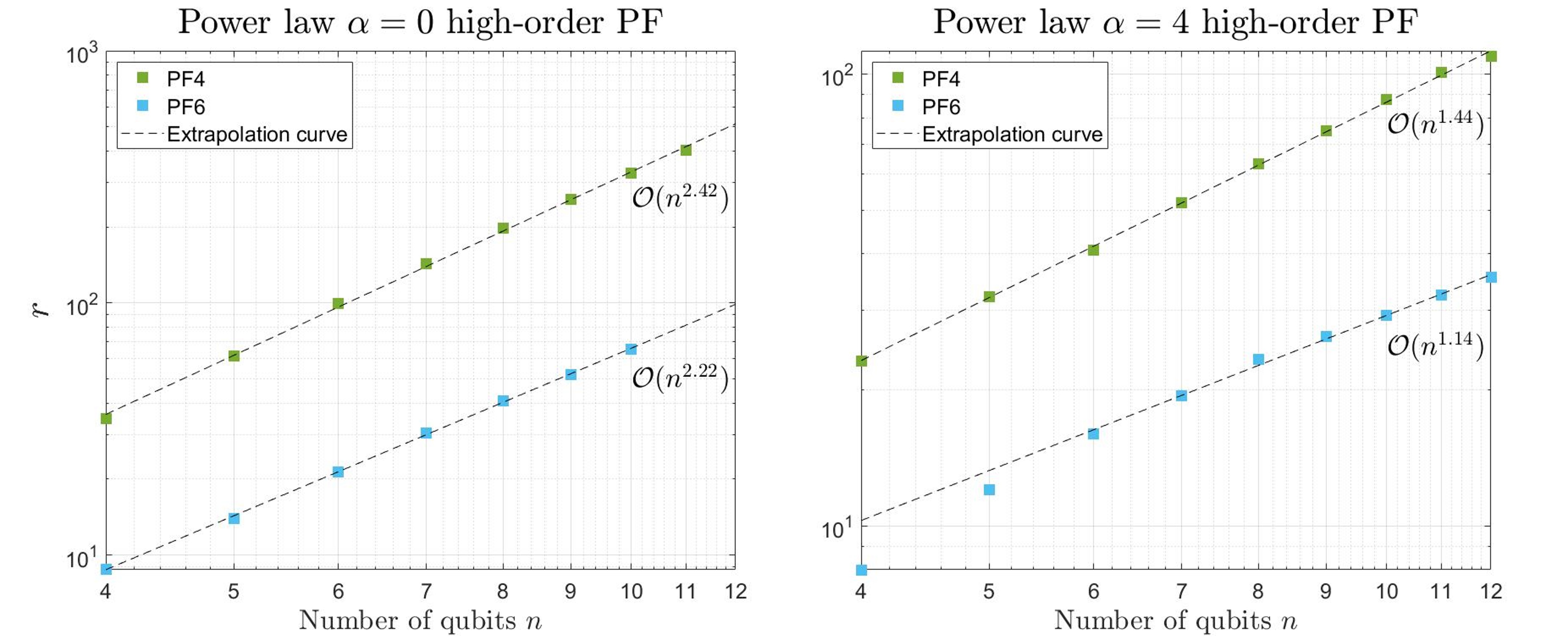}
    \caption{PF4 and PF6 for Hamiltonians with power-law interactions $\alpha=0,4$.}
   \label{fig:PF46power}
\end{figure}

\begin{table}[!htb]
\begin{tabular}{|c|c|c|c|c|c|c|}
\hline
\multirow{2}{*}{Order} & \multicolumn{2}{c|}{Nearest-neighbor}      & \multicolumn{2}{c|}{Power-law  $\alpha=4$}      & \multicolumn{2}{c|}{Power-law  $\alpha=0$}      \\ \cline{2-7}
                       & Empirical & Theoretical & Empirical & Theoretical & Empirical & Theoretical \\ \hline
$p=4$                    &    $\mathcal O(n^{1.40})$        &   $\mathcal O(n^{1.375})$            &  $\mathcal O(n^{1.44})$          &    $\mathcal O(n^{1.50})$             &     $\mathcal O(n^{2.42})$       &        $\mathcal O(n^{2.75})$         \\ \hline
$p=6$                    &     $\mathcal O(n^{1.22})$       &   $\mathcal O(n^{1.25})$            &  $\mathcal O(n^{1.14})$          &    $\mathcal O(n^{1.17})$            &      $\mathcal O(n^{2.22})$       &     $\mathcal O(n^{2.5})$            \\ \hline
\end{tabular}
\caption{Empirical asymptotic scaling and theoretical scaling of $r$ for  nearest-neighbor, and power-law interaction $\alpha=0,4$ Hamiltonians with PF4 and PF6.}\label{tab:pf46}
\end{table}

\paragraph{Other 1-design inputs}
Except for this paragraph, all the numerical examples are obtained according to Haar-random inputs for simplicity. However, our results can also be applied to other 1-design input ensembles. Here we numerically test the performance of two other ensembles: (1) local Haar-random inputs, i.e., $\bigotimes_{i=1}^n u_i\ket{0}^{\otimes n}$ with each $u_i$ being a single-qubit Haar-random unitary, and (2) a uniformly random computational basis state, i.e., a locally random state where each $u_i$ is chosen independently from $\{\id,X\}$, each with probability $1/2$. We compare the performance of PF1 for these two 1-design inputs with Haar-random inputs in Figure~\ref{fig:otherinputs}, and conclude that they perform similarly.

\begin{figure}[!htb]
    \centering
    \includegraphics[width=0.6\textwidth]{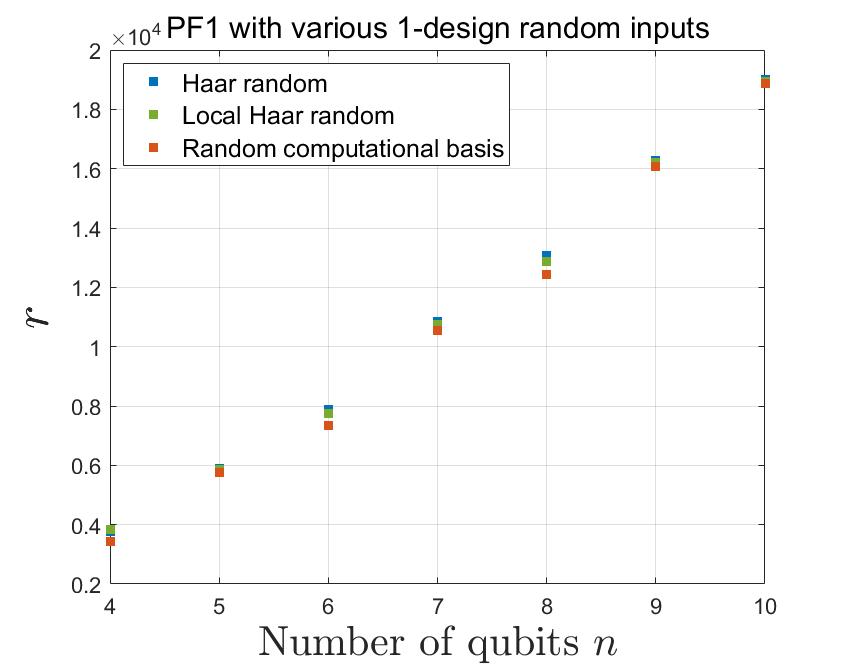}
    \caption{Comparison of three 1-design input ensembles: Haar random, locally random, and random computational basis state.}
    \label{fig:otherinputs}
\end{figure}

\paragraph{Statistical fluctuations}
We also test the statistical fluctuations of Haar-random inputs.
For a given Hamiltonian $H$, in all the numerical tests, we generate 20 Haar-random inputs and obtain an average Trotter number $r(t,\varepsilon,H)=\min\{ r: \mathbb{E}_{\psi}\|\mathscr{U}^r_p(t/r)\ket{\psi} - e^{-itH}\ket{\psi}\|_{2}\le \varepsilon \}$ such that the average error is below $\varepsilon=10^{-3}$.
We define the variable $\sqrt{S}(\psi)=\|\mathscr{U}^r_p(t/r)\ket{\psi} - e^{-itH}\ket{\psi}\|_{2}$. When $r=r(t,\varepsilon,H)$, the mean value of this variable is $\mathbb{E}_{\psi}(\sqrt{S}(\psi))=R_{\ell_2}(e^{-itH},\mathscr{U}^r_p(t/r))=\varepsilon=10^{-3}$. We numerically compute the standard deviation of $\sqrt{S}(\psi)$ $\mathcal{SD}(\sqrt{S}(\psi))$ for 20 Haar-random inputs. We test the one-dimensional Heisenberg model with PF1 and PF2 and the power-law interaction model with PF1, as shown in Fig.~\ref{fig:SD}. We observe that when $\mathbb{E}_{\psi}(\sqrt{S}(\psi))=10^{-3}$, the standard deviation decreases with the dimension of the system $d=2^n$ and its extrapolation curves match $\mathcal{SD}(\sqrt{S}(\psi))=\mathcal O(d^{-0.5})$.

In our theoretical analysis, we have not obtained the mean value or standard deviation of $\sqrt{S}(\psi)$. Instead, we estimate properties of its square $S(\psi)=2- \braket{\psi|(\mathscr{U}^r_p(t/r))^{\dagger}e^{-iHt}+e^{iHt}\mathscr{U}^r_p(t/r)|\psi}$, namely the mean value and the standard deviation of $S(\psi)$.
From Lemma~\ref{Lemma:S}, when $\mathbb{E}_{\psi}(S(\psi))$ reaches a constant, its standard deviation has the scaling
$\mathcal{SD}(S(\psi))=\mathcal O(d^{-0.5})$. This suggests that the relationship between the mean value and the standard deviation of $\sqrt{S}$ is similar to that of $S$.
\begin{figure}[!htb]
    \centering
    \includegraphics[width=1\textwidth]{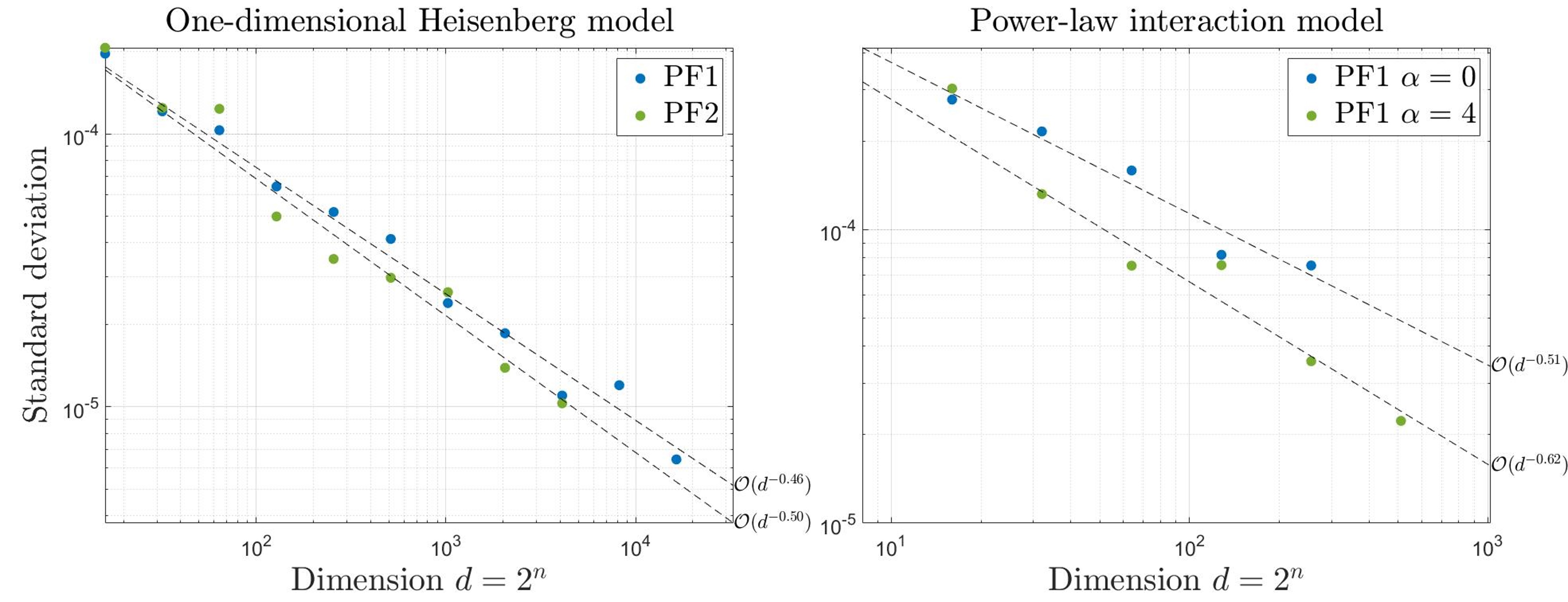}
    \caption{Standard deviation of algorithmic error for Haar-random inputs.}
   \label{fig:SD}
\end{figure}

\end{appendix}
\end{document}